\documentclass[11pt]{article}

\usepackage[letterpaper,margin=1in]{geometry}
\usepackage[T1]{fontenc}
\usepackage{lmodern}
\usepackage{microtype}
\usepackage{amsmath,amssymb,amsthm,mathtools}
\usepackage{graphicx}
\usepackage{booktabs}
\usepackage[colorlinks=true,linkcolor=blue,citecolor=blue,urlcolor=blue]{hyperref}

\newcommand{\abs}[1]{\lvert#1\rvert}

\newcommand{\bbe}[1]{\mathbb{E}\left[#1\right]}
\newcommand{\bbp}[1]{\mathbb{P}\left(#1\right)}

\newcommand{\Var}{\mathrm{Var}}

\DeclareMathOperator{\rk}{rank}
\DeclareMathOperator{\diag}{diag}
\DeclareMathOperator{\hol}{hol}
\DeclareMathOperator{\range}{range}

\newcommand{\mcA}{\mathcal{A}}
\newcommand{\mcB}{\mathcal{B}}

\newcommand{\mcE}{\mathcal{E}}
\newcommand{\mcF}{\mathcal{F}}

\newcommand{\mcN}{\mathcal{N}}
\newcommand{\mcO}{\mathcal{O}}

\usepackage{bm}
\usepackage{xcolor}

\usepackage[numbers,sort&compress]{natbib}

\usepackage[capitalize,noabbrev]{cleveref}

\newtheorem{theorem}{Theorem}[section]
\newtheorem{proposition}[theorem]{Proposition}
\newtheorem{corollary}[theorem]{Corollary}
\newtheorem{lemma}[theorem]{Lemma}
\theoremstyle{definition}
\newtheorem{definition}[theorem]{Definition}
\newtheorem{example}[theorem]{Example}
\theoremstyle{remark}
\newtheorem{remark}[theorem]{Remark}

\title{Community Detection from Aggregated Relational Data:\\
Identifiability and Spectral Recovery}
\author{Andrew Davison\\
Independent Researcher\\
\texttt{adavison1994@gmail.com}
\and
Owen G. Ward\\
Department of Statistics and Actuarial Science,\\ Simon Fraser University\\
\texttt{owen\_ward@sfu.ca}}
\date{\today}

\begin{document}
\maketitle
\begin{abstract}
  Aggregated relational data replaces a fully observed network with a compressed
  version, where survey respondents are asked questions such as ``How many of your
  friends belong to group $G$?''. We consequently only observe a matrix $R = OA$ of the
  adjacency matrix $A$ with a fixed design $O$. In this paper we understand the
  effects of this compression on community detection under degree-corrected
  stochastic block models. First, we give guarantees for when it is possible to
  recover the full network parameters, and then only just the communities under
  conditions on a diagonal corrected version of $\mathbb{E}[R]$, which can be
  used to inform experimental design. Secondly, we show that spectral
  clustering with single linkage exactly recovers communities
  when $n\rho_n\gtrsim K^2+K\log n$ when there are $K$ balanced communities
  under ``uniformly spread'' or near-orthogonal designs, and characterize more generally how the
  design influences rates of convergence. We also introduce second-order
  ``friends-of-friends'' queries and develop the corresponding theory.
\end{abstract}

\section{Introduction}
\label{sec:intro}

Collecting full network data is expensive and often infeasible: respondents
rarely enumerate all of their contacts, and study designs frequently prohibit full
enumeration. Survey practice therefore relies on \emph{aggregated relational data}
(ARD)---answers to questions of the form ``How many of your contacts belong to group
$G$?'', asked across a range of groups $G$. Formally, in place of the adjacency matrix
$A \in \{0,1\}^{n \times n}$ of the network, we observe
\[
    R = OA \in \mathbb{R}^{L \times n},
\]
where the $L$ rows of the fixed design matrix $O$ encode the survey questions.

This clearly represents a potential loss of information; the contributions of
our paper relate to examining the extent of that loss, particularly focusing on community
detection tasks. Our main contributions are as follows.

\begin{enumerate}
    \item We characterize, in terms of the design $O$, when the design separates the
    communities in the block signal $OP^0$, in addition to when this signal determines
    the full network parameters. In particular, recovering all network parameters
    requires collecting at least as many queries as there are communities
    (\Cref{thm:full-identifiability}), whereas a single query can sometimes suffice to
    separate the communities alone (\Cref{cor:partition}).
    \item We show that spectral clustering on the matrix $R$ recovers the communities
    under (degree-corrected) stochastic block models, and identify the effect of the
    design $O$ on the sparsity regime for which exact recovery is available
    (\Cref{sec:recovery}). For balanced communities and uniformly controlled
    design and block conditioning, designs satisfying the quantitative
    spread condition of \Cref{cor:sbm-regimes} give exact recovery by
    spectral embedding and single linkage at $n\rho_n\gtrsim K^2+K\log n$. We
    more generally characterize weak and exact recovery bounds in terms of
    the properties of $O$. 
    \item We develop the theory of second-order ``friend-of-friend'' queries---which we call
    \emph{fan-out} queries---and show how such
    queries benefit community separation and permit weak and exact recovery at thresholds
    matching those of first-order queries under the stated
    effective-design conditioning and spread or near-orthogonality assumptions
    (\Cref{prop:stacked-rank,thm:fanout-main,cor:fanout-regimes}).
\end{enumerate}

\section{Notation and assumptions}
\label{sec:prelim}

\subsection{Notation}
\label{subsec:notation}

We write $[n] = \{1, \dots, n\}$. For a matrix $A \in \mathbb{R}^{n \times m}$ we
write $A_{i \cdot}$ for its $i$-th row, $A_{\cdot j}$ for its $j$-th column, and
$\sigma_1(A) \geq \sigma_2(A) \geq \cdots$ for its
singular values. We write $e_{n,i}$ (or $e_i$ when the dimension is clear) for
the $i$-th standard basis vector of $\mathbb{R}^n$, $\mathbf{1}_n$ for the all-ones
vector, and $\mathbf{1}_S \in \{0,1\}^n$ for the indicator vector of a set
$S \subseteq [n]$. For $\theta \in \mathbb{R}^n$, $\diag(\theta)$ is the diagonal
matrix with entries $\theta_i$; conversely $\diag(A)$ zeroes the off-diagonal of a
square $A$. The operator and entrywise matrix norms
$\| \cdot \|_{p \to q}$ and $\| \cdot \|_{p, q}$ are defined in \Cref{sec:linalg}.

\subsection{Stochastic block models}
\label{subsec:models}

Let $z : [n] \to [K]$ assign each of the $n$ nodes to one of $K \geq 2$ communities,
and let $Z \in \{0,1\}^{n \times K}$ be the corresponding \emph{membership matrix},
$Z_{ia} = \mathbf{1}\{z_i = a\}$. We write $C_a = \{i : z_i = a\}$ for the $a$-th
community, $n_a = \abs{C_a}$ for its size, and
$N := Z^\top Z = \diag(n_1, \dots, n_K)$. Let
$B \in [0,1]^{K \times K}$ be a symmetric \emph{connectivity matrix}.

\begin{definition}[SBM and DCSBM]
    \label{def:sbm}
    In the \emph{stochastic block model} $\mathrm{SBM}(z, B)$, the adjacency matrix
    $A \in \{0,1\}^{n \times n}$ is symmetric and hollow with
    \begin{equation}
        A_{ij} \sim \mathrm{Bernoulli}(B_{z_i z_j}), \qquad i < j,
    \end{equation}
    independent across $i < j$. In the \emph{degree-corrected} stochastic block model
    $\mathrm{DCSBM}(z, B, \theta)$ one additionally fixes node parameters
    $\theta \in \mathbb{R}_{>0}^n$, $\Theta = \diag(\theta)$ so
    \begin{equation}
        A_{ij} \sim \mathrm{Bernoulli}(\theta_i \theta_j B_{z_i z_j}), \qquad i < j.
    \end{equation}
\end{definition}

Throughout the paper the assignment $z$ and the degree parameters $\theta$ are fixed.
For convenience we state our results assuming that the communities are balanced, in
that $n_{\max} \asymp n_{\min} \asymp n/K$. Throughout 
it will be convenient to work with the rank-$\leq K$ ``self-loop-inclusive'' mean matrix
\begin{equation}
    \label{eq:P0}
    P^0 :=
    \begin{cases}
        Z B Z^\top, & \text{SBM}, \\
        \Theta Z B Z^\top \Theta, & \text{DCSBM},
    \end{cases}
\end{equation}
so that the mean of $A$ is the hollow matrix $P := \bbe{A} = P^0 - \diag(P^0)$. We 
assume that $n_a \geq 1$, so $Z$ has full column rank $K$, and
$B$ has pairwise distinct columns (or rows by symmetry). For
the DCSBM we impose for each $a$ the gauge normalization $\sum_{i \in C_a} \theta_i = n_a$.

\subsection{The aggregated observation model}
\label{subsec:obs-model}

Rather than observing $A$ directly, we instead observe
\begin{equation}
    \label{eq:obs}
    \begin{gathered}
    R=OA\in\mathbb R^{L\times n},\\
    \overline R:=\bbe{R}=OP,\qquad M:=OP^0.
    \end{gathered}
\end{equation}
where $O \in \mathbb{R}^{L \times n}$ is a design matrix encoding $L$ queries, which
may itself be unobserved. \emph{Indicator queries} provide one example of such an $O$:
when the rows of $O$ are indicators of groups $G_1, \dots, G_L \subseteq [n]$,
\begin{equation}
    (OA)_{g, j}
    = \sum_{i \in G_g} A_{ij}
    = \#\{i \in G_g : A_{ij}=1\},
\end{equation}
which corresponds to the traditional aggregated relational data design: column $j$ of
$R$ records, for each query group, how many of node $j$'s contacts fall in it. A
simple degree query would take $O = \mathbf{1}_n^\top$.

We also consider \emph{fan-out} queries.
For a square matrix $X$ write $\hol(X) := X - \diag(X)$ for its hollowed version.
The \emph{fan-out matrix} is
\begin{equation}
    \label{eq:fanout-def}
    \begin{aligned}
        A^{\langle 2 \rangle} &:= \hol(A^2) = A^2 - D_A,\\
        D_A &:= \diag(A^2) = \diag(d_1, \dots, d_n).
    \end{aligned}
\end{equation}
where $d_i$ is the degree of node $i$. Again we can ask indicator queries: if row $g$
of $O$ is the indicator of $G_g$, then for every node $j$,
\begin{equation}
    \label{eq:fanout-question}
    \big( O A^{\langle 2 \rangle} \big)_{g, j}
    = \sum_{k \,:\, A_{jk} = 1} \#\{ i \in G_g : A_{ki} = 1,\ i \neq j \},
\end{equation}
which is the answer to the \emph{fan-out question}: ``Summing over each of your
friends, how many of their friends---not counting yourself---belong to group
$G_g$?''

% Excluding the ego is exactly backtracking suppression at length two: the walks
% $j \to k \to i$ with $i \neq j$ are precisely the non-backtracking walks of length
% two (the intermediate node $k$ is automatically distinct from both endpoints, as $A$
% is hollow). Hence $A^{\langle 2 \rangle}$ is simultaneously the self-avoiding-walk
% matrix and the non-backtracking count at length two. We highlight that several remedies
% used to handle recovery in the sparse-graph regime---self-avoiding walks \cite{massoulie_community_2014},
% the non-backtracking operator \cite{krzakala_spectral_2013,
% bordenave_nonbacktracking_2018}, and graph powering \cite{abbe_graph_2020}---all
% coincide with \eqref{eq:fanout-def} at this order.

We highlight that our theory does not require knowledge of $O$, as we will study a spectral
clustering algorithm applied directly on $R$. That said, knowledge of $O$ is 
useful to certify separation of the communities
and conditioning (see e.g \Cref{thm:full-identifiability}).

\section{Related work}
\label{sec:related}

We discuss how our results relate to the existing literature relating to statistical 
theory for aggregated relational data, and the identifiability and 
estimation theory of stochastic
block models.

ARD originates from the ``network scale-up'' method for estimating the size of
hard-to-reach populations \cite[e.g.][]{killworth_estimation_1998,killworth_social_1998},
and was given statistical foundations by a sequence of likelihood-based
models \cite{zheng_how_2006, mccormick_how_2010, mccormick_latent_2015} and
design-based estimators \cite{feehan_generalizing_2016}; see \cite{laga_thirty_2021}
for a review. \cite{breza_using_2020} use ARD to fit generative models to simulate
network statistics, and \cite{breza_consistently_2023} prove
consistency of ARD-based estimates for various network models, including 
community recovery in the stochastic
block model by thresholding ARD responses. These results are the closest to ours,
but use random traits and operate only in the dense regime; ours hold for a fixed
design $O$ in the sparse regime $n \rho_n \gtrsim \log n$, and concern spectral
estimators.

\cite{alidaee_recovering_2020} recover the (low-rank) expected adjacency matrix
from ARD by nuclear-norm penalized regression, and \cite{auerbach_identification_2022}
matches nodes by similar columns of the squared adjacency matrix of a fully
observed network, echoing our
higher-order-query analysis. \cite{jones_scalable_2021} fit a
mixed-membership block model to the ARD matrix $OA$, using spectral clustering of
a normalized version of $OA$ to initialize a variational fit; our results supply guarantees
for the analogous spectral step under (DC)SBM assumptions. To our knowledge no prior work
gives guarantees for spectral community-detection for the (DC)SBM directly from
$R = OA$, nor a quantitative characterization of community separation in terms of a
fixed design $O$.

There is an extensive literature on the behavior of spectral clustering on the fully
observed adjacency matrix $A$, including e.g \cite{lei_consistency_2015,
rohe_spectral_2011,su_strong_2019, sussman_consistent_2012,
athreya_statistical_2018,lyzinski_perfect_2014,
rubin-delanchy_statistical_2022, cape_two--infinity_2018};
the review article \cite{abbe_community_2018} summarizes their application 
to community detection.
We note that our theoretical results adapt the arguments of
\cite{cape_two--infinity_2018,cai_rate-optimal_2020,abbe_entrywise_2020,lei_unified_2019}
to handle non-identity aggregation operators $O$.
For balanced, uniformly conditioned designs that are quantitatively spread
or near-orthogonal, which include $O=I_n$, we will show that
spectral embedding followed by single linkage exhibits exact recovery
behavior with degrees of order
$K^2+K\log n$. This agrees with the full-adjacency convex-recovery
bound achieved in \cite{chen_statistical_2016}. The information-theoretic
exact-recovery degree threshold is of order $K\log n$
\cite{chen_statistical_2016} with fixed-community sharp
thresholds given by \cite{abbe_general_2015,abbe_exact_2016}.
We do not establish the information-theoretic thresholds
when only $R$ is observed.

\cite{ghoroghchian_graph_2021} give closed-form conditions for recovering the
community composition of each coarse node from a \emph{two-sided} coarsening
$B A B^\top$ of the SBM, which cannot be directly applied to our setting.
Community detection from other compressed or partial
observations---randomized sketches of a fully observed graph
\cite{tremblay_compressive_2016,guo_randomized_2023}
and sampled node pairs or adjacency sub-blocks
\cite{li_network_2020,chen_network_2018}---are related in spirit but
use different observation operators than the known aggregation $O$ of ARD.

\section{Identifiability}
\label{sec:identifiability}

We now study when the aggregation $R=OA$ preserves the community structure of the
mean signal $M:=OP^0$, which differs from the observable mean $\bbe{R}=OP$ by the
matrix $O\diag(P^0)$. Proofs of all results in this section are
collected in \Cref{app:ident-proofs}.

\begin{lemma}
    \label{lem:column-constancy}
    For the SBM, the $j$-th column of $P^0$ depends on $j$ only through its community
    label: $P^0_{\cdot j} = Z B_{\cdot z_j}$. Hence $M_{\cdot j} = \mu_{z_j}$, where
    $\mu_b := (OZ) B_{\cdot b} \in \mathbb{R}^L$.
    For the DCSBM, $P^0_{\cdot j} = \theta_j \, \Theta Z B_{\cdot z_j}$, so
    $M_{\cdot j} = \theta_j \, \nu_{z_j}$ with $\nu_b := (O \Theta Z) B_{\cdot b}$.
\end{lemma}

Thus $M$ assigns to each community $b$ a single \emph{profile}
$\mu_b$ (SBM) or, since $\theta_j > 0$ is unknown, a single \emph{ray}
$\mathbb{R}_{\geq 0}\, \nu_b$ (DCSBM).

\begin{definition}[Community collapse]
    \label{def:collapse}
    Communities $b\neq b'$ \emph{collapse} under the design $O$ if they have the
    same profile, $\mu_b=\mu_{b'}$, in the SBM, or the same ray,
    $\mathbb R_{\geq0}\nu_b=\mathbb R_{\geq0}\nu_{b'}$, in the DCSBM. The design
    \emph{separates} the communities if no pair collapses.
\end{definition}

We begin by asking, given the model $(z, B)$, which designs separate its
communities. Define, for each ordered pair of communities $b \neq b'$, the \emph{differential
connectivity profile}
\begin{equation}
    \label{eq:diff-profile}
    w_{bb'} := Z (B_{\cdot b} - B_{\cdot b'}) \in \mathbb{R}^n
\end{equation}
whose $i$-th entry measures how much more likely a community-$b$ node is to attach to
node $i$ than a community-$b'$ node is.

\begin{proposition}[Pairwise collapse]
    \label{prop:pairwise}
    Communities $b \neq b'$ collapse under $O$ if and only if $O w_{bb'} = 0$.
\end{proposition}

\begin{corollary}[Separation of the partition]
    \label{cor:partition}
    The design $O$ separates the communities if and only if $w_{bb'} \notin \ker O$
    for every pair $b \neq b'$, and then the equality classes of the columns of $M$
    are exactly the communities. In particular, degree queries separate the
    communities if and only if the $K$ self-loop-inclusive degree profiles
    $D_1, \dots, D_K$ are all distinct.
\end{corollary}

\begin{example}
    \label{ex:planted}
    Take degree queries $O = \mathbf{1}_n^\top$, and write 
    $D_b := \sum_a n_a B_{ab}$.
    For the symmetric planted partition with $K = 2$ and
    $B = \bigl(\begin{smallmatrix} p & q \\ q & p \end{smallmatrix}\bigr)$,
    the differential profile is
    \begin{equation}
        \label{eq:planted-degree-gap}
        \mu_1-\mu_2=D_1 - D_2 = (n_1 - n_2)(p - q),
    \end{equation}
    which vanishes precisely when $n_1 = n_2$, since $p\neq q$.
    Equal-sized communities therefore collapse under degree queries, whereas the
    design separates them as soon as $n_1 \neq n_2$.
\end{example}

\begin{example}[Degree confounding in the DCSBM]
    \label{ex:dcsbm-degree}
    With $O = \mathbf{1}_n^\top$ in the DCSBM, $\nu_b = \sum_a s_a B_{ab}$ where
    $s_a = \sum_{i \in C_a} \theta_i$ is a scalar. When all profiles are nonzero, every $\nu_b$ is a positive
    scalar, so all rays $\mathbb{R}_{\geq 0} \nu_b$ coincide: a large observed degree
    is explained equally well by a large $\theta_j$ or by membership in a
    high-connectivity community. These degree-query profiles
    therefore do not separate the communities.
\end{example}

\begin{remark}[Collapse and recovery]
    \label{rem:collapse}
    The population embedding of \Cref{sec:recovery}, row-normalized in the DCSBM,
    has one row per community, and two communities share a row exactly when they
    collapse (\Cref{lem:pop-embedding}); no rule that labels node $j$ through
    $M_{\cdot j}$, or its ray, separates collapsed communities. Collapse concerns
    $M$, not the law of $R$: for the degree query,
    $\mathrm{Cov}(d_i,d_j)=P_{ij}(1-P_{ij})$ for $i\neq j$, which in
    \Cref{ex:planted} is $p(1-p)$ within and $q(1-q)$ across communities, so for
    $p+q\neq1$ the joint law of the degrees determines the partition although the
    communities collapse. Mean-based methods such as Algorithm~1 cannot use this.
    Quantitatively, $\min_{b\neq b'}\|Ow_{bb'}\|_2$ controls the centroid
    separation $\delta$ through \eqref{eq:delta-margin}.
\end{remark}

Notably, separation may hold with very few queries---even $L = 1$, as
the degree case shows---because it only requires the $K$ profiles to be
\emph{distinct}, not that $O$ be invertible on the block subspace. Recovering the
connectivity matrix $B$ from $M$ is strictly stronger.

\begin{example}
    \label{ex:rank-blind}
    Take $K = 3$ communities of equal size $m = n/3$ and
    \[
        B = \begin{pmatrix} 0.4 & 0.2 & 0.3 \\ 0.2 & 0.4 & 0.3 \\
        0.3 & 0.3 & 0.3 \end{pmatrix},
    \]
    so that $\rk(B) = 2<K$, while the columns of $B$ are pairwise distinct. Every
    differential profile
    points in a single direction, $B_{\cdot b} - B_{\cdot b'} \in
    \operatorname{span}\{d\}$ with $d := (1, -1, 0)^\top$, so by
    \Cref{prop:pairwise} a design separates either \emph{all} pairs or \emph{none},
    and does so exactly when $(OZ) d \neq 0$---that is, when some query weights
    $C_1$ and $C_2$ unequally.
\end{example}

Let $\mathcal{S} := \range(Z) = \operatorname{span}\{\mathbf{1}_{C_1}, \dots,
\mathbf{1}_{C_K}\}$ denote the $K$-dimensional \emph{block subspace} of the graph. 
Note that every column of $P^0$ lies in $\mathcal{S}$.

\begin{theorem}[Recovering the model from $M$]
    \label{thm:full-identifiability}
    Consider an SBM with membership matrix $Z$ and connectivity matrix $B$ in the
    interior of $[0,1]^{K \times K}$. The block signal
    $M = OP^0$ determines the partition $z$ and the connectivity matrix $B$, up to a
    common relabelling of the communities, if and only if the design is injective on
    the block subspace $\mathcal{S} = \range(Z)$. This is equivalent to requiring that
    $\rk(OZ) = K$, or that $\ker(O) \cap \mathcal{S} = \{0\}$. 
    Thus $L \geq K$ is
    necessary. If the rank condition fails, every $B$ has arbitrarily close
    alternatives $B' \neq B$ with $OZB'Z^\top = OZBZ^\top$.
\end{theorem}

We have similar results for degree corrected SBMs.

\begin{proposition}[Separation and recovery in the DCSBM]
    \label{prop:dcsbm}
    Under the DCSBM, suppose $\nu_b := (O \Theta Z)
    B_{\cdot b} \neq 0$ for all $b$. 
    \begin{enumerate}
        \item \textup{(Partition.)} The design $O$ separates the communities if and
        only if the $K$ rays $\{\mathbb{R}_{\geq 0} \nu_b\}_{b=1}^K$ are distinct,
        i.e.\ no two of $\nu_1, \dots, \nu_K$ are positively collinear, and then the
        rays of the columns of $M$ determine the partition. Degree queries
        ($O = \mathbf{1}_n^\top$) never satisfy this.
        \item \textup{(Full recovery.)} If $\rk(O \Theta Z) = K$ and the columns
        of $B$ are nonzero and pairwise not positively proportional, then $M$ determines
        $P^0$, and the gauge normalization determines
        $(\theta, B, z)$ up to relabelling.
    \end{enumerate}
\end{proposition}

\begin{remark}[Second-order queries]
    \label{rem:second-order-identifiability}
    For second-order queries, the ideal population signal is
    $M_2=O(P^0)^2=OZBNBZ^\top$, where $N=Z^\top Z$.
    Thus communities $b\ne b'$ have distinct population profiles precisely when
    $OZBNB(e_b-e_{b'})\ne0$: the first-order separation criterion applies with
    $B$ replaced by $BNB$, although a given design may satisfy one criterion
    without satisfying the other. If $B$ is nonsingular, then
    $\rk(OZBNB)=\rk(OZB)=\rk(OZ)$, so second-order queries with the same
    design cannot repair a rank deficiency in the ideal population signal.
\end{remark}

\section{Community Recovery}
\label{sec:recovery}

We now turn to community recovery under the SBM and DCSBM, for both first-order and
fan-out designs. All proofs are collected in
\Cref{app:sbm-proofs,app:dcsbm-proofs,app:fanout-proofs}. Recall that 
we write $M := O P^0$, and let $r := \rk(M)$, and $V_M$ for its top-$r$
right singular vectors.

\paragraph{Algorithm 1}
Given $R$, the number of communities $K$, and the retained rank $r$,
compute the top-$r$ right singular
vectors $\widehat V\in\mathbb R^{n\times r}$. Let $X=\widehat V$ for
the SBM. For the DCSBM, let $X$ be its row-normalized version, mapping
zero rows to a fixed unit vector. We consider two clustering methods.

\emph{Approximate $k$-means.} For a fixed $\varepsilon_{\rm km}\geq0$,
choose $K$ centers and assignments with squared objective at most
$1+\varepsilon_{\rm km}$ times the global minimum:
\begin{equation}
 \label{eq:algorithm-kmeans-objective}
 \sum_i\|X_{i\cdot}-\hat\mu_{\hat z_i}\|_2^2
 \leq(1+\varepsilon_{\rm km})
       \min_{c,z'}\sum_i\|X_{i\cdot}-c_{z'_i}\|_2^2,
\end{equation}
where $c=(c_1,\ldots,c_K)\in(\mathbb R^r)^K$ and $z'\in[K]^n$.
For our exact-recovery guarantees, we also recompute the means of the
assigned clusters and reassign each row to its nearest mean.

\emph{Single linkage.} Construct a minimum spanning tree on the rows of
$X$ with Euclidean edge weights and delete its $K-1$ largest edges;
return the connected components \cite{gower_minimum_1969}.

We measure the error of $\hat z$ against the truth $z$ up to relabelling,
\begin{equation}
    \label{eq:err}
    n_{\mathrm{err}}(\hat z, z) := \min_{\pi \in \mathrm{Sym}(K)}
    \#\{ i \in [n] : \pi(\hat z_i) \neq z_i \},
\end{equation}
We say the algorithm is \emph{weakly consistent} if $n_{\mathrm{err}}/n \to 0$ in
probability. Our exact-recovery results are high-probability statements of
the form $\mathbb P(n_{\mathrm{err}}=0)\geq1-Cn^{-\alpha}$ for each fixed $n$.
Throughout we work in the sparsified regime $P^0 = \rho_n \bar P$ with $\bar P$ of
constant order and $\rho_n \in (0,1]$ controlling the average degree $\asymp n
\rho_n$. We operate in the
relatively sparse regime $n \rho_n
\geq C_0 \log n$ for a constant $C_0=C_0(\|\bar P\|_\infty,\alpha)$.

\subsection{How the design influences bounds}
\label{subsec:design-anatomy}

The design $O$ will enter all of our guarantees through five scalars, which we
discuss below. The natural \emph{signal-to-noise ratio} for the problem is
\begin{equation}
    \label{eq:eps-def}
    \varepsilon_{\rm SNR} \;:=\; \frac{\| O \|_{2 \to 2}\, \sqrt{n \rho_n}}{\sigma_r(OP^0)},
\end{equation}
the noise measured against the signal spectral gap. We require a
\emph{block-subspace conditioning} term
\begin{equation}
    \label{eq:price}
    \kappa_{\mathcal{S}}(O)^2
    \;:=\; \frac{\| O \|_{2\to2}^2\, n_{\min}}{\sigma_+(OZ)^2},
\end{equation}
where $\sigma_+(X)$ is the smallest nonzero singular value of $X$. It measures
the operator norm of $O$ against its smallest nonzero action on the block
subspace. The idea is that this is lower if we create
queries that correlate with
community structure. Directions of the block subspace that $O$ annihilates do
not enter $\kappa_{\mathcal S}(O)$; they lower the rank $r$ of the signal and
enter through the centroid separation $\delta$ of \Cref{subsec:sbm-main}. For
degree queries, for instance, $L=1$, $r=1$, and $\kappa_{\mathcal S}(O)=1$ for
balanced communities. We also introduce
a \emph{per-node leverage} $b_\infty := \| O^\top U_M \|_{2\to\infty}$, which
measures how strongly any single node loads the informative left
singular directions $U_M$ of $M$. For indicator
queries in which each node is counted by $O(1)$ groups, $b_\infty = O(1)$
independently of $L$ and $n$. The largest column norm is $b_O := \max_j \| O_{\cdot j} \|_2$.
Finally we introduce the \emph{Gram overlap}
$g_1(O):=\max_j\|O^\top O_{\cdot j}\|_1$ and
$\varsigma(O):=g_1(O)/(\|O\|_{2\to2}b_O)$.

\begin{definition}[Design regimes]
    \label{def:design-regimes}
    Fix constants $C_2,C_\varsigma>0$. We call a design $O$ \emph{spread} if
    $b_O\leq C_2\|O\|_{2\to2}\sqrt{K/n}$, and \emph{near-orthogonal} if
    $g_1(O)\leq C_\varsigma\|O\|_{2\to2}b_O$. A design that is neither is
    \emph{generic in its columns}, and for it we use only
    $b_O\leq\|O\|_{2\to2}$. For the stacked designs of \Cref{subsec:fanout} we use
    the same terms with $(\|O\|_{2\to2},b_O,g_1(O))$ replaced by the quantities
    $(\Omega_\omega,b_\omega,g_\omega)$ of \Cref{app:fanout-proofs}.
\end{definition}

\begin{example}
\label{ex:enriched-groups}
Suppose prior information lets us form groups enriched for different
communities. Let $n_a=m=n/K$. Keep a fraction $1-\tau$ of community $a$
in query group $G_a$, and distribute the remaining fraction $\tau$ equally
among all $K$ groups. Then
\[
    |G_g\cap C_a|
    =m\left((1-\tau)\mathbf1\{g=a\}+\frac{\tau}{K}\right),
\]
for fixed $0<\tau<1$, assuming integer counts. By \Cref{lem:design-examples},
the design is spread with $C_2=1$, and the conditioning bound of
\Cref{cor:sbm-regimes} holds with $c_1=1-\tau$. It is not near-orthogonal, since
nodes in the same group have identical columns.
\end{example}

More generally, by \Cref{lem:design-examples}, indicator queries of a partition of
$[n]$ into $L$ equal-sized groups are spread when $L\lesssim K$ and near-orthogonal
when $L\asymp n$; the identity, the case $L=n$, is near-orthogonal but not spread.

\subsection{Main guarantee: first-order designs}
\label{subsec:sbm-main}

We will now show that in certain sparsity regimes, depending
on whether the design is spread or near-orthogonal or simply generic,
we can achieve exact recovery for community detection. For the stochastic block model, let $\mu_a$ be the population centroids of
\Cref{lem:pop-embedding}(i). Set
\[
    h_\star:=n_{\min}^{-1/2},\qquad
    \Delta_\star:=\delta:=\min_{a\neq b}\|\mu_a-\mu_b\|_2.
\]

For the degree-corrected SBM, let $\tilde\mu_a$ be the population row directions of
\Cref{lem:pop-embedding}(ii). Set
\begin{equation}
    \label{eq:dcsbm-rownorm}
    \begin{aligned}
    h_\star&:=\|V_M\|_{2\to\infty},\,
    \delta_{\mathrm{row}}:=\min_i\|(V_M)_{i\cdot}\|_2,\\
    \delta_\theta&:=\min_{a\neq b}
    \left\|\frac{\tilde\mu_a}{\|\tilde\mu_a\|_2}
    -\frac{\tilde\mu_b}{\|\tilde\mu_b\|_2}\right\|_2,\,
    \Delta_\star:=\tfrac12\delta_{\mathrm{row}}\delta_\theta.
    \end{aligned}
\end{equation}

\begin{theorem}
    \label{thm:sbm-main}
    \label{thm:dcsbm-main}
    Run Algorithm~1 with $r$ components and the corresponding
    model flag. In the DCSBM case assume $\delta_{\mathrm{row}}>0$.
    Fix $\alpha>1$. There are constants $C_0,c_\star,C$, depending only on
    $\|\bar P\|_\infty$ and $\alpha$, such that if
    $n\rho_n \geq C_0 \log n$,
    $\varepsilon_{\rm SNR} \leq c_\star$ and $b_O\log n\leq c_\star\sigma_r(OP^0)$,
    then with probability at least $1 - 4n^{-\alpha}$
    there is an orthogonal $Q \in \mathbb{R}^{r \times r}$ with $\| \widehat V Q -
    V_M \|_{2 \to \infty} \leq \eta_\infty^\star$, where
    \begingroup
    \setlength{\multlinegap}{0pt}
    \begin{multline}
        \label{eq:eta-infty}
        \eta_\infty^\star
        =C\frac{\sqrt{\rho_n}(\sqrt r+\sqrt{\log n})\,\|O\|_{2\to2}}{\sigma_r(OP^0)}
        +C\varepsilon_{\rm SNR}h_\star\\
        {}+C\frac{(b_\infty+b_Oh_\star+\vartheta)\log n}{\sigma_r(OP^0)}
        +C\varepsilon_{\rm SNR}^2\frac{b_\infty+\vartheta}{\|O\|_{2\to2}},
    \end{multline}
    \endgroup
    with $\vartheta:=b_O\varepsilon_{\rm SNR}$. If in addition
    $\varsigma(O)\varepsilon_{\rm SNR}\leq c_\star$ and
    $\varsigma(O)\|O\|_{2\to2}\log n\leq c_\star\sigma_r(OP^0)$, the same holds
    with $\vartheta:=0$.
    We write this expression as $\eta_\infty$ in the SBM and
    $\eta_\infty^\theta$ in the DCSBM.
On the same event, the same orthogonal matrix $Q$ satisfies
    \begin{equation}
        \label{eq:master-frobenius}
        \|\widehat VQ-V_M\|_F\leq F_\star
        :=2\sqrt{2r}\,C_0\varepsilon_{\rm SNR}.
    \end{equation}
    If $\Delta_\star>0$ and $\eta_\infty^\star<\Delta_\star/4$,
    Algorithm~1 with single linkage recovers $z$ exactly,
    up to relabeling, by \Cref{lem:single-linkage} and, for the DCSBM,
    \Cref{lem:spherical}. For the approximate-$k$-means option, require
    additionally the Frobenius condition in
    \begin{equation}
        \label{eq:master-rounding}
        C_{\rm km}\frac{F_\star^2}{\Delta_\star^2}<n_{\min},
        \qquad \eta_\infty^\star<\Delta_\star/4,
    \end{equation}
    where $C_{\rm km}=16(1+\sqrt{1+\varepsilon_{\rm km}})^2$. Under these conditions,
    Algorithm~1's approximate-$k$-means option with nearest-centroid refinement recovers $z$ exactly, up to
    relabelling, by \Cref{lem:kmeans}.
\end{theorem}

\begin{corollary}[Exact-recovery regimes, SBM and DCSBM]
    \label{cor:sbm-regimes}
    \label{cor:dcsbm-regimes}
    Assume the model hypotheses of \Cref{thm:sbm-main}, with
    $r=\rk(OP^0)$. For the SBM require
    \begin{equation}
        \label{eq:regime-conditions}
        \begin{gathered}
            \sigma_+(OZ)\geq c_1\|O\|_{2\to2}\sqrt{n_{\min}},\qquad
            \delta\geq c_\delta\,n_{\max}^{-1/2},\\
            \sigma_r(OP^0)\geq c_1\bar\sigma\|O\|_{2\to2}\rho_nn_{\min},
        \end{gathered}
    \end{equation}
    for constants $\bar\sigma,c_\delta>0$. For the DCSBM require
    $\theta_{\max}/\theta_{\min}\leq\bar\kappa$, where $\theta_{\min},\theta_{\max}$
    are the extreme entries of $\theta$, and \eqref{eq:regime-conditions} with
    $O\Theta Z$ and $\tilde n_{\min}$ in place of $OZ$ and $n_{\min}$, and with
    $\delta_\theta\geq c_\delta$ and
    $\delta_{\mathrm{row}}\geq c_\delta\theta_{\min}\tilde n_{\max}^{-1/2}$ in place of
    the condition on $\delta$; here $\tilde n_{\min},\tilde n_{\max}$ are the extreme
    values of $\tilde n_a:=\sum_{i\in C_a}\theta_i^2$. If $\rk(OZ)=K$ (resp.\
    $\rk(O\Theta Z)=K$), the signal condition holds with
    $\bar\sigma=\sigma_r(\bar B)$ by \Cref{lem:signal}. If $r=K$, the separation
    conditions hold with $c_\delta=1$ by \Cref{lem:pop-embedding}; in the SBM,
    $\delta\geq\sqrt{2/n_{\max}}$.
    Each condition below ensures exact recovery by Algorithm~1 with
    probability at least $1-4n^{-\alpha}$. The $k$-means option includes
    the nearest-centroid refinement.
    
    \textup{(i) (Spread designs.)} If $O$ is spread
        (\Cref{def:design-regimes}), it suffices that
        \begin{equation}
            \label{eq:sbm-sharp-degree}
            n\rho_n\geq c
            \begin{cases}
                K^2+K\log n,&\text{single linkage},\\
                K^3+K\log n,&\text{approx. }k\text{-means}.
            \end{cases}
        \end{equation}
    
        \textup{(ii) (Near-orthogonal designs.)} If $O$ is near-orthogonal, then
        \eqref{eq:sbm-sharp-degree} suffices.
        
        \textup{(iii) (Generic column designs.)} If
        $b_\infty\leq C_2\|O\|_{2\to2}\sqrt{K/n}$, write
        $s:=b_O/\|O\|_{2\to2}\in(0,1]$, $\varsigma=\varsigma(O)$,
        $A_s:=s^{2/3}n^{1/3}(K^{5/3}+K(\log n)^{2/3})$ and
        $G_\varsigma:=\varsigma^2K^2+\varsigma K\log n$. Single linkage suffices when
        \begin{equation}
            \label{eq:sbm-generic-degree}
            n\rho_n\geq c\big(K^2+K\log n+\min\{A_s,G_\varsigma\}\big),
        \end{equation}
        and approximate $k$-means when in addition $n\rho_n\geq cK^3$.
    The constant $c$ depends only on $c_1,C_2,C_\varsigma,\|\bar B\|_\infty,
    \bar\sigma,c_\delta,\alpha$, the fixed balance bounds, and, where applicable,
    $\bar\kappa$ and $\varepsilon_{\rm km}$.
\end{corollary}

% \begin{remark}[Selecting the embedding rank]
%     \label{rem:rank-selection}
%     We assume that the rank $r$ is known, which is frequent in the
%     literature. When it is not, Weyl's inequality and
%     \eqref{eq:total-noise-event} show that
%     counting the singular values of $R$ strictly above $2 C_0 \| O \|_{2\to2}
%     \sqrt{n\rho_n}$ recovers $r$ exactly whenever $\varepsilon_{\rm SNR} \leq 1/(4C_0)$.
%     This rule requires knowledge of $\|O\|_{2\to2}$, $\rho_n$ and the constant
%     $C_0$, which is not explicit.
%     \cite{taing_misspecifying_2026} gives a precise account of the effects of
%     misspecifying $r$.
% \end{remark}

\subsection{Second-order designs}
\label{subsec:fanout}

We now consider supplementing these observations with counts
of longer walks through the network, including the friend-of-friend
queries in \eqref{eq:fanout-question}. We first study how combining query
orders affects the rank of the block signal, then give recovery guarantees
at order two. For the stochastic block model, the signals at each order
are defined using $(P^0)^m=ZB^{(m)}Z^\top$, where
$B^{(m)}=B(NB)^{m-1}$, as in \Cref{rem:second-order-identifiability}.
Allowing a possibly
different battery $O_k \in \mathbb{R}^{L_k \times n}$ at each order and stacking,
write
\begin{equation}
    \label{eq:stacked}
    \begin{gathered}
        R^{(m)} := \big[ (O_1 A)^\top, \dots, (O_m A^m)^\top \big]^\top, \\
        H := \big[ (O_1 Z B^{(1)})^\top, \dots, (O_m Z B^{(m)})^\top \big]^\top,
    \end{gathered}
\end{equation}
Define the block signal $M^{(m)}:=HZ^\top$. Its right singular subspace 
equals $\mathcal S$ if and
only if $\rk(H)=K$. For $m\geq3$, $R^{(m)}$ is an ideal observation, and
\Cref{prop:stacked-rank} concerns only $H$.
Here and below, we say that a
property holds for \emph{generic} $B$ if it holds off the zero set of some nonzero
polynomial in the entries of $B$---a closed, Lebesgue-null set.

\begin{proposition}
    \label{prop:stacked-rank}
    For generic $B$,
    \begin{equation}
        \label{eq:stacked-rank}
        \rk(H) \;=\; \min\!\Big(K,\ \textstyle\sum_{k=1}^m \rk(O_k Z)\Big),
    \end{equation}
    so the population right singular space equals the entire block subspace
    iff $\sum_{k=1}^m\rk(O_kZ)\geq K$. Population partition
    separation instead requires pairwise distinct columns of $H$.
\end{proposition}

We develop the statistical theory at order two, for the fan-out matrix $A^{\langle
2 \rangle} = A^2 - D_A$ of \eqref{eq:fanout-def}.
Fix non-zero batteries $O_1, O_2$ and a block weight $\omega > 0$,
and set
\begin{equation}
    \label{eq:fanout-stack}
    R_\omega := \begin{bmatrix} O_1 A \\ \omega\, O_2 A^{\langle 2 \rangle}
    \end{bmatrix},
    H_\omega := \begin{bmatrix} O_1 Z B \\ \omega\, O_2 Z B^{(2)} \end{bmatrix},
\end{equation}
$M_\omega := H_\omega Z^\top$ and $B^{(2)} := B N B$.
For the degree-corrected stochastic block model, let
$\widetilde N:=Z^\top\Theta^2Z$ and set
\begin{equation}
    \label{eq:fanout-dcsbm-H}
    H^\theta_\omega:=
    \begin{bmatrix}
        O_1\Theta ZB\\
        \omega O_2\Theta ZB\widetilde N B
    \end{bmatrix}.
\end{equation}
The master theorem, \Cref{thm:fanout-main} in
\Cref{app:fanout-proofs}, gives the following recovery guarantees.

\begin{corollary}[Exact-recovery regimes, fan-out queries for SBM and DCSBM]
    \label{cor:fanout-regimes}
    Assume the model hypotheses of \Cref{thm:fanout-main}, with balanced
    communities, and fix $\alpha>1$. Take
    $\omega=((1+\|\bar B\|_\infty)n\rho_n)^{-1}$ and write
    \[
        \Omega_\omega:=\|O_1\|_{2\to2}+\|O_2\|_{2\to2},\qquad
        \bar H^\star:=\rho_n^{-1}H_\omega^\star,
    \]
    where $(H_\omega^\star,n_{\min}^\star)=(H_\omega,n_{\min})$ for
    the stochastic block model and
    $(H_\omega^\theta,\tilde n_{\min})$ for the degree-corrected model.
    Assume, with $r_\omega:=\rk(H_\omega^\star)$,
    \begin{equation}
        \label{eq:fanout-well-cond}
        \sigma_{r_\omega}(\bar H^\star)\geq c_1\Omega_\omega\sqrt{n_{\min}^\star},
    \end{equation}
    and the separation conditions of \Cref{cor:sbm-regimes} for the stacked
    embedding, with $\delta$, $\delta_\theta$ and $\delta_{\mathrm{row}}$ computed
    for $M_\omega^\star$; if $r_\omega=K$, these hold with $c_\delta=1$.
    For the degree-corrected model, also require
    $\theta_{\max}/\theta_{\min}\leq\bar\kappa$.
    Each condition below ensures exact recovery by Algorithm~1 applied to
    $R_\omega$ with $r_\omega$ components and the corresponding model flag, with
    probability at least $1-4n^{-\alpha}$. The $k$-means method includes
    nearest-centroid refinement. The column quantities below are defined
    in \eqref{eq:fanout-Omega}--\eqref{eq:fanout-gram} in \Cref{app:fanout-proofs}.

    \textup{(i) (Spread designs.)} If the stacked design is spread
    (\Cref{def:design-regimes}), it suffices that
    \begin{equation}
        \label{eq:fanout-spread-degree}
        n\rho_n\geq c
        \begin{cases}
            K^2+K\log n,&\text{single linkage},\\
            K^3+K\log n,&\text{approx. }k\text{-means}.
        \end{cases}
    \end{equation}
    \textup{(ii) (Near-orthogonal designs.)} If the stacked design is
    near-orthogonal and
    $b_\infty^\omega+b_{\infty,2}^\omega\leq C_2\Omega_\omega\sqrt{K/n}$, then
    \eqref{eq:fanout-spread-degree} suffices. This includes $O_1=O_2=I_n$.
    The constant $c$ depends only on $c_1,c_\delta,C_2,C_\varsigma,\|\bar P\|_\infty,\alpha$,
    the fixed balance bounds, and
    $\bar\kappa$ and $\varepsilon_{\rm km}$. \Cref{cor:fanout-generic} in
    \Cref{app:fanout-proofs} treats generic column designs.
\end{corollary}

\section{Experiments}
\label{sec:experiments}

We study how query composition and second-order responses affect recovery.
The first three experiments use two communities of size $n/2$ and
connectivity matrix $B=\rho_n\bar B_{s,g}$, where
\[
    \bar B_{s,g}=\begin{pmatrix}
        1+s+g/2&1-s\\1-s&1+s-g/2
    \end{pmatrix}.
\]
Here $s$ controls the contrast between within- and between-community
connections, while $g$ controls the difference between the two
within-community connection probabilities. We choose $\rho_n$ to give
average expected degree $d=c\log n$. The first two experiments use
$\bar B_s:=\bar B_{s,0}$; the third varies $g$ to change the expected degrees.

We compare designs on the same graph draws and report the mean misclassification fraction
$n_{\mathrm{err}}/n$ and the frequency of exact recovery. Error bars are
95\% normal intervals for mean misclassification and 95\% Wilson intervals
for exact-recovery frequency. We used Algorithm~1's embedding at the population
signal rank, followed by $k$-means with 20 initializations and one
nearest-centroid refinement, normalizing rows for the degree-corrected model.
Lloyd's algorithm with 20 initializations need not return the
$(1+\varepsilon_{\rm km})$-approximate solution assumed in \Cref{lem:kmeans}.
We do not run single linkage.

For the queries we partition the vertices into groups $G_1,G_2$, where
$G_1$ contains a fraction $a$ of the first community and a fraction $1-a$ of
the second, and $G_2$ is its complement, choosing vertices uniformly
at random. We use the true community assignments to form groups with the
prescribed overlap $a$. The design has
$\kappa_{\mathcal S}(O)=|2a-1|^{-1}$ for $a>1/2$. In the stochastic block
model with $g=0$, setting $a=1/2$ makes the two community profiles identical.

\begin{figure*}[!t]
    \centering
    \includegraphics[width=0.8\textwidth]{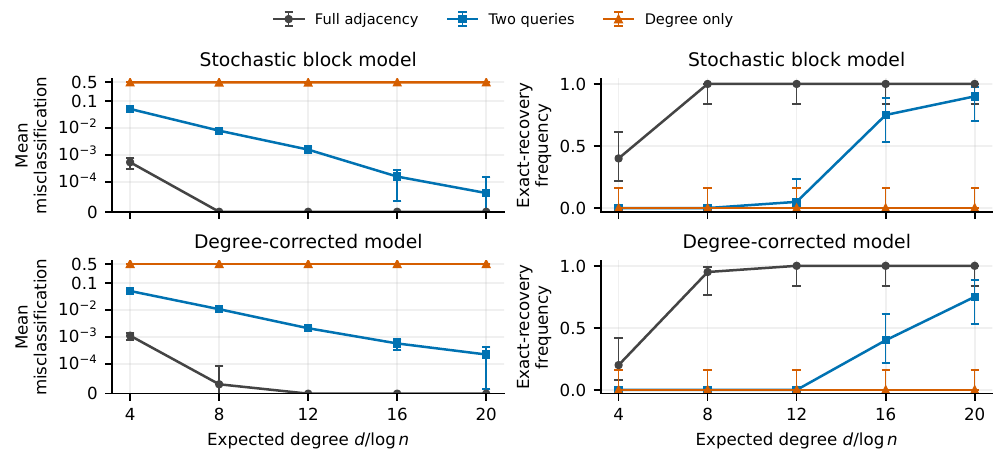}
    \caption{First-order recovery at $n=1600$, $s=0.6$, $g=0$, and $a=0.75$,
    with 20 graph draws per point. Degree parameters $0.8$ and $1.2$ occur
    equally often within each community of the degree-corrected model.
    The left vertical axes are logarithmic above $10^{-4}$ and linear
    below it to include zero error.}
    \label{fig:exp-first-order}
\end{figure*}

The first comparison fixes $s=0.6$ and $a=0.75$
(\Cref{fig:exp-first-order}). In the stochastic block model, the two-query
design has $\kappa_{\mathcal S}(O)=2$, compared with $1$ for full adjacency,
and doubles $\varepsilon_{\rm SNR}$ in \Cref{thm:sbm-main}. In the simulations,
exact recovery from two queries requires a larger average expected degree:
at $n=1600$ and
$d=8\log n$, two queries give mean error $0.79\%$ and no exact recoveries,
while full adjacency recovers exactly in all 20 runs. The difference between
small error and exact recovery illustrates the distinction in
\Cref{thm:sbm-main} between controlling total embedding error and keeping
every row within the population separate. Degree alone remains near
$50\%$ error, consistent with the coincident degree profiles in
\Cref{cor:partition}. The degree-corrected comparison shows the same
qualitative dependence on the design.

\begin{figure*}[!t]
    \centering
    \includegraphics[width=0.75\textwidth]{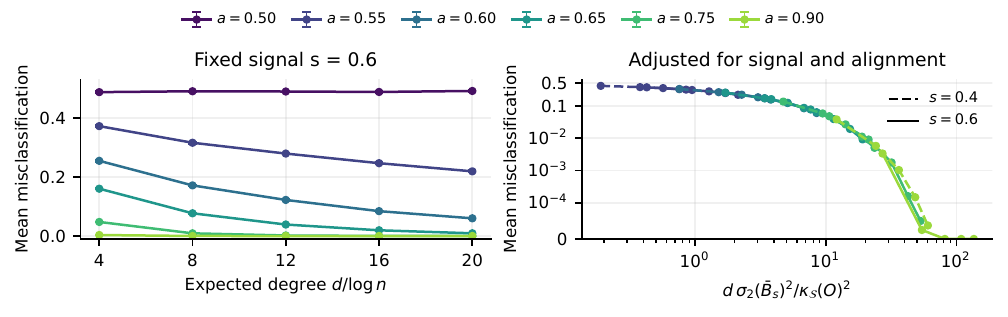}
    \caption{Query alignment at $n=1600$ and $g=0$, with 50 graph draws per setting.
    Left: $s=0.6$. Right: $s\in\{0.4,0.6\}$ after rescaling by signal
    strength and design conditioning; $a=0.5$ is omitted because its
    rescaled signal is zero. The right vertical axis is logarithmic above
    $10^{-4}$ and linear below it.}
    \label{fig:exp-conditioning}
\end{figure*}

To isolate the effect of group composition in the stochastic block model,
we hold the number of queries fixed at two and vary $a$ from $0.5$ to $0.9$,
using $s\in\{0.4,0.6\}$.
For $a>1/2$, \eqref{eq:eps-def} gives
\[
    \varepsilon_{\rm SNR}^2=
    \frac{4\kappa_{\mathcal S}(O)^2}{n\rho_n\,\sigma_2(\bar B_s)^2},
    \qquad \sigma_2(\bar B_s)=2s.
\]
Thus the rescaling in \Cref{fig:exp-conditioning} compares the same
perturbation parameter across different signal strengths and query
compositions. The approximate agreement of the rescaled curves supports the
joint dependence on signal strength and design conditioning in
\Cref{thm:sbm-main}.

\begin{figure*}[!t]
    \centering
    \includegraphics[width=0.75\textwidth]{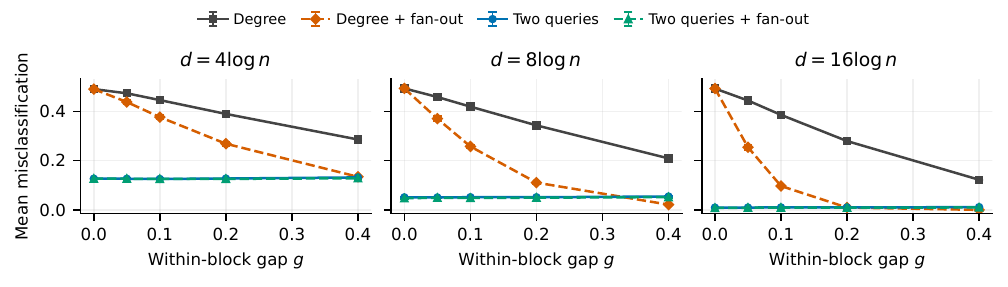}
    \caption{Second-order queries under $\bar B_{0.4,g}$ at $n=3200$
    and alignment $a=0.75$, with 30 graph draws per setting and weight
    $\omega=\{(1+\|\bar B_{0.4,g}\|_\infty)n\rho_n\}^{-1}$.}
    \label{fig:exp-fanout}
\end{figure*}

The second-order response comparison fixes $s=0.4$ and $a=0.75$ and varies
$g\in\{0,0.05,0.1,0.2,0.4\}$. We stack each first-order response with
its weighted fan-out response as in \eqref{eq:fanout-stack}.
For degree queries, the first-order population profile difference is
$n\rho_n g/2$. When $g>0$, these profiles are already distinct, but
the stacked signal has rank two for the tested parameters and increases
the separation of the population embedding. \Cref{thm:fanout-main}
compares error with this new separation while accounting for the additional
fluctuations from second-order responses.
At $n=3200$, $d=16\log n$, and $g=0.2$, adding fan-out reduces mean error
from $27.9\%$ to $0.98\%$ (\Cref{fig:exp-fanout}). The gain is much smaller
for the already informative two-query design. At $g=0$, the two communities have equal profiles
at both orders, and neither degree-based method recovers them, as predicted
by \Cref{prop:stacked-rank}.

We also vary the number of communities, using $K$ queries. For balanced
communities at fixed $n$, the average expected degree needed to reach a
given exact-recovery frequency increases with $K$ for full adjacency.
For the query design, this dependence is nonmonotone, consistent with its
improving conditioning as $K$ grows. \Cref{app:experiments-k} gives the
design, supporting plots, and comparisons with unequal community sizes.

\section{Discussion}
\label{sec:discussion}

We develop a theory to understand how aggregated relational
data can be used to solve network tasks without the full adjacency
matrix, and show that such data can be used (under certain conditions
on the design) to achieve exact community detection. We show that these occur at the same degree threshold as clustering on the
full adjacency matrix depending on the design,
with $n\rho_n\gtrsim K^2+K\log n$ sufficing for ``spread'' and near-orthogonal
designs (which include $O=I_n$). We also develop a theory for 
second-order queries which behaves similarly. Our recovery
results operate at $n\rho_n\gtrsim\log n$; going beyond
this, where only partial recovery is possible, would require studying other techniques
\cite{massoulie_community_2014, bordenave_nonbacktracking_2018, abbe_graph_2020}
which lie outside of our current framework.
The information-theoretic thresholds for
community detection \emph{from aggregated data}---as opposed to the full adjacency
matrix, whose thresholds lower-bound them---are not resolved here and remain a
worthwhile future direction.

\bibliographystyle{plainnat}
\bibliography{references}

\appendix
\begin{center}
    \rule{\textwidth}{4pt}\par
    \vspace{0.8em}
    {\Large\bfseries Supplementary Materials:\\
    Community Detection from Aggregated Relational Data: Identifiability and Spectral Recovery\par}
    \vspace{0.8em}
    \rule{\textwidth}{1pt}
\end{center}
\vspace{1em}

\section{Linear algebra: norms, subspace distances, and perturbation theorems}
\label{sec:linalg}
We collect some linear-algebraic tools used throughout.

\subsection{Norms}

For a matrix $A \in \mathbb{R}^{n \times m}$,
we write $A_{i \cdot}$ for the $i$-th row of $A$, a row vector, and
$A_{\cdot j}$ for the $j$-th column of $A$. We consider the following
operator and entry-wise
matrix norms defined for $p, q \geq 1$:
\begin{equation}
    \| A \|_{p \to q} := \max_{x \,:\, \| x \|_p \leq 1} \| Ax \|_{q}, \qquad
    \| A \|_{p,q} := \left( \sum_{i=1}^n
    \left( \sum_{j=1}^m |A_{ij}|^p \right)^{q/p} \right)^{1/q}.
\end{equation}
Recall that for the vector $\ell_p$ norms we have the
following variational
characterization: for $x \in \mathbb{R}^d$, we have that
\begin{equation}
    \| x \|_p := \max_{y \,:\, \| y \|_q \leq 1} \langle y, x \rangle
    \qquad \text{ where } \qquad \frac{1}{p} + \frac{1}{q} = 1.
\end{equation}
As a result, we can therefore write the $\ell_p \to \ell_q$ operator norm of $A$
as, writing $q'$ for the H\"{o}lder conjugate of $q$,
\begin{equation}
    \| A \|_{p \to q} = \max_{x, y \,:\, \|x\|_p \leq 1, \| y \|_{q'} \leq 1}
    \langle y, A x \rangle.
\end{equation}
As some special cases, we have that $\| A \|_{2 \to 2} = \sigma_1(A)$,
$\| A  \|_{2, 2} = \|A \|_F = (\sum_i \sigma_i(A)^2)^{1/2}$ is the
Frobenius norm of $A$, and
\begin{equation}
    \| A \|_{2 \to \infty}
    = \max_i \|A_{i\cdot} \|_2
    := \| A \|_{2,\infty}.
\end{equation}
(See Proposition~6.1 of \cite{cape_two--infinity_2018} for a proof.)
Moreover, $\| A \|_{\infty \to \infty} = \max_i \sum_j | A_{ij} |$ is the largest
absolute row sum, and $\| A \|_{1 \to \infty} = \max_{i, j} | A_{ij} |$ is the
largest absolute entry, which we write $\| A \|_\infty$.
We also have the following relations, valid for any $A \in \mathbb{R}^{n \times m}$ of
rank $r$ and any conformable $B$:
\begin{equation}
    \label{eq:norm-relations}
    \| A \|_{2 \to \infty} \leq \| A \|_{2 \to 2} \leq \| A \|_F
    \leq \sqrt{r}\, \| A \|_{2 \to 2}, \qquad
    \| A B \|_{2 \to \infty} \leq
    \| A \|_{2 \to \infty}\, \| B \|_{2 \to 2}.
\end{equation}
The first chain holds because $\| A_{i \cdot} \|_2 = \| e_i^T A \|_2 \leq \| A \|_{2 \to
2}$ for every row $i$, and $\| A \|_F^2 = \sum_{i \leq r} \sigma_i(A)^2$ lies between
$\sigma_1(A)^2$ and $r\, \sigma_1(A)^2$; the second holds because $\| (AB)_{i \cdot}
\|_2 = \| A_{i \cdot} B \|_2 \leq \| A_{i \cdot} \|_2\, \| B \|_{2 \to 2}$.
We write
$A = U \Sigma V^T$ for the SVD of $A$, where $U$ and $V$ are the matrices
whose columns are the left and right singular vectors of $A$ respectively,
and $\Sigma$ is a diagonal matrix consisting of the singular values
$\sigma_i(A)$ of $A$. In particular,
$\Sigma = \diag(\sigma_1(A), \sigma_2(A), \ldots )$
where $\sigma_i(A)$
are decreasing in $i$ with $\sigma_i(A) \geq 0$ for all $i$.

We write $\mathbb{O}_{p \times r} := \{ W \in \mathbb{R}^{p \times r} \,:\,
W^\top W = I_r \}$ for the Stiefel manifold of $p \times r$ matrices
with $r$ $p$-dimensional orthonormal columns. If $p = r$, then
we just write $\mathbb{O}_p$ for the manifold of $p \times p$
orthonormal matrices.

We use the following elementary facts throughout the appendices, usually without
reference. Courant--Fischer below refers to the variational characterizations
$\sigma_j(A)=\max_{\dim S=j}\min_{x\in S,\|x\|_2=1}\|Ax\|_2$ and, for symmetric $X$,
$\lambda_j(X)=\max_{\dim S=j}\min_{x\in S,\|x\|_2=1}x^\top Xx$, where
$\lambda_1(X)\geq\lambda_2(X)\geq\cdots$ are the eigenvalues of $X$
\cite{horn_matrix_2013}.

\begin{lemma}[Norm inequalities]
    \label{lem:norm-facts}
    Let $A$, $B$ and $X$ be real matrices, $x$ and $y$ real vectors, of conformable
    dimensions, and let $\star\in\{2\to2,F\}$.

    \textup{(i) Transposes and columns.} $\|A^\top\|_{2\to2}=\|A\|_{2\to2}$,
    $\|A\|_{2\to2}^2=\|AA^\top\|_{2\to2}$, and
    $\max_j\|A_{\cdot j}\|_2=\|A^\top\|_{2\to\infty}\leq\|A\|_{2\to2}$.

    \textup{(ii) Products.} $\|AB\|_{2\to2}\leq\|A\|_{2\to2}\|B\|_{2\to2}$,
    $\|AB\|_F\leq\min\{\|A\|_{2\to2}\|B\|_F,\ \|A\|_F\|B\|_{2\to2}\}$,
    $\|xy^\top\|_{2\to2}=\|xy^\top\|_F=\|x\|_2\|y\|_2$, and
    $\|AB\|_{2\to\infty}\leq\|A\|_{\infty\to\infty}\|B\|_{2\to\infty}$. In
    particular, $\|AXB\|_\star\leq\|X\|_\star$ whenever
    $\|A\|_{2\to2}\leq1$ and $\|B\|_{2\to2}\leq1$, for instance when $A$ and $B$
    are orthogonal projections.

    \textup{(iii) Symmetric matrices.} If $A=A^\top$, the singular values of $A$ are
    the absolute values of its eigenvalues, and
    $\|A\|_{2\to2}\leq\|A\|_{\infty\to\infty}$.

    \textup{(iv) Diagonal matrices.} $\|\diag(x)\|_{2\to2}=\|x\|_\infty$,
    $\|\diag(x)B\|_{2\to\infty}\leq\|x\|_\infty\|B\|_{2\to\infty}$ and
    $\|\diag(x)B\|_F\leq\|x\|_2\|B\|_{2\to\infty}$. For square $X$,
    $\|\diag(X)\|_{2\to2}=\max_i|X_{ii}|\leq\|X\|_{2\to2}$, and hence
    $\|\hol(X)\|_{2\to2}\leq2\|X\|_{2\to2}$.

    \textup{(v) Submatrices and stacked matrices.} Deleting rows or columns of $A$
    does not increase $\|A\|_{2\to2}$, $\|A\|_F$ or $\|A\|_{2\to\infty}$. For
    $A=[A_1^\top,A_2^\top]^\top$ and $B=[B_1^\top,B_2^\top]^\top$ with conformable
    blocks, $A^\top B=A_1^\top B_1+A_2^\top B_2$,
    $\|A\|_F^2=\|A_1\|_F^2+\|A_2\|_F^2$,
    $\|A\|_{2\to\infty}=\max_k\|A_k\|_{2\to\infty}$ and
    \[
        \max_k\|A_k\|_{2\to2}\leq\|A\|_{2\to2}
        \leq\big(\|A_1\|_{2\to2}^2+\|A_2\|_{2\to2}^2\big)^{1/2}
        \leq\|A_1\|_{2\to2}+\|A_2\|_{2\to2}.
    \]

    \textup{(vi) Singular value comparisons.} If $A$ has $q$ columns and rank $q$,
    then $\sigma_q(A)\,\sigma_j(B)\leq\sigma_j(AB)\leq\|A\|_{2\to2}\,\sigma_j(B)$ for
    every $j$. If $X\succeq Y\succeq0$ are symmetric, then
    $\lambda_j(X)\geq\lambda_j(Y)$ for every $j$; in particular,
    $\sigma_j([A_1^\top,A_2^\top]^\top)\geq\sigma_j(A_1)$.
\end{lemma}

\begin{proof}
    (i) $A$ and $A^\top$ have the same nonzero singular values, and those of
    $AA^\top$ are their squares. The column bound is the first inequality of
    \eqref{eq:norm-relations} applied to $A^\top$.

    (ii) The spectral bound follows from the definition of the operator norm. For the
    Frobenius bounds,
    $\|AB\|_F^2=\sum_j\|AB_{\cdot j}\|_2^2\leq\|A\|_{2\to2}^2\|B\|_F^2$, and the
    second follows by transposition. The matrix $xy^\top$ has rank one, with
    nonzero singular value $\|x\|_2\|y\|_2$. Row $i$ of $AB$ is
    $\sum_jA_{ij}B_{j\cdot}$, whose norm is at most
    $\sum_j|A_{ij}|\,\|B_{j\cdot}\|_2\leq\|A\|_{\infty\to\infty}\|B\|_{2\to\infty}$.
    The last claim follows from the first two bounds.

    (iii) The first claim is the spectral theorem. For the second, let $Av=\lambda v$
    with $v\neq0$ and choose $i$ with $|v_i|=\|v\|_\infty$. Then
    $|\lambda|\,|v_i|=|\sum_jA_{ij}v_j|\leq\|A\|_{\infty\to\infty}|v_i|$, so every
    eigenvalue, and hence $\|A\|_{2\to2}$, is at most $\|A\|_{\infty\to\infty}$.

    (iv) Multiplying by $\diag(x)$ on the left scales row $i$ by $x_i$, which gives
    the first three claims; for the third,
    $\|\diag(x)B\|_F^2=\sum_ix_i^2\|B_{i\cdot}\|_2^2
    \leq\|x\|_2^2\|B\|_{2\to\infty}^2$. For square $X$,
    $|X_{ii}|=|e_i^\top Xe_i|\leq\|X\|_{2\to2}$, and
    $\hol(X)=X-\diag(X)$ gives the last bound.

    (v) Deleting rows or columns multiplies $A$ by coordinate selection matrices,
    whose spectral norm is one; deleting rows removes rows, and deleting columns
    shortens each row. The block identities follow by writing out the products,
    and the spectral bounds from
    $\|A_kx\|_2\leq\|Ax\|_2$ and $\|Ax\|_2^2=\|A_1x\|_2^2+\|A_2x\|_2^2$.

    (vi) Since $\sigma_q(A)\|y\|_2\leq\|Ay\|_2\leq\|A\|_{2\to2}\|y\|_2$ for every
    $y$, Courant--Fischer applied to $B$ and $AB$ gives the singular value
    comparison. Courant--Fischer for eigenvalues gives the monotonicity, and the
    last claim follows from
    $[A_1^\top,A_2^\top]\,[A_1^\top,A_2^\top]^\top
    =A_1^\top A_1+A_2^\top A_2\succeq A_1^\top A_1$ and
    $\sigma_j(C)^2=\lambda_j(C^\top C)$.
\end{proof}

\begin{lemma}[Matrices with orthonormal columns]
    \label{lem:orthonormal-facts}
    Let $W\in\mathbb O_{p\times r}$, $Q\in\mathbb O_r$, and let $X$ and $G$ be
    matrices of conformable dimensions.

    \textup{(i) Norms.} $\|W\|_{2\to2}=\|W^\top\|_{2\to2}=1$, $\|W\|_F=\sqrt r$ and
    $\|W\|_{2\to\infty}\leq1$. Hence, for $\star\in\{2\to2,F\}$,
    $\|WX\|_\star=\|X\|_\star$ and $\|W^\top X\|_\star\leq\|X\|_\star$, and
    $\|W^\top x\|_2=\|x\|_2$ for every $x$ in the range of $W$.

    \textup{(ii) Projections.} $WW^\top$ and $I-WW^\top$ are the orthogonal
    projections onto the range of $W$ and its orthogonal complement, and
    \[
        \|G^\top(I-WW^\top)X\|_{2\to\infty}
        \leq\|G^\top X\|_{2\to\infty}+\|G^\top W\|_{2\to\infty}\|W^\top X\|_{2\to2}.
    \]

    \textup{(iii) Orthogonal invariance.} Right multiplication by $Q$ preserves the
    Euclidean norm of every row, so $\|XQ\|_\star=\|X\|_\star$ for
    $\star\in\{2\to2,F,2\to\infty\}$, and normalizing the nonzero rows of $X$ to unit
    length commutes with right multiplication by $Q$. Moreover,
    $\|QYQ^\top\|_{2\to2}=\|Y\|_{2\to2}$ for $Y\in\mathbb R^{r\times r}$, the rows of
    $Q$ are orthonormal, and any two elements of $\mathbb O_{p\times r}$ with the same
    range differ by right multiplication by an element of $\mathbb O_r$.

    \textup{(iv) Singular vectors.} Let $R$ have top-$r$ singular value decomposition
    $\widehat U\widehat\Sigma\widehat V^\top$ with $\widehat\sigma_r:=\sigma_r(R)>0$.
    Then $R\widehat V=\widehat U\widehat\Sigma$ and
    $R^\top\widehat U=\widehat V\widehat\Sigma$, so
    $\widehat U=R\widehat V\widehat\Sigma^{-1}$,
    $\widehat V=R^\top\widehat U\widehat\Sigma^{-1}$ and
    $\|\widehat\Sigma^{-1}\|_{2\to2}=1/\widehat\sigma_r$. If $\rk(M)=r$ and
    $M=U\Sigma V^\top$ is its compact singular value decomposition, then $U$ and $V$
    span the column and row spaces of $M$, and
    $M_{\cdot j}=U\Sigma V_{j\cdot}^\top$ for every $j$. Since $U\Sigma$ has rank $r$,
    two rows of $V$ are equal, or positive multiples of one another, exactly when the
    corresponding columns of $M$ are.
\end{lemma}

\begin{proof}
    (i) $W^\top W=I_r$ gives $\|Wy\|_2=\|y\|_2$ for every $y$, so
    $\|W\|_{2\to2}=\|W^\top\|_{2\to2}=1$ by \Cref{lem:norm-facts}(i),
    $\|W\|_F^2=\operatorname{tr}(W^\top W)=r$, and every row has norm at most $\|W\|_{2\to2}=1$ by
    \eqref{eq:norm-relations}. Applying $\|Wy\|_2=\|y\|_2$ to $y=Xz$ and to each
    column of $X$ gives $\|WX\|_\star=\|X\|_\star$, and \Cref{lem:norm-facts}(ii)
    gives $\|W^\top X\|_\star\leq\|X\|_\star$. If $x=Wy$, then $W^\top x=y$ and
    $\|x\|_2=\|y\|_2$.

    (ii) $WW^\top$ is symmetric and idempotent with the same range as $W$. The
    display follows from $(I-WW^\top)X=X-W(W^\top X)$, the triangle inequality and
    \eqref{eq:norm-relations}.

    (iii) Row $i$ of $XQ$ is $X_{i\cdot}Q$, and $\|yQ\|_2=\|y\|_2$ for every row
    vector $y$ because $QQ^\top=I_r$; normalizing $X_{i\cdot}Q$ gives
    $(X_{i\cdot}/\|X_{i\cdot}\|_2)Q$. The matrix $QYQ^\top$ has the same singular
    values as $Y$, and $QQ^\top=I_r$ states that the rows of $Q$ are orthonormal. If
    $W'\in\mathbb O_{p\times r}$ has the same range as $W$, then $WW^\top W'=W'$, so
    $W'=W(W^\top W')$ and $(W^\top W')^\top W^\top W'=W'^\top WW^\top W'=W'^\top W'=I_r$.

    (iv) The first claims are the singular value decomposition written column by
    column. If $\rk(M)=r$, then $M=U(\Sigma V^\top)$ and $M^\top=V(\Sigma U^\top)$
    with $\Sigma V^\top$ and $\Sigma U^\top$ of rank $r$, so the column and row spaces
    of $M$ are the ranges of $U$ and $V$. The identity
    $M_{\cdot j}=Me_j=U\Sigma V_{j\cdot}^\top$ and the injectivity of $U\Sigma$ give
    the last claim.
\end{proof}

\subsection{Subspace distances and perturbation bounds}

For $V, V' \in \mathbb{O}_{p \times r}$, the singular values of $V^\top V'$ lie in
$[0, 1]$ by \Cref{lem:orthonormal-facts}(i), and we define
\begin{equation}
    \Theta(V, V') = \diag(\cos^{-1}(\sigma_1(V^\top V')),
       \ldots, \cos^{-1}(\sigma_r(V^\top V'))).
\end{equation}
Since $V^\top V'$ and $V'^\top V$ have the same singular values,
$\Theta(V, V') = \Theta(V', V)$.
We can then define distances between the column spaces of $V$ and $V'$
via taking norms of the matrix
$\sin \Theta(V, V')$ where $\sin(\cdot)$ is
applied entrywise; for instance, we can define the metrics
\begin{align}
    D_{2 \to 2}(V, V') &:= \| \sin \Theta(V, V') \|_{2 \to 2}, \\
    D_{F}(V, V') &:= \| \sin \Theta(V, V') \|_F.
\end{align}
these define metrics on the Grassmannian of $r$-dimensional subspaces of
$\mathbb{R}^p$ that are invariant to the choice of orthonormal basis for each subspace
\cite{cai_rate-optimal_2020}. We can also define the metrics
\begin{align}
    d_{2 \to 2}(V, V') & := \min_{ Q \in \mathbb{O}_r} \| V - V' Q \|_{2 \to 2}, \\
    d_{F}(V, V') & := \min_{Q \in \mathbb{O}_r} \| V - V' Q \|_F.
\end{align}
The following lemma shows that the two types of metrics are equivalent in a
dimension-free fashion \cite[Lemma~1]{cai_rate-optimal_2020}, and identifies a
single alignment that attains the upper bounds in both norms.

\begin{lemma}[Procrustes alignment]
    \label{lem:procrustes}
    Let $V, \widehat V \in \mathbb{O}_{p \times r}$, write
    $\Theta := \Theta(\widehat V, V)$, and let
    $V^\top \widehat V = Q_1 (\cos\Theta) Q_2^\top$ be a singular value
    decomposition. We call $Q := Q_2 Q_1^\top \in \mathbb{O}_r$ the
    \emph{Procrustes alignment} of $\widehat V$ to $V$. It has the following
    properties, where $\star \in \{2 \to 2, F\}$.

    \textup{(i)} $Q$ minimizes $\| \widehat V Q' - V \|_F$ over
    $Q' \in \mathbb{O}_r$.

    \textup{(ii)} $\| (I - VV^\top) \widehat V \|_\star = \| \sin\Theta \|_\star$.

    \textup{(iii)} $V^\top \widehat V Q = Q_1 (\cos\Theta) Q_1^\top$, and
    $\| V^\top \widehat V Q - I \|_{2 \to 2} \leq \| \sin\Theta \|_{2 \to 2}^2$.

    \textup{(iv)} $\| \widehat V Q - V \|_\star \leq \sqrt 2\, \| \sin\Theta \|_\star$.

    \textup{(v)} $\| \widehat V Q - V \|_\star = \| \widehat V - V Q^\top \|_\star$,
    and the same holds for $\| \cdot \|_{2 \to \infty}$.

    In particular, $D_\star(\widehat V, V) \leq d_\star(\widehat V, V) \leq \sqrt 2\,
    D_\star(\widehat V, V)$.
\end{lemma}

\begin{proof}
    Write $\cos\Theta = \diag(\cos\theta_1, \ldots, \cos\theta_r)$, whose entries
    lie in $[0, 1]$.

    (i) For $Q' \in \mathbb{O}_r$,
    $\| \widehat V Q' - V \|_F^2 = 2r - 2 \operatorname{tr}(Q'^\top \widehat V^\top V)
    = 2r - 2 \operatorname{tr}(Q_1^\top Q'^\top Q_2 \cos\Theta)$. The diagonal entries
    of the orthogonal matrix $Q_1^\top Q'^\top Q_2$ are at most one, so the trace is
    at most $\sum_k \cos\theta_k$, with equality for $Q' = Q_2 Q_1^\top$.

    (ii) The Gram matrix of $(I - VV^\top) \widehat V$ is
    $I - \widehat V^\top V V^\top \widehat V = Q_2 (I - \cos^2\Theta) Q_2^\top
    = Q_2 (\sin^2\Theta) Q_2^\top$, so its singular values are the
    $\sin\theta_k$.

    (iii) $V^\top \widehat V Q = Q_1 (\cos\Theta) Q_2^\top Q_2 Q_1^\top
    = Q_1 (\cos\Theta) Q_1^\top$, so
    $\| V^\top \widehat V Q - I \|_{2 \to 2} = \max_k (1 - \cos\theta_k)
    \leq \max_k (1 - \cos^2\theta_k) = \| \sin\Theta \|_{2 \to 2}^2$.

    (iv) By (iii),
    $(\widehat V Q - V)^\top (\widehat V Q - V) = 2I - V^\top \widehat V Q
    - (V^\top \widehat V Q)^\top = 2 Q_1 (I - \cos\Theta) Q_1^\top$, whose
    eigenvalues $2(1 - \cos\theta_k)$ are at most $2 \sin^2\theta_k$. The largest
    eigenvalue gives the spectral bound and the trace the Frobenius bound.

    (v) $\widehat V Q - V = (\widehat V - V Q^\top) Q$; apply
    \Cref{lem:orthonormal-facts}(iii).

    For the final claim, $(I - VV^\top)(\widehat V - V Q') = (I - VV^\top) \widehat V$
    for every $Q' \in \mathbb{O}_r$, so \Cref{lem:norm-facts}(ii) and (ii) give
    $\| \widehat V - V Q' \|_\star \geq \| \sin\Theta \|_\star$. The upper bound is
    (iv) combined with (v).
\end{proof}

\begin{lemma}[Weyl's inequality \cite{horn_matrix_2013}]
    \label{lem:weyl}
    For $A, B \in \mathbb{R}^{p \times q}$ and every $k$,
    $| \sigma_k(A + B) - \sigma_k(A) | \leq \| B \|_{2 \to 2}$. In particular, if
    $\rk(M) = r$, then $\sigma_r(M + N) \geq \sigma_r(M) - \| N \|_{2 \to 2}$ and
    $\sigma_{r+1}(M + N) \leq \| N \|_{2 \to 2}$.
\end{lemma}

The Davis--Kahan $\sin\Theta$ theorem \cite{davis_rotation_1970} concerns a
symmetric matrix and a symmetric perturbation of it. It bounds the principal angles
between an eigenspace of one and the corresponding eigenspace of the other by a
residual of the perturbation divided by a gap between eigenvalues. Wedin's
$\sin\Theta$ theorem \cite{wedin_perturbation_1972} is the analogue for the left and
right singular subspaces of rectangular matrices. It has one residual for each of
the two subspaces, and its gap is between singular values. The matrix $R = OA$ is
rectangular and not symmetric, so all of our subspace perturbation bounds are
instances of Wedin's theorem. When the matrices are symmetric, the top-$r$ left and
right singular subspaces are both spanned by eigenvectors for the $r$ eigenvalues of
largest absolute value, and Wedin's theorem gives a bound of Davis--Kahan type; this
is the case $O = I_n$ discussed in \Cref{loo:rem-literature}. We state Wedin's
theorem for a general pair of matrices in \Cref{loo:wedin-res}, which the
leave-one-out arguments apply with the full-rank unperturbed matrix $R^{(i)}$, and
for a rank-$r$ signal in \Cref{lem:wedin}.

Both lemmas use the following notation. For $M, N \in \mathbb{R}^{p \times q}$, let $R := M + N$ and fix $r \leq \min(p,
q)$. Let $U, V$ collect top-$r$ left and right singular vectors of $M$, and let
$\widehat U \widehat\Sigma \widehat V^\top$ be the top-$r$ singular value
decomposition of $R$. Let $Q_U$ and $Q_V$ be the Procrustes alignments
(\Cref{lem:procrustes}) of $\widehat U$ to $U$ and of $\widehat V$ to $V$.

\begin{lemma}[Wedin's $\sin\Theta$ theorem]
    \label{loo:wedin-res}
    If
    \begin{equation}
        \label{eq:loo-wedin-gap}
        \gamma := \sigma_r(R) - \sigma_{r+1}(M) > 0,
    \end{equation}
    then
    \begin{equation}
        \label{eq:wedin-spectral}
        \max\big\{ \| \sin\Theta(\widehat U, U) \|_{2 \to 2},\
        \| \sin\Theta(\widehat V, V) \|_{2 \to 2} \big\}
        \;\leq\; \frac{\max\{ \| N \widehat V \|_{2 \to 2},\
        \| N^\top \widehat U \|_{2 \to 2} \}}{\gamma}
    \end{equation}
    and
    \begin{equation}
        \label{eq:loo-wedin-res}
        \max\big\{ \| \sin\Theta(\widehat U, U) \|_F,\
        \| \sin\Theta(\widehat V, V) \|_F \big\}
        \;\leq\; \frac{\| N \widehat V \|_F + \| N^\top \widehat U \|_F}{\gamma}.
    \end{equation}
    Consequently,
    \[
        \max\big\{\| \widehat U Q_U - U \|_F,\ \| \widehat V Q_V - V \|_F\big\}
        \;\leq\;
        \frac{\sqrt 2\big(\| N \widehat V \|_F + \| N^\top \widehat U \|_F\big)}{\gamma},
    \]
    and the same bound holds for $\| \widehat U - U Q_U^\top \|_F$ and
    $\| \widehat V - V Q_V^\top \|_F$.
\end{lemma}
    
\begin{proof}
    By \Cref{lem:orthonormal-facts}(iv),
    $M \widehat V - \widehat U \widehat\Sigma = - N \widehat V$ and
    $M^\top \widehat U - \widehat V \widehat\Sigma = - N^\top \widehat U$. These are
    the two residuals in Wedin's theorem \cite{wedin_perturbation_1972}, restated
    for the spectral norm in \cite[Section~2.3]{cai_rate-optimal_2020}. With the gap
    $\gamma$, the theorem
    bounds the two spectral distances by the larger spectral norm of the residuals
    divided by $\gamma$, which is \eqref{eq:wedin-spectral}. In Frobenius norm,
    Wedin's theorem \cite{wedin_perturbation_1972} bounds
    $(\| \sin\Theta(\widehat U, U) \|_F^2 + \| \sin\Theta(\widehat V, V) \|_F^2)^{1/2}$
    by $(\| N \widehat V \|_F^2 + \| N^\top \widehat U \|_F^2)^{1/2} / \gamma$, which
    is at most the right side of \eqref{eq:loo-wedin-res}. 
    The Procrustes bounds
    follow from \Cref{lem:procrustes}(iv) and (v).
\end{proof}

\begin{lemma}[Wedin's theorem for a rank-$r$ signal]
    \label{lem:wedin}
    \label{lem:sin-theta}
    Suppose that $\rk(M) = r$, write $M = U \Sigma V^\top$ for its compact singular
    value decomposition and $\sigma_r := \sigma_r(M)$, and write $Q := Q_V$. If
    $\| N \|_{2 \to 2} \leq \sigma_r/2$, then
    \begin{equation}
        \label{eq:wedin}
        \max\big\{ \| \sin\Theta(\widehat U, U) \|_{2 \to 2},\
        \| \sin\Theta(\widehat V, V) \|_{2 \to 2} \big\}
        \;\leq\; \frac{2 \| N \|_{2 \to 2}}{\sigma_r},
    \end{equation}
    and
    \begin{equation}
        \label{eq:sin-theta}
        \| \widehat V Q - V \|_{2 \to 2}
        \leq \sqrt 2\, \| \sin\Theta(\widehat V, V) \|_{2 \to 2},
        \qquad
        \| \widehat V Q - V \|_F
        \leq \sqrt{2}\, \| \sin \Theta(\widehat V, V) \|_F
        \leq \frac{2\sqrt{2r}\,\| N \|_{2 \to 2}}{\sigma_r}.
    \end{equation}
    The same bounds hold for $\widehat U Q_U - U$, and, by
    \Cref{lem:procrustes}(v), for $\widehat V - V Q^\top$ and
    $\widehat U - U Q_U^\top$.
\end{lemma}

\begin{proof}
    By \Cref{lem:weyl}, $\gamma = \sigma_r(R) \geq \sigma_r - \| N \|_{2 \to 2} \geq
    \sigma_r / 2$. By \Cref{lem:orthonormal-facts}(i), $\| N \widehat V \|_{2 \to 2}$
    and $\| N^\top \widehat U \|_{2 \to 2}$ are at most $\| N \|_{2 \to 2}$, so
    \eqref{eq:wedin-spectral} gives \eqref{eq:wedin}. Since $\sin\Theta$ is an $r
    \times r$ diagonal matrix, $\| \sin\Theta \|_F \leq \sqrt r\, \| \sin\Theta \|_{2
    \to 2}$, and \Cref{lem:procrustes}(iv) gives \eqref{eq:sin-theta}.
\end{proof}

\subsection{Spherical projection and $k$-means clustering}

The degree-corrected estimator normalizes rows to unit length before clustering;
the following elementary bound converts unnormalized error into normalized error at
the price of the population row norm. We adopt the convention that $u / \| u \|_2$
denotes an arbitrary fixed unit vector when $u = 0$, matching step~2 of
Algorithm~1.

\begin{lemma}[Spherical projection]
    \label{lem:spherical}
    For any $u \in \mathbb{R}^K$ and $v \in \mathbb{R}^K \setminus \{0\}$,
    \begin{equation}
        \label{eq:spherical}
        \Big\| \frac{u}{\| u \|_2} - \frac{v}{\| v \|_2} \Big\|_2
        \;\leq\; \frac{2\, \| u - v \|_2}{\| v \|_2} .
    \end{equation}
\end{lemma}

\begin{proof}
    If $u = 0$ the left side is at most $2$ and $\| u - v \|_2 = \| v \|_2$, so
    \eqref{eq:spherical} holds with equality of the two sides' bounds. If $u \neq
    0$, decompose
    \[
        \frac{u}{\|u\|} - \frac{v}{\|v\|}
        = \frac{u - v}{\|v\|} + \frac{u}{\|u\|} \cdot \frac{\|v\| - \|u\|}{\|v\|},
    \]
    and bound the second summand's norm by $\abs{\|v\| - \|u\|} / \|v\| \leq \| u -
    v \| / \|v\|$ via the reverse triangle inequality.
\end{proof}

The following lemma adapts the $k$-means argument of
\cite[Lemma~5.3]{lei_consistency_2015}; it additionally records the centroid
accuracy needed for exact recovery. The approximation tolerance
$\varepsilon_{\rm km}$ refers to the squared $k$-means objective.

\begin{lemma}[$k$-means: weak and exact recovery]
    \label{lem:kmeans}
    Let $V \in \mathbb{R}^{n\times q}$ have rows $V_{i\cdot}=\mu_{z_i}$ taking
    $K\geq2$ distinct values, with
    \[
        \delta:=\min_{a\neq b}\|\mu_a-\mu_b\|_2>0,
        \qquad n_{\min}:=\min_a\#\{i:z_i=a\}>0.
    \]
    For $\widehat V\in\mathbb{R}^{n\times q}$ and an orthogonal
    $Q\in\mathbb{R}^{q\times q}$, write
    \[
        \varrho:=\|\widehat V-VQ\|_F,
        \qquad \eta:=\|\widehat V-VQ\|_{2\to\infty},
        \qquad C_{\rm km}:=16(1+\sqrt{1+\varepsilon_{\rm km}})^2.
    \]
    Let $\hat z_i\in[K]$ and $\hat\mu_1,\ldots,\hat\mu_K$ be a
    $(1+\varepsilon_{\rm km})$-approximate $k$-means solution, where
    $\varepsilon_{\rm km}\geq0$, and set
    $\widehat C_{i\cdot}:=\hat\mu_{\hat z_i}$. Approximation is measured relative
    to the minimum squared objective over all $K$ centers and assignments;
    in particular,
    $\|\widehat V-\widehat C\|_F^2\leq
    (1+\varepsilon_{\rm km})\varrho^2$.
    Suppose
    \begin{equation}
        \label{eq:kmeans-cond}
        C_{\rm km}\,\frac{\varrho^2}{\delta^2}<n_{\min}.
    \end{equation}
    Then the following statements hold.

    \textup{(i) Misclassification and fitted centers.}
    There is a permutation $\pi$ of $[K]$ such that
    \begin{equation}
        \label{eq:kmeans-bound}
        \#\{i:\pi(\hat z_i)\neq z_i\}
        \leq C_{\rm km}\,\frac{\varrho^2}{\delta^2},
        \qquad
        \max_{b\in[K]}\|\hat\mu_b-\mu_{\pi(b)}Q\|_2<\delta/4.
    \end{equation}

    \textup{(ii) Exact recovery after reassignment.}
    If, in addition,
    \begin{equation}
        \label{eq:kmeans-exact-cond}
        \eta<\delta/4,
    \end{equation}
    then assigning every row of $\widehat V$ to its nearest fitted center
    recovers $z$ exactly, up to relabelling. The same conclusion holds if the fitted
    centers are first recomputed as the means of their assigned rows, as in
    Algorithm~1.
\end{lemma}

\begin{proof}[Proof by the $k$-means argument of
    {\cite[Lemma~5.3]{lei_consistency_2015}}]
    The population centers and labels are a feasible $k$-means solution, so
    the objective guarantee and the triangle inequality give
    \[
        \|\widehat C-VQ\|_F
        \leq\|\widehat C-\widehat V\|_F+\varrho
        \leq(1+\sqrt{1+\varepsilon_{\rm km}})\varrho.
    \]
    Hence the set
    $\mathcal B:=\{i:\|\widehat C_{i\cdot}-\mu_{z_i}Q\|_2\geq\delta/4\}$
    has at most $C_{\rm km}\varrho^2/\delta^2<n_{\min}$ elements. Every
    population community therefore contains a row outside $\mathcal B$, whose
    fitted center lies within $\delta/4$ of its population center. These $K$
    balls are disjoint and there are only $K$ fitted centers. They consequently
    determine a bijection of the centers, and every row outside $\mathcal B$
    has the correct label under that bijection. This proves (i).

    Under \eqref{eq:kmeans-exact-cond}, each observed row has distance less
    than $\eta+\delta/4<\delta/2$ from its matched fitted center and distance
    greater than $\delta-\eta-\delta/4>\delta/2$ from every other fitted
    center. Nearest-center reassignment is therefore exact. All fitted
    clusters are nonempty by (i), and recomputing their means does not increase
    the objective. Part (i) thus also applies to the recomputed centers,
    proving the last assertion.
\end{proof}

\begin{lemma}[Exact recovery by single linkage]
    \label{lem:single-linkage}
    Let $V\in\mathbb R^{n\times q}$ have rows $V_{i\cdot}=\mu_{z_i}$
    taking $K\geq2$ distinct values, all represented, with separation
    $\delta=\min_{a\ne b}\|\mu_a-\mu_b\|_2>0$.
    Suppose there exists an orthogonal matrix $Q$ such that
    $\eta:=\|\widehat VQ-V\|_{2\to\infty}<\delta/4$.
    Regard the rows of $\widehat V$ as points and connect them with
    $n-1$ edges so that all points are connected and the sum of the
    Euclidean edge lengths is as small as possible. Such a set of
    edges is called a minimum spanning tree. Delete its $K-1$ longest
    edges and assign points to the same cluster when they remain connected.
    This procedure recovers $z$ exactly, up to relabeling.
\end{lemma}

\begin{proof}[Proof using
    {\cite[Section~3.2, Lemma~1 and Theorem~3]{ackerman_clusterability_2009}}]
    Orthogonal transformations preserve distances, so we may work with
    the rows of $\widehat VQ$. The triangle inequality gives
    \[
        \max_{i,j:\,z_i=z_j}\|\widehat V_{i\cdot}-\widehat V_{j\cdot}\|_2
        \leq2\eta<\delta-2\eta
        \leq\min_{i,j:\,z_i\ne z_j}\|\widehat V_{i\cdot}-\widehat V_{j\cdot}\|_2.
    \]
    Thus the largest distance within a community is smaller than the
    smallest distance between communities. This is the separation
    condition in the cited result, which gives exact recovery by
    single linkage stopped at $K$ clusters. Deleting the $K-1$ longest edges
    of a minimum spanning tree produces the same clusters
    \cite{gower_minimum_1969}.
\end{proof}

\begin{remark}[Comparison with the $k$-means condition]
    \label{rem:rounding-conditions}
    Write $F=\|\widehat VQ-V\|_F$ and
    $\eta=\|\widehat VQ-V\|_{2\to\infty}$.
    In addition to $\eta<\delta/4$, \Cref{lem:kmeans} assumes
    $C_{\rm km}F^2<n_{\min}\delta^2$ to control the fitted centers.
    Since $F^2\leq n\eta^2$,
    \[
        \frac{C_{\rm km}F^2}{n_{\min}\delta^2}
        \leq C_{\rm km}\frac{n}{n_{\min}}
               \left(\frac{\eta}{\delta}\right)^2.
    \]
    For balanced communities, $n/n_{\min}\asymp K$, so
    $\eta<\delta\sqrt{n_{\min}/(C_{\rm km}n)}$, of order
    $\delta/\sqrt K$ for fixed $\varepsilon_{\rm km}$, suffices for
    this additional condition. For fixed $K$, this only requires a
    smaller constant in the rowwise bound. For growing $K$, balance
    alone does not make $\eta<\delta/4$ imply the Frobenius condition.
\end{remark}

\section{Concentration inequalities}
\label{sec:concentration}
This appendix collects concentration inequalities used throughout. In line with the
conventions of \Cref{sec:prelim}, every statement is for a \emph{fixed} design ($O$,
or a weight matrix $W$) and a random adjacency matrix. In particular,
we only make use of the fact that the noise matrix $E = A - P$ 
has independent entries (up to the diagonal symmetry) 
with entry variances at most $\rho_n \|
\bar P \|_\infty$, meaning every statement covers the SBM and the DCSBM alike.

\begin{lemma}[Spectral norm of a sparse Bernoulli matrix]
    \label{lem:bernoulli-spectral}
    Let $A \in \{0,1\}^{n \times n}$ be symmetric and hollow ($A_{ii} = 0$) with
    independent entries $A_{ij} \sim \mathrm{Bernoulli}(p_{ij})$ for $i < j$ and
    mean $P = \mathbb{E}[A]$, and let $d$ be any number with $n \max_{ij} p_{ij}
    \leq d$---in our setting $d = n\rho_n (1 + \| \bar P \|_\infty)$ always
    qualifies. For each fixed $\alpha>1$ there is a constant $C_0=C_0(\alpha)$
    such that, if $d \geq C_0 \log n$, then
    \begin{equation}
        \| A - P \|_{2 \to 2} \leq C_0 \sqrt{d}
    \end{equation}
    with probability at least $1 - n^{-\alpha}$.
\end{lemma}

This is \cite[Theorem~5.2]{lei_consistency_2015}.

\begin{lemma}[Aggregated noise]
    \label{lem:noise}
    There is a constant $C_0$, depending only on $\| \bar P \|_\infty$ and $\alpha$, such that,
    if $n \rho_n \geq C_0 \log n$, then with probability at least $1 - n^{-\alpha}$ we have
    \begin{equation}
        \label{eq:noise-bound}
        \| OE \|_{2 \to 2} \leq \| O \|_{2 \to 2} \, \| E \|_{2 \to 2}
        \leq C_0 \, \| O \|_{2 \to 2} \, \sqrt{n \rho_n} .
    \end{equation}
\end{lemma}

\begin{proof}
    The first inequality is sub-multiplicativity of the operator norm. For the
    second, $E = A - \bbe{A}$ is a symmetric matrix with independent zero-mean
    entries above the diagonal, bounded by $1$ and with variances at most
    $\max_{ij} P_{ij} \leq \rho_n \| \bar P \|_\infty$; \Cref{lem:bernoulli-spectral}
    gives $\| E \|_{2 \to 2} \leq C_0 \sqrt{n \rho_n}$ with probability at least
    $1 - n^{-\alpha}$ once $n \rho_n \geq C_0 \log n$, absorbing $(1 + \| \bar P
    \|_\infty)$ factors into $C_0$.
\end{proof}

\begin{lemma}
    \label{lem:vector-bernstein}
    Let $X_1,\dots,X_n$ be independent, mean-zero scalars with
    $|X_j|\leq1$ and $\Var(X_j)\leq\sigma^2$, and let
    $w_1,\dots,w_n\in\mathbb R^K$ be fixed vectors.
    Let $W\in\mathbb R^{n\times K}$ have rows $w_j^\top$, and write
    $v:=\sigma^2\sum_j\|w_j\|_2^2=\sigma^2\|W\|_F^2$ and
    $b:=\max_j\|w_j\|_2$.
    For every $\epsilon\in(0,1)$, each of the following bounds holds
    with probability at least $1-\epsilon$.

    \textup{(i)} The Bernstein-type inequality
    \begin{equation}
        \label{eq:vector-bernstein}
        \Big\|\sum_{j=1}^n X_jw_j\Big\|_2
        \leq\sqrt{2v\log\!\big((K+1)/\epsilon\big)}
            +\tfrac23 b\log\!\big((K+1)/\epsilon\big).
    \end{equation}

    \textup{(ii)} The mean-plus-deviation bound, from
    \cite[Theorem~1.1(c)]{klein_concentration_2005},
    \begin{equation}
        \label{eq:vector-mean-tail}
        \Big\|\sum_{j=1}^n X_jw_j\Big\|_2
        \leq C\left\{\sqrt v
            +\sigma\|W\|_{2\to2}\sqrt{\log(1/\epsilon)}
            +b\log(1/\epsilon)\right\},
    \end{equation}
    where $C$ is universal.
    Both statements hold conditionally when the $w_j$ are random but
    measurable with respect to a $\sigma$-field independent of $(X_j)_j$,
    with $W$, $v$, and $b$ evaluated at the realized weights.
\end{lemma}

\begin{proof}
    Part (i) is rectangular matrix Bernstein
    \cite[Theorem~1.6.2]{tropp_introduction_2015} applied to the
    $K\times1$ matrices $X_jw_j$.
    For part (ii), apply \cite[Theorem~1.1(c)]{klein_concentration_2005}
    to the norm as a supremum of centered linear forms. Put
    $Z_0=\|\sum_{j=1}^nX_jw_j\|_2$. Independence gives
    \[
        \mu:=\bbe{Z_0}\leq\sqrt v,\qquad
        v_{\rm dir}:=\sup_{\|u\|_2=1}
            \sum_j\Var(X_j)\langle w_j,u\rangle^2
        \leq\sigma^2\|W\|_{2\to2}^2.
    \]
    If $b=0$ the conclusion is deterministic. Otherwise, represent
    $Z_0$ using a countable dense subset of the unit sphere and rescale
    each function $x\mapsto x\langle w_j,u\rangle$ by $b$.
    The cited inequality gives
    \[
        \bbp{Z_0\geq\mu+x}
        \leq\exp\!\left\{-\frac{x^2}{2(v_{\rm dir}+2b\mu)+3bx}\right\}.
    \]
    Set $t=\log(1/\epsilon)$ and
    $x=\sqrt{2(v_{\rm dir}+2b\mu)t}+3bt$.
    The exponent is at most $-t$, and
    $2\sqrt{b\mu t}\leq\mu+bt$ gives
    $Z_0\leq2\mu+\sqrt{2v_{\rm dir}t}+4bt$ outside an event of
    probability at most $\epsilon$. This proves
    \eqref{eq:vector-mean-tail}. Applying each argument under the
    conditional law proves the final assertion.
\end{proof}

\begin{lemma}[Entrywise aggregated noise]
    \label{lem:entrywise-noise}
Fix any deterministic $W\in\mathbb R^{n\times q}$, with $1\leq q\leq n$.
    There is a constant $C_1$, depending only on $\|\bar P\|_\infty$ and $\alpha$,
    such that, with probability at least $1-n^{-\alpha}$,
    \begin{equation}
        \label{eq:entrywise-noise}
        \|EW\|_{2\to\infty}
        \leq C_1\left(\sqrt{\rho_n}\,\|W\|_F+\sqrt{\rho_n\log n}\,\|W\|_{2\to2}
                      +\|W\|_{2\to\infty}\log n\right).
    \end{equation}
    In particular, take $W=O^\top U_M$, where $U_M$ contains the $r$ left
    singular vectors of $M=OP^0$, and put
    $b_\infty=\|O^\top U_M\|_{2\to\infty}$. Then
    \[
        \|(OE)^\top U_M\|_{2\to\infty}
        \leq C_1\left(\sqrt{\rho_n}(\sqrt r+\sqrt{\log n})\,\|O\|_{2\to2}
                         +b_\infty\log n\right).
    \]
    This specialization uses $\|O^\top U_M\|_F\leq\sqrt r\|O\|_{2\to2}$
    and $\|O^\top U_M\|_{2\to2}\leq\|O\|_{2\to2}$.
    The fixed-weight statement also applies to singular vectors of $OP$;
    no rank assumption on that matrix is needed here.
\end{lemma}

\begin{proof}
    For each $i$, $(EW)_{i\cdot}=\sum_{j\ne i}E_{ij}W_{j\cdot}$
    is a sum of independent centered vectors, with scalar variances
    at most $\rho_n\|\bar P\|_\infty$. Apply
    \Cref{lem:vector-bernstein}(ii) with $\epsilon=n^{-(\alpha+1)}$ and union bound
    over the $n$ rows. This proves \eqref{eq:entrywise-noise} without
    requiring independence between different rows. Symmetry gives
    $(OE)^\top U_M=EO^\top U_M$, proving the specialization.
    For the total perturbation, additionally use
    $\|D^\top U_M\|_{2\to\infty}\leq
    \rho_n\|\bar P\|_\infty b_\infty$; this deterministic term is
    absorbed by the displayed $b_\infty\log n$ term.
\end{proof}

\begin{lemma}[Good event]
    \label{loo:event}
    There are constants $C_0, C_{10}$ depending only on $\| \bar P \|_\infty$ and $\alpha$ such
    that, if $n \rho_n \geq C_0 \log n$, the event
    \begin{equation}
        \label{eq:loo-event}
        \mcE \;:=\;
        \Big\{ \| E \|_{2 \to 2} \leq C_0 \sqrt{n \rho_n} \Big\}
        \;\cap\;
        \Big\{ \max_{i \in [n]} \| E_{\cdot i} \|_2 \leq C_{10} \sqrt{n \rho_n} \Big\}
    \end{equation}
    has probability at least $1 - 2n^{-\alpha}$. On $\mcE$, for every $i \in [n]$
    simultaneously,
    \begin{equation}
        \label{eq:loo-event-consequences}
        \max\big\{\|OE\|_{2\to2},\ \|OE^{(i)}\|_{2\to2}\big\}
        \leq C_0\|O\|_{2\to2}\sqrt{n\rho_n},\qquad
        \|OE_{\cdot i}\|_2
        \leq C_{10}\|O\|_{2\to2}\sqrt{n\rho_n}.
    \end{equation}
    Here $E^{(i)}$ is $E$ with row and column $i$ set to zero. Enlarging
    $C_0,C_{10}$ to include the deterministic correction in
    Lemma~\ref{lem:diag} %% updated this
    % \eqref{eq:total-perturbation},
    the same event gives
    \begin{equation}
        \label{eq:total-noise-event}
        \max\{\|\mcN\|_{2\to2},\|OE^{(i)}-D\|_{2\to2}\}
        \leq C_0\|O\|_{2\to2}\sqrt{n\rho_n}, \qquad 
        \|\mcN_{\cdot i}\|_2 \leq C_{10}\|O\|_{2\to2}\sqrt{n\rho_n}.
    \end{equation}
\end{lemma}

\begin{proof}
    The first event has probability at least $1 - n^{-\alpha}$ by
    \Cref{lem:bernoulli-spectral} (equivalently \Cref{lem:noise}). For the second,
    fix $i$ and expand
    \[
        \begin{aligned}
        \| E_{\cdot i} \|_2^2 &= \sum_{j \neq i} E_{ij}^2
        = \sum_{j \neq i} \big( A_{ij}(1 - P_{ij})^2 + (1 - A_{ij}) P_{ij}^2 \big)\\
        &\leq \deg_i + \sum_{j \neq i} P_{ij}^2
        \leq \deg_i + n \rho_n^2 \| \bar P \|_\infty^2.
        \end{aligned}
    \]
where $\deg_i=\sum_{j\ne i}A_{ij}$ has mean
    $d_i=\sum_{j\ne i}P_{ij}\leq n\rho_n\|\bar P\|_\infty$.
    Bernstein's inequality gives
    \[
        \deg_i\leq d_i+\sqrt{2(\alpha+1)d_i\log n}
                         +\tfrac23(\alpha+1)\log n
    \]
    with probability at least $1 - n^{-(\alpha+1)}$. Under $n\rho_n\geq C_0\log n$ this
    is bounded by $C n\rho_n$, with $C$ depending on $\alpha$ and
    $\|\bar P\|_\infty$. The deterministic term
    $n\rho_n^2\|\bar P\|_\infty^2$ has the same upper-bound scale.
    A union bound over $i$ and intersection with the spectral event prove
    the claimed probability after choosing $C_{10}$ sufficiently large.

    As a consequence, sub-multiplicativity gives $\| OE \|_{2\to2} \leq \| O
    \|_{2\to2} \| E \|_{2\to2}$ and $\| OE^{(i)} \|_{2\to2} \leq \| O \|_{2\to2} \|
    E^{(i)} \|_{2\to2}$. Writing $D_i:=I-e_ie_i^\top$, we have
    $E^{(i)}=D_iED_i$ and $\|D_i\|_{2\to2}=1$, hence
    $\|E^{(i)}\|_{2\to2}\leq\|E\|_{2\to2}$ by sub-multiplicativity. The column
    bound is $\| O
    E_{\cdot i} \|_2 \leq \| O \|_{2 \to 2} \| E_{\cdot i} \|_2$.
    Finally, $\|D\|_{2\to2}\leq\rho_n\|\bar P\|_\infty\|O\|_{2\to2}$
    and $\|D_{\cdot i}\|_2\leq\rho_n\|\bar P\|_\infty b_O$.
    The triangle inequality proves \eqref{eq:total-noise-event}, since
    $\rho_n\leq\sqrt{n\rho_n}$ and $b_O\leq\|O\|_{2\to2}$.
\end{proof}

\section{Procrustes reductions and the leave-one-out chain}
\label{sec:procrustes}
Throughout, as in the notation of \Cref{sec:recovery},
\begin{equation}
    \label{eq:loo-notation}
    \begin{gathered}
    M=OP^0,\qquad N:=\mcN=OE-D,\qquad R=M+N,\\
    M=U\Sigma V^\top\ \text{(compact SVD, }\rk M=r\text{)},\qquad
    \sigma_r=\sigma_r(M).
    \end{gathered}
\end{equation}
and $\widehat U \in \mathbb{R}^{L \times r}$, $\widehat V \in \mathbb{R}^{n \times
r}$ are the top-$r$ left and right singular vectors of $R$. For each node $i \in
[n]$, the leave-one-out matrices are
\begin{equation}
    \label{eq:loo-defs}
    E^{(i)} := (I - e_i e_i^\top)\, E\, (I - e_i e_i^\top),
    \qquad
    R^{(i)} := M + O E^{(i)}-D,
\end{equation}
so $E^{(i)}$ is $E$ with its $i$-th row and column zeroed (since $E_{ii} = 0$), and
$\widehat U^{(i)}, \widehat V^{(i)}$ are the top-$r$ singular vectors of $R^{(i)}$.
Write $N^{(i)}:=OE^{(i)}-D$. Fix the failure exponent $\alpha>1$ of
\Cref{sec:concentration}. In every statement, $C, C', c$ denote positive constants
depending only on $\|\bar P \|_\infty$ and $\alpha$. All
randomness is controlled under the good event $\mcE$ of \Cref{loo:event}.

The bounds in this appendix are stated in terms of the following quantities,
where we abbreviate $o:=\|O\|_{2\to2}$ and
$\varepsilon:=\varepsilon_{\rm SNR}=o\sqrt{n\rho_n}/\sigma_r$ \eqref{eq:eps-def}.
\begin{itemize}
    \item \emph{Design quantities} (\Cref{subsec:design-anatomy}): the
    leverage $b_\infty=\|O^\top U\|_{2\to\infty}$, the column norm
    $b_O=\max_j\|O_{\cdot j}\|_2$, the Gram overlap $g_1=g_1(O)$, and
    $\varsigma=\varsigma(O)=g_1/(ob_O)$. They satisfy
    $b_\infty\leq b_O\leq o$ and $g_1\geq b_O^2$, the latter from the
    diagonal term $\langle O_{\cdot j},O_{\cdot j}\rangle$.
    \item \emph{Deterministic scales}: the first-order noise scale
    $\gamma_1:=\sqrt{\rho_n}(\sqrt r+\sqrt{\log n})$, the incoherence
    $h:=\|V\|_{2\to\infty}$, and the sizes $q_D:=\rho_n\|\bar P\|_\infty b_O$
    and $q_1:=o\,q_D$ of the diagonal correction, so that
    $\|D_{\cdot i}\|_2\leq q_D$ and
    $\|D^\top U\|_{2\to\infty}\leq\rho_n\|\bar P\|_\infty b_\infty$.
    \item \emph{Smallness parameters}: besides $\varepsilon$,
    \[
        \lambda:=\frac{b_O\log n}{\sigma_r},\qquad
        \lambda_\varsigma:=\frac{\varsigma o\log n}{\sigma_r}
        =\frac{g_1\log n}{b_O\sigma_r}\;\geq\;\lambda.
    \]
    \item \emph{Row noise}: the random quantities to be controlled are
    \begin{equation}
        \label{eq:loo-scalars}
        \Gamma := \| N^\top \widehat U \|_{2 \to \infty},
        \qquad
        \Gamma_V:=\|E\widehat V\|_{2\to\infty}=\max_i\|\widehat V^\top E_{\cdot i}\|_2,
        \qquad
        S:=\Gamma+b_O\Gamma_V,
    \end{equation}
    where the second identity for $\Gamma_V$ uses $E=E^\top$. For the
    \emph{fixed} subspace $U$, \Cref{lem:entrywise-noise} achieves the rate
    $\Gamma_1 := \sqrt{\rho_n}\,\|O^\top U\|_F
    +\sqrt{\rho_n\log n}\,\|O^\top U\|_{2\to2}+b_\infty\log n$.
    \item \emph{Estimated leverage and residuals}:
    \[
        \widehat b:=\|O^\top\widehat U\|_{2\to\infty},\qquad
        \widehat h:=\|\widehat V\|_{2\to\infty},\qquad
        b_\infty^{(i)}:=\|O^\top\widehat U^{(i)}\|_{2\to\infty},\qquad
        h^{(i)}:=\|\widehat V^{(i)}\|_{2\to\infty},
    \]
    and $\Psi_i:=\|\Delta_i^\top\widehat U\|_F+\|\Delta_i\widehat V\|_F$,
    $\Psi:=\max_i\Psi_i$, for the leave-one-out perturbation $\Delta_i$ of
    \Cref{loo:structure} below.
\end{itemize}

The Procrustean reduction below shows the entire entrywise error is governed by
$\Gamma = \| N^\top \widehat U \|_{2\to\infty}$. Reading it row by row,
\begin{equation}
    \label{eq:loo-row}
    \big( N^\top \widehat U \big)_{i \cdot}^\top
    = \widehat U^\top N_{\cdot i}
    = \widehat U^\top (O E_{\cdot i}-D_{\cdot i}),
    \qquad i \in [n],
\end{equation}
the $i$-th row pairs $\widehat U$ with the aggregated noise carried by node $i$'s
own edges. The obstruction is a dependence: $\widehat U$ is a global function of the
entire matrix $A$, hence of $E_{\cdot i}$. Estimating the offending component by operator norms costs an
additive $\| N \|_{2\to2}^2/\sigma_r \lesssim \varepsilon_{\rm SNR}^2 \sigma_r$, which is higher than the true bound we can achieve. The leave-one-out
argument chain we employ restores some degree of independence we can use to achieve
tighter bounds. Splitting along \eqref{eq:loo-defs}, we write
\begin{equation}
    \label{eq:loo-split}
    \widehat U^\top N_{\cdot i}
    = \underbrace{\widehat U^{(i)\top} N_{\cdot i}}_{\textbf{(A) independent}}
    + \underbrace{\big(\widehat U - \widehat U^{(i)}\big)^\top
    N_{\cdot i}}_{\textbf{(B) coupling}}.
\end{equation}
For term (A), independence of $\widehat U^{(i)}$ and $E_{\cdot i}$
allows us to apply Bernstein's inequality conditional on
$\widehat U^{(i)}$, leaving the bias term
$-\widehat U^{(i)\top}D_{\cdot i}$.
For term (B), we use that $R-R^{(i)}$ is a rank-two perturbation arising
from the edges incident to node $i$. This lets us control the product
of $\widehat U-\widehat U^{(i)}$ with $N_{\cdot i}$ using the column norm
$b_O$ rather than the operator norm $\|O\|_{2\to2}$. The residuals of this
perturbation also contain $\|O_{\cdot i}\|_2\|\widehat V^\top E_{\cdot i}\|_2$,
node $i$'s own column paired with the right singular vectors, and we control
$\widehat V^\top E_{\cdot i}$ by the same splitting. The argument therefore
bounds $\Gamma$ and $\Gamma_V$ jointly, with the leverage of
$\widehat U^{(i)}$ bounded through that of $\widehat U$
(\Cref{loo:stability-full}) and the latter through \Cref{lem:master-left}.
This is an adaptation of the leave-one-out arguments of
\cite{abbe_entrywise_2020,lei_unified_2019} to the aggregated noise $N=OE-D$; the adaptation is mostly bookkeeping, but it surfaces two new design
quantities---the column norm $b_O$ and the Gram overlap $g_1(O)$---which we
track explicitly rather than hide in the constants.

\begin{lemma}[Procrustean reduction]
    \label{lem:master}
    Suppose $\| N \|_{2 \to 2} \leq \sigma_r/4$, and let $Q$ be the orthogonal
    Procrustes alignment of \Cref{lem:wedin}. Then
    \begin{equation}
        \label{eq:master}
        \| \widehat V Q - V \|_{2 \to \infty}
        \;\leq\; \frac{4}{3\sigma_r}\Big(
        \| N^\top \widehat U \|_{2 \to \infty}
        + \| V \|_{2 \to \infty}\, \| N \|_{2 \to 2} \Big)
        + \| V \|_{2 \to \infty}\,
        \| \sin\Theta(\widehat V, V) \|_{2 \to 2}^2 .
    \end{equation}
\end{lemma}

\begin{proof}
    Decompose $\widehat V Q - V = (I - VV^\top)\widehat V Q + V(V^\top \widehat V Q -
    I)$ and bound the two pieces; the identity is exact because $(I - VV^\top)V = 0$
    and $V^\top V = I$.

    \emph{Perpendicular part.} From $R^\top \widehat U = \widehat V \widehat \Sigma$ we
    have $\widehat V = R^\top \widehat U \widehat\Sigma^{-1} = (V\Sigma U^\top +
    N^\top)\widehat U \widehat\Sigma^{-1}$, and since $(I - VV^\top)V = 0$,
    \[
        (I - VV^\top)\widehat V = (I - VV^\top) N^\top \widehat U \widehat\Sigma^{-1}.
    \]
    Weyl's inequality gives $\widehat\sigma_r\geq\sigma_r-\|N\|_{2\to2}
    \geq3\sigma_r/4$. Therefore
    $\| \widehat\Sigma^{-1} \|_{2 \to 2} = \widehat\sigma_r^{-1} \leq
    (\sigma_r - \| N \|_{2 \to 2})^{-1} \leq \tfrac{4}{3}\sigma_r^{-1}$.
    Using this and $\|(I -
    VV^\top) X\|_{2\to\infty} \leq \|X\|_{2\to\infty} + \|V\|_{2\to\infty}\|V^\top
    X\|_{2\to2}$ with $X = N^\top\widehat U$,
    \[
        \| (I - VV^\top)\widehat V \|_{2 \to \infty}
        \leq \tfrac{4}{3}\sigma_r^{-1}\Big( \| N^\top \widehat U \|_{2 \to \infty}
        + \| V \|_{2 \to \infty}\, \| V^\top N^\top \widehat U \|_{2 \to 2} \Big)
        \leq \tfrac{4}{3}\sigma_r^{-1}\Big( \| N^\top \widehat U \|_{2 \to \infty}
        + \| V \|_{2 \to \infty}\, \| N \|_{2 \to 2} \Big).
    \]
    where the last step uses $\| V^\top N^\top \widehat U \|_{2 \to 2} \leq \| V \|_{2
    \to 2}\| N \|_{2 \to 2}\| \widehat U \|_{2 \to 2} = \| N \|_{2 \to 2}$. The same
    bound holds for $\| (I - VV^\top)\widehat V Q \|_{2 \to \infty}$, since right
    multiplication by the orthogonal $Q$ leaves each row's $\ell_2$ norm unchanged.

    \emph{Aligned part.} $\| V(V^\top \widehat V Q - I) \|_{2 \to \infty} \leq \| V \|_{2
    \to \infty}\| V^\top \widehat V Q - I \|_{2 \to 2}$. Writing the singular value
    decomposition $V^\top \widehat V = Q_1 (\cos\Theta) Q_2^\top$, the Procrustes
    alignment of \Cref{lem:wedin}---the minimizer of $\| \widehat V Q - V \|$ over
    orthogonal $Q$---is $Q = Q_2 Q_1^\top$, so $V^\top \widehat V Q = Q_1 (\cos\Theta)
    Q_1^\top$ and $\| V^\top \widehat V Q - I \|_{2 \to 2} = \| I - \cos\Theta \|_{2
    \to 2} = \max_k(1 - \cos\theta_k) \leq \max_k(1 - \cos^2\theta_k) = \|
    \sin\Theta(\widehat V, V) \|_{2 \to 2}^2$. Summing the two bounds gives
    \eqref{eq:master}.
\end{proof}

The same reduction, applied to the left factor and measured in design
coordinates, controls the leverage $\widehat b$ of $\widehat U$.

\begin{lemma}[Procrustean reduction for the left factor, in design coordinates]
    \label{lem:master-left}
    Suppose $\|N\|_{2\to2}\leq\sigma_r/4$, and let $Q_U$ be the orthogonal
    Procrustes alignment of $\widehat U$ to $U$. Then
    \begin{equation}
        \label{eq:master-left}
        \|O^\top(\widehat UQ_U-U)\|_{2\to\infty}
        \leq\frac{4}{3\sigma_r}\Big(\|O^\top N\widehat V\|_{2\to\infty}
            +b_\infty\|N\|_{2\to2}\Big)
        +b_\infty\|\sin\Theta(U,\widehat U)\|_{2\to2}^2,
    \end{equation}
    and
    \begin{equation}
        \label{eq:master-left-cross}
        \|O^\top N\widehat V\|_{2\to\infty}
        \leq\min\big\{b_O\|N\|_{2\to2},\ g_1\Gamma_V+q_1\big\}.
    \end{equation}
\end{lemma}

\begin{proof}
    Write $\widehat UQ_U-U=(I-UU^\top)\widehat UQ_U+U(U^\top\widehat UQ_U-I)$.
    From $R\widehat V=\widehat U\widehat\Sigma$,
    \[
        \widehat U=(M+N)\widehat V\widehat\Sigma^{-1}
        =U\Sigma V^\top\widehat V\widehat\Sigma^{-1}+N\widehat V\widehat\Sigma^{-1},
        \qquad\text{so}\qquad
        (I-UU^\top)\widehat U=(I-UU^\top)N\widehat V\widehat\Sigma^{-1},
    \]
    since the first term is annihilated by $I-UU^\top$. As in the proof of
    \Cref{lem:master}, Weyl's inequality gives
    $\|\widehat\Sigma^{-1}\|_{2\to2}\leq4/(3\sigma_r)$. For any $X$,
    \[
        \|O^\top(I-UU^\top)X\|_{2\to\infty}
        \leq\|O^\top X\|_{2\to\infty}
        +\|O^\top U\|_{2\to\infty}\|U^\top X\|_{2\to2};
    \]
    with $X=N\widehat V\widehat\Sigma^{-1}$ this is at most
    $\frac{4}{3\sigma_r}(\|O^\top N\widehat V\|_{2\to\infty}+b_\infty\|N\|_{2\to2})$,
    and right multiplication by $Q_U$ does not change row norms. The aligned
    part is bounded by
    $b_\infty\|U^\top\widehat UQ_U-I\|_{2\to2}\leq
    b_\infty\|\sin\Theta(U,\widehat U)\|_{2\to2}^2$, exactly as in the proof of
    \Cref{lem:master}. This proves \eqref{eq:master-left}.

    For \eqref{eq:master-left-cross}, the first bound is
    $\|O^\top N\widehat V\|_{2\to\infty}
    \leq\|O^\top\|_{2\to\infty}\|N\widehat V\|_{2\to2}\leq b_O\|N\|_{2\to2}$,
    by \eqref{eq:norm-relations}. For the second, $O^\top N=O^\top OE-O^\top D$.
    Row $j$ of $O^\top OE\widehat V$ is
    $\sum_k(O^\top O)_{jk}(E\widehat V)_{k\cdot}$, whose norm is at most
    $\sum_k|(O^\top O)_{jk}|\,\|(E\widehat V)_{k\cdot}\|_2\leq g_1\Gamma_V$.
    Row $j$ of $O^\top D\widehat V$ is
    $(O^\top O_{\cdot j})^\top\diag(P^0)\widehat V$, whose norm is at most
    $\|O^\top O_{\cdot j}\|_2\,\rho_n\|\bar P\|_\infty
    \leq ob_O\rho_n\|\bar P\|_\infty=q_1$.
\end{proof}

On $\mcE$ with $\varepsilon\leq1/(8C_0)$, \eqref{eq:total-noise-event} and
\Cref{lem:wedin} give $\|N\|_{2\to2}\leq C_0\varepsilon\sigma_r\leq\sigma_r/8$
and $\|\sin\Theta(U,\widehat U)\|_{2\to2}\leq2C_0\varepsilon\leq1/4$. Then
\Cref{lem:master-left} gives the two usable forms
\begin{equation}
    \label{eq:loo-leverage}
    \widehat b\leq b_\infty+\|O^\top(\widehat UQ_U-U)\|_{2\to\infty}
    \leq2b_\infty+\frac43\min\Big\{C_0b_O\varepsilon,\
        \frac{g_1\Gamma_V+q_1}{\sigma_r}\Big\},
\end{equation}
and \Cref{lem:master}, with $\widehat h=\|\widehat VQ\|_{2\to\infty}$, gives
\begin{equation}
    \label{eq:loo-incoherence}
    \widehat h\leq h+\|\widehat VQ-V\|_{2\to\infty}
    \leq\frac{4\Gamma}{3\sigma_r}+3h.
\end{equation}

The next lemma packages the structure of the leave-one-out perturbation: it is
rank two, supported on node $i$'s edges, and independent of the leave-one-out
graph.

\begin{lemma}[Structure of the coupling]
    \label{loo:structure}
    For every $i \in [n]$,
    \begin{equation}
        \label{eq:loo-rank2}
        \Delta_i := R - R^{(i)} = O\big( E - E^{(i)} \big)
        = (O E_{\cdot i})\, e_i^\top + O_{\cdot i}\, E_{\cdot i}^\top ,
    \end{equation}
    a matrix of rank at most two. Moreover $E^{(i)}$ is a function of $\{ A_{jk} : j,
    k \neq i \}$ alone, hence independent of $E_{\cdot i} = \{ E_{ij} \}_{j \neq i}$;
    in particular $\widehat U^{(i)}$ and $\widehat V^{(i)}$ are independent of
    $N_{\cdot i} = O E_{\cdot i}-D_{\cdot i}$.
\end{lemma}

\begin{proof}
    Since $E_{ii} = 0$, expanding \eqref{eq:loo-defs} gives $E - E^{(i)} = E_{\cdot
    i} e_i^\top + e_i E_{\cdot i}^\top$ (the $(i,i)$ entry, counted twice, vanishes),
    and multiplying by $O$ on the left yields \eqref{eq:loo-rank2} because $O e_i =
    O_{\cdot i}$ and, by symmetry of $E$, the $i$-th row of $E$ equals $E_{\cdot
    i}^\top$. Independence holds because the entries $\{ E_{jk} : j, k \neq i \}$
    and $\{ E_{ij} : j \neq i \}$ are disjoint collections of independent variables,
    and $R^{(i)} = M + OE^{(i)}-D$ is a deterministic function of the former.
\end{proof}

\begin{lemma}[Leave-one-out stability]
    \label{loo:stability-full}
    On the event $\mcE$ of \Cref{loo:event}, suppose $\varepsilon \leq 1/(8 C_0 \vee
    8 C_{10})$. Then for every $i \in [n]$ simultaneously there are orthogonal
    $Q_i, Q_i' \in \mathbb{R}^{r \times r}$ with
    \begin{equation}
        \label{eq:loo-stability-refined}
        \| \widehat U - \widehat U^{(i)} Q_i \|_F
        \vee \| \widehat V - \widehat V^{(i)} Q_i' \|_F
        \leq \frac{2\sqrt2}{\sigma_r}\Psi_i,
        \qquad
        \Psi_i \leq \Gamma + q_D + b_O\Gamma_V
        + C_{10}\varepsilon\sigma_r\Big(\frac{\widehat b}{o}+\widehat h\Big).
    \end{equation}
    Consequently, for every $i\in[n]$,
    \begin{equation}
        \label{eq:loo-stability-rows}
        b_\infty^{(i)}\leq\widehat b+\frac{2\sqrt2\,b_O}{\sigma_r}\Psi,
        \qquad
        h^{(i)}\leq\widehat h+\frac{2\sqrt2}{\sigma_r}\Psi.
    \end{equation}
\end{lemma}

\begin{proof}
    We apply \Cref{loo:wedin-res} with $R^{(i)}$ and $\Delta_i$ from
    \Cref{loo:structure} in the roles of $M$ and $N$, so that $M + N = R$.

    \emph{Gap.} By Weyl's inequality and $\rk M=r$,
    \[
        \sigma_r(R)\geq\sigma_r-\|N\|_{2\to2},
        \qquad
        \sigma_{r+1}(R^{(i)})\leq\sigma_{r+1}(M)+\|N^{(i)}\|_{2\to2}
        =\|N^{(i)}\|_{2\to2}.
    \]
    By \eqref{eq:total-noise-event}, on $\mcE$ both total noise norms are at most
    $C_0\|O\|_{2\to2}\sqrt{n\rho_n}=C_0\varepsilon\sigma_r\leq\sigma_r/8$.
    Subtracting the second Weyl bound from the first therefore gives
    \[
        \gamma = \sigma_r(R) - \sigma_{r+1}(R^{(i)})
        \geq \sigma_r - 2 C_0 \varepsilon\, \sigma_r \geq \tfrac{3}{4}\sigma_r .
    \]
    \Cref{loo:wedin-res} then bounds both Procrustes distances by
    $\sqrt2\,\Psi_i/\gamma\leq2\sqrt2\,\Psi_i/\sigma_r$.

    \emph{Residuals.} Using \eqref{eq:loo-rank2} and $\| x y^\top \|_F = \| x \|_2
    \| y \|_2$,
    \begin{align*}
        \| \Delta_i^\top \widehat U \|_F
        &= \big\| e_i (O E_{\cdot i})^\top \widehat U
            + E_{\cdot i} (O_{\cdot i}^\top \widehat U) \big\|_F
        \leq \| \widehat U^\top N_{\cdot i} \|_2 +q_D
            + \| E_{\cdot i} \|_2\, \| \widehat U^\top O_{\cdot i} \|_2, \\
        \| \Delta_i \widehat V \|_F
        &= \big\| (O E_{\cdot i}) (e_i^\top \widehat V)
            + O_{\cdot i} (E_{\cdot i}^\top \widehat V) \big\|_F
        \leq \| OE_{\cdot i} \|_2\, \| \widehat V_{i \cdot} \|_2
            + \| O_{\cdot i} \|_2\, \| \widehat V^\top E_{\cdot i} \|_2 .
    \end{align*}
    Here $OE_{\cdot i}=N_{\cdot i}+D_{\cdot i}$ accounts for $q_D$.
    Bound $\| \widehat U^\top N_{\cdot i} \|_2\leq\Gamma$,
    $\| \widehat U^\top O_{\cdot i} \|_2 \leq \widehat b$,
    $\| \widehat V_{i \cdot} \|_2 \leq \widehat h$ and
    $\| \widehat V^\top E_{\cdot i} \|_2 \leq \Gamma_V$. Inserting $\|E_{\cdot i}\|_2\leq C_{10}\sqrt{n\rho_n}
    =C_{10}\varepsilon\sigma_r/o$, $\|OE_{\cdot i}\|_2\leq C_{10}\varepsilon\sigma_r$
    and $\|O_{\cdot i}\|_2\leq b_O$ from \eqref{eq:loo-event}--\eqref{eq:loo-event-consequences}
    gives the bound on $\Psi_i$ in \eqref{eq:loo-stability-refined}. All bounds are
    uniform in $i$ because $\mcE$ controls every column of $E$ simultaneously.

    \emph{Rows.} For \eqref{eq:loo-stability-rows}, \eqref{eq:norm-relations} gives
    \[
        b_\infty^{(i)}=\|O^\top\widehat U^{(i)}Q_i\|_{2\to\infty}
        \leq\widehat b+\|O^\top\|_{2\to\infty}\|\widehat U^{(i)}Q_i-\widehat U\|_{2\to2}
        \leq\widehat b+b_O\|\widehat U-\widehat U^{(i)}Q_i\|_F,
    \]
    and $h^{(i)}\leq\widehat h+\|\widehat V-\widehat V^{(i)}Q_i'\|_{2\to\infty}
    \leq\widehat h+\|\widehat V-\widehat V^{(i)}Q_i'\|_F$. Taking
    $\Psi_i\leq\Psi$ in \eqref{eq:loo-stability-refined} concludes.
\end{proof}

\begin{lemma}[Sharp row concentration]
    \label{loo:concentration-full}
    There are constants $c_\star \in (0,1)$ and $C_{11}$, depending only on $\| \bar P
    \|_\infty$ and $\alpha$, such that the following holds. Suppose
    \begin{enumerate}
        \item[\textup{(H)}] $n\rho_n \geq C_0 \log n$, $\varepsilon \leq c_\star$
        and $\lambda\leq c_\star$.
    \end{enumerate}
    Then with probability at least $1 - 4 n^{-\alpha}$,
    \begin{equation}
        \label{eq:loo-gamma-final}
        S=\Gamma+b_O\Gamma_V
        \leq C_{11} \Big[
        \gamma_1(o+b_O)+(b_\infty+b_Oh+\vartheta)\log n
        +\varepsilon^2\sigma_r\Big(h+\frac{b_\infty+\vartheta}{o}\Big)
        +q_D\Big(1+\frac{o\log n}{\sigma_r}\Big)\Big]
    \end{equation}
    with $\vartheta=b_O\varepsilon$. If in addition
    \begin{enumerate}
        \item[\textup{(H$'$)}] $\varsigma\varepsilon\leq c_\star$ and
        $\lambda_\varsigma\leq c_\star$,
    \end{enumerate}
    then \eqref{eq:loo-gamma-final} also holds with $\vartheta=0$.
\end{lemma}

\begin{proof}
    The proof has three steps.
    \begin{enumerate}
        \item[]\textbf{Step 1.} A conditional Bernstein bound for the independent
        terms $\widehat U^{(i)\top}N_{\cdot i}$ and $\widehat V^{(i)\top}E_{\cdot i}$,
        uniform in $i$, on a new event $\mcA$ of probability at least
        $1 - 2n^{-(\alpha+1)}$.
        \item[]\textbf{Step 2.} The coupling terms: splitting as in
        \eqref{eq:loo-split} and inserting \Cref{loo:stability-full} gives a
        self-bounding inequality for $S$ in terms of $\widehat b$ and $\widehat h$.
        \item[]\textbf{Step 3.} Substituting
        \eqref{eq:loo-leverage}--\eqref{eq:loo-incoherence} for $\widehat b$ and
        $\widehat h$ and absorbing: the first branch of \eqref{eq:loo-leverage}
        under (H) gives $\vartheta=b_O\varepsilon$ (Step~3a), and the second
        branch under (H$'$) gives $\vartheta=0$ (Step~3b).
    \end{enumerate}
    Assume $c_\star \leq 1/(8 C_0 \vee 8 C_{10})$. On $\mcE$,
    \Cref{loo:event,loo:stability-full,lem:master-left},
    \eqref{eq:loo-leverage}--\eqref{eq:loo-incoherence} and the Wedin bound
    \Cref{lem:wedin} (which needs $\| N \|_{2\to2} \leq \sigma_r/2$) all apply;
    $c_\star$ shrinks once more in Step~3.

    \textbf{Step 1: the independent terms.}
    Fix $i$ and condition on $\mcF_i := \sigma(A_{jk} : j, k \neq i)$. By
    \Cref{loo:structure}, $\widehat U^{(i)}$ and $\widehat V^{(i)}$ are
    $\mcF_i$-measurable while $E_{\cdot i}$ is independent of $\mcF_i$, so,
    conditionally,
    \[
        \widehat U^{(i)\top} N_{\cdot i}+\widehat U^{(i)\top}D_{\cdot i}
        = \sum_{j \neq i} E_{ij}\; \widehat U^{(i)\top} O_{\cdot j},
        \qquad
        \widehat V^{(i)\top}E_{\cdot i}
        = \sum_{j \neq i} E_{ij}\; \big(\widehat V^{(i)}_{j\cdot}\big)^\top
    \]
    are sums of independent, mean-zero vectors in $\mathbb{R}^r$ with
    $\mcF_i$-measurable coefficients, and the $E_{ij}$ have variances at most
    $\rho_n\|\bar P\|_\infty$. For the first sum the coefficient matrix
    $W_i=O^\top\widehat U^{(i)}$ has $\|W_i\|_F\leq\sqrt r\,o$,
    $\|W_i\|_{2\to2}\leq o$ and $\|W_i\|_{2\to\infty}=b_\infty^{(i)}$; for the
    second, $W_i'=\widehat V^{(i)}$ has $\|W_i'\|_F=\sqrt r$,
    $\|W_i'\|_{2\to2}=1$ and $\|W_i'\|_{2\to\infty}=h^{(i)}$.
    Apply \Cref{lem:vector-bernstein}(ii) conditionally on $\mcF_i$ to each
    sum, with $\epsilon=n^{-(\alpha+2)}$. There is
    $C_1=C_1(\|\bar P\|_\infty,\alpha)$ such that the event $\mcA_i$ on which
    \begin{equation}
        \label{eq:loo-A-final}
        \big\|\widehat U^{(i)\top}N_{\cdot i}\big\|_2
        \leq C_1\big(\gamma_1o+b_\infty^{(i)}\log n\big)+q_D,
        \qquad
        \big\|\widehat V^{(i)\top}E_{\cdot i}\big\|_2
        \leq C_1\big(\gamma_1+h^{(i)}\log n\big)
    \end{equation}
    satisfies $\bbp{\mcA_i^c\mid\mcF_i}\leq2n^{-(\alpha+2)}$ almost surely,
    without restricting to $\mcE$; here $q_D$ bounds
    $\|\widehat U^{(i)\top}D_{\cdot i}\|_2\leq\|D_{\cdot i}\|_2$. The thresholds
    are $\mcF_i$-measurable, so taking expectations removes the conditioning,
    $\bbp{\mcA_i^c}=\bbe{\bbp{\mcA_i^c\mid\mcF_i}}\leq2n^{-(\alpha+2)}$, and
    $\mcA:=\bigcap_{i=1}^n\mcA_i$ satisfies $\bbp{\mcA^c}\leq2n^{-(\alpha+1)}$.

    \textbf{Step 2: the coupling terms.}
    On $\mcE\cap\mcA$, fix $i$ and split both rows as in \eqref{eq:loo-split},
    aligned by the $Q_i,Q_i'$ of \Cref{loo:stability-full}:
    \[
        \| \widehat U^\top N_{\cdot i} \|_2
        \leq \big\| \widehat U^{(i)\top} N_{\cdot i} \big\|_2
        + \| \widehat U - \widehat U^{(i)} Q_i \|_F\, \| N_{\cdot i} \|_2,
        \qquad
        \| \widehat V^\top E_{\cdot i} \|_2
        \leq \big\| \widehat V^{(i)\top} E_{\cdot i} \big\|_2
        + \| \widehat V - \widehat V^{(i)} Q_i' \|_F\, \| E_{\cdot i} \|_2,
    \]
    the first summands being unchanged by the orthogonal $Q_i,Q_i'$. Insert
    \eqref{eq:loo-A-final}, \eqref{eq:loo-stability-refined}, the bounds
    $\|N_{\cdot i}\|_2\leq C_{10}\varepsilon\sigma_r$ and
    $\|E_{\cdot i}\|_2\leq C_{10}\varepsilon\sigma_r/o$ of
    \eqref{eq:loo-event} and \eqref{eq:total-noise-event}, then
    \eqref{eq:loo-stability-rows} and $\lambda=b_O\log n/\sigma_r$. Taking the
    maximum over $i\in[n]$,
    \[
    \begin{aligned}
        \Gamma&\leq C_1\gamma_1o+C_1\widehat b\log n+q_D
        +\big(2\sqrt2C_1\lambda+2\sqrt2C_{10}\varepsilon\big)\Psi,\\
        b_O\Gamma_V&\leq C_1\gamma_1b_O+C_1b_O\widehat h\log n
        +\big(2\sqrt2C_1\lambda+2\sqrt2C_{10}\varepsilon\,b_O/o\big)\Psi.
    \end{aligned}
    \]
    Adding, and using $b_O\leq o$ and
    $\Psi\leq S+q_D+C_{10}\varepsilon\sigma_r(\widehat b/o+\widehat h)$ from
    \eqref{eq:loo-stability-refined},
    \begin{equation}
        \label{eq:loo-selfbound}
        S\leq C_1\gamma_1(o+b_O)+C_1\log n\,(\widehat b+b_O\widehat h)+q_D
        +\kappa\Big[S+q_D+C_{10}\varepsilon\sigma_r\Big(\frac{\widehat b}{o}
        +\widehat h\Big)\Big],
        \qquad
        \kappa:=4\sqrt2\,(C_1\lambda+C_{10}\varepsilon).
    \end{equation}
    This is the self-bounding inequality: it remains to substitute
    \eqref{eq:loo-leverage}--\eqref{eq:loo-incoherence} for $\widehat b$ and
    $\widehat h$ and absorb. Since
    $S\leq\|N\|_{2\to2}+b_O\|E\|_{2\to2}<\infty$, a term $cS$ with $c\leq1/2$
    on the right may be subtracted from both sides.

    \textbf{Step 3a: the branch $\vartheta=b_O\varepsilon$.}
    Use the first branch of \eqref{eq:loo-leverage},
    $\widehat b\leq\bar b:=2b_\infty+\tfrac43C_0b_O\varepsilon$, and
    \eqref{eq:loo-incoherence} with $\Gamma\leq S$. Then
    \[
        \widehat b+b_O\widehat h\leq\bar b+3b_Oh+\frac{4b_O}{3\sigma_r}S,
        \qquad
        \frac{\widehat b}{o}+\widehat h\leq\frac{\bar b}{o}+3h+\frac{4}{3\sigma_r}S,
    \]
    and, since $\kappa\leq1$ for $c_\star$ small, \eqref{eq:loo-selfbound} becomes
    \[
        S\leq C_1\gamma_1(o+b_O)+C_1\log n\,(\bar b+3b_Oh)+2q_D
        +\kappa C_{10}\varepsilon\sigma_r\Big(\frac{\bar b}{o}+3h\Big)
        +\Big[\tfrac43C_1\lambda+\kappa\big(1+\tfrac43C_{10}\varepsilon\big)\Big]S.
    \]
    For $c_\star$ small enough, depending only on $C_1$ and $C_{10}$, (H) makes
    the bracket at most $\tfrac12$; subtract and double. Finally
    $\kappa C_{10}\varepsilon\sigma_r
    =4\sqrt2C_1C_{10}\,b_O\varepsilon\log n+4\sqrt2C_{10}^2\varepsilon^2\sigma_r$
    and $\bar b/o+3h\leq6$, so the term
    $\kappa C_{10}\varepsilon\sigma_r(\bar b/o+3h)$ is at most a constant times
    $b_O\varepsilon\log n+\varepsilon^2\sigma_r\big(h+(b_\infty+b_O\varepsilon)/o\big)$.
    Collecting, and using $C_1\bar b\log n\leq C(b_\infty+b_O\varepsilon)\log n$,
    \[
        S\leq C\Big[\gamma_1(o+b_O)+(b_\infty+b_Oh+b_O\varepsilon)\log n
        +\varepsilon^2\sigma_r\Big(h+\frac{b_\infty+b_O\varepsilon}{o}\Big)+q_D\Big],
    \]
    which is \eqref{eq:loo-gamma-final} with $\vartheta=b_O\varepsilon$.

    \textbf{Step 3b: the branch $\vartheta=0$.}
    Under (H$'$), use the second branch of \eqref{eq:loo-leverage},
    together with $\Gamma_V\leq S/b_O$ and $g_1/b_O=\varsigma o$:
    \[
        \widehat b\leq2b_\infty+\frac{4\varsigma o}{3\sigma_r}S+\frac{4q_1}{3\sigma_r}.
    \]
    Now
    \[
        \widehat b+b_O\widehat h\leq2b_\infty+3b_Oh+\frac{4q_1}{3\sigma_r}
        +\frac{4(\varsigma o+b_O)}{3\sigma_r}S,
        \qquad
        \frac{\widehat b}{o}+\widehat h\leq\frac{2b_\infty}{o}+3h
        +\frac{4q_1}{3o\sigma_r}+\frac{4(\varsigma+1)}{3\sigma_r}S,
    \]
    so the coefficient of $S$ in \eqref{eq:loo-selfbound} is
    \[
        \tfrac43C_1(\lambda_\varsigma+\lambda)
        +\kappa\Big[1+\tfrac43C_{10}(\varsigma\varepsilon+\varepsilon)\Big]
        \leq\tfrac12
    \]
    for $c_\star$ small, by (H), (H$'$) and $\lambda\leq\lambda_\varsigma$.
    Subtracting and doubling,
    \[
        S\leq2\Big[C_1\gamma_1(o+b_O)
        +C_1\log n\Big(2b_\infty+3b_Oh+\frac{4q_1}{3\sigma_r}\Big)+2q_D
        +\kappa C_{10}\varepsilon\sigma_r\Big(\frac{2b_\infty}{o}+3h
        +\frac{4q_1}{3o\sigma_r}\Big)\Big].
    \]
    In the last term, the $\lambda$-part of $\kappa$ contributes at most a
    constant times
    $b_O\varepsilon\log n\,(b_\infty/o+h)+\varepsilon\lambda q_D
    \leq\log n\,(b_\infty+b_Oh)+q_D$, because $b_O\varepsilon/o\leq1$ and
    $\varepsilon,\lambda\leq1$; the $\varepsilon$-part contributes at most a
    constant times
    $\varepsilon^2\sigma_r(h+b_\infty/o)+\varepsilon^2q_1/o
    \leq\varepsilon^2\sigma_r(h+b_\infty/o)+q_D$. Since
    $q_1\log n/\sigma_r=q_D\,o\log n/\sigma_r$, this is
    \eqref{eq:loo-gamma-final} with $\vartheta=0$.

    A union bound, requiring no independence between $\mcE$ and $\mcA$, gives
    failure probability at most $\bbp{\mcE^c}+\bbp{\mcA^c}\leq2n^{-\alpha}
    +2n^{-(\alpha+1)}\leq4n^{-\alpha}$.
\end{proof}

Feeding \Cref{loo:concentration-full} into the Procrustean reduction
\Cref{lem:master} yields the bound that the master theorems consume.

\begin{theorem}[Sharp $\ell_{2\to\infty}$ bound for the ARD embedding]
    \label{loo:twoinfty-sharp}
    Under hypothesis (H) of \Cref{loo:concentration-full}, with probability at
    least $1 - 4n^{-\alpha}$ there is an orthogonal $Q \in \mathbb{R}^{r \times r}$
    with
    \begin{equation}
        \label{eq:loo-twoinfty-sharp}
        \| \widehat V Q - V \|_{2 \to \infty}
        \;\leq\; C_{13}\bigg[
        \frac{\gamma_1(o+b_O)+(b_\infty+b_Oh+\vartheta)\log n}{\sigma_r}
        +\varepsilon h+\varepsilon^2\frac{b_\infty+\vartheta}{o}
        +\frac{q_D}{\sigma_r}\Big(1+\frac{o\log n}{\sigma_r}\Big)\bigg],
    \end{equation}
    with $\vartheta=b_O\varepsilon$; under (H$'$) also with $\vartheta=0$. Here
    $C_{13}$ depends only on $\| \bar P \|_\infty$ and $\alpha$.
\end{theorem}

\begin{proof}
    Work on the event of \Cref{loo:concentration-full}, which is contained in $\mcE$. By
    \Cref{lem:master} (valid since $\| N \|_{2\to2} \leq C_0\varepsilon\sigma_r
    \leq \sigma_r/8$) and \Cref{lem:wedin},
    \[
        \| \widehat V Q - V \|_{2 \to \infty}
        \leq \frac{4}{3\sigma_r}\Big( \Gamma
        + h\, \| N \|_{2\to2} \Big)
        + h\,
        \frac{4 \| N \|_{2\to2}^2}{\sigma_r^2}
        \leq \frac{4\Gamma}{3\sigma_r}
        + C h\, \varepsilon.
    \]
    Insert $\Gamma\leq S$ and \eqref{eq:loo-gamma-final}, and divide each summand
    by $\sigma_r$, using $\varepsilon^2\sigma_r/\sigma_r=\varepsilon^2$ and
    $\varepsilon^2h\leq\varepsilon h$.
\end{proof}

\begin{remark}[Relation to the entrywise eigenvector literature]
    \label{loo:rem-literature}
    This remark places the proof relative to
    \cite{abbe_entrywise_2020,lei_unified_2019,zhang_leave-one-out_2024}. The first
    two works analyse eigenvectors of a
    symmetric random matrix, comparing the full and leave-one-out eigenvectors
    through the Davis--Kahan theorem and closing the argument with a
    self-bounding step. Here the argument runs on both singular subspaces of the
    rectangular matrix $R$: the aggregation $O$ is deterministic, so it commutes with
    the conditioning on $\mcF_i$, and the residual form of Wedin
    (\Cref{loo:wedin-res}) controls $\widehat U-\widehat U^{(i)}$ and
    $\widehat V-\widehat V^{(i)}$ simultaneously, producing the self-bounding
    coefficient $\kappa$ of \eqref{eq:loo-selfbound}, of order
    $\varepsilon_{\rm SNR}+\lambda$. At $O=I_n$, $R=A$ is
    symmetric, so $\widehat V=\widehat US_\pm$ with $S_\pm$ the diagonal matrix of
    signs of the $r$ eigenvalues of largest absolute value; hence $\Gamma_V\leq\Gamma+q_D$ and
    $\widehat b=\widehat h$, and the bounds for the two singular subspaces
    coincide. Zhang and
    Zhou's leave-one-column-out singular-subspace theorem treats a one-column
    augmentation; it does not directly cover the rank-two update
    $OE_{\cdot i}e_i^\top+O_{\cdot i}E_{\cdot i}^\top$ here. The only new
    bookkeeping relative to those works is that every noise functional is of the
    mixed form $OE$, which is why the design functionals $\| O \|_{2\to2}$, $b_O$,
    $b_\infty$, $\| O^\top U \|_F$ and $g_1(O)$---and no other property of
    $O$---appear in the final bounds.
\end{remark}

\section{Proofs of the main guarantees for the SBM}
\label{app:sbm-proofs}
This appendix proves the SBM results of \Cref{sec:recovery}:
\Cref{lem:pop-embedding}(i), the master guarantee \Cref{thm:sbm-main}, the regime
corollary \Cref{cor:sbm-regimes}, and the weak-recovery bound of
Lemma~\ref{lem:weak-recovery-precise}. %% updated this
% \Cref{rem:weak-recovery}. 
We prove the master theorem as follows:
\begin{enumerate}
    \item We show that the right singular vectors of $M=OP^0$ have identical
    rows for vertices in the same community. Distinct profiles give $\delta>0$;
    at full rank $\delta\geq\sqrt{2/n_{\max}}$ (\Cref{lem:pop-embedding}).
    We write $R=M+\mcN$, where $\mcN=OE-D$ and $D$ is bounded in \Cref{lem:diag}.
    \item We use \Cref{loo:twoinfty-sharp} to bound
    $\|\widehat VQ-V_M\|_{2\to\infty}$ by $\eta_\infty$ in
    \eqref{eq:eta-infty}, for an orthogonal matrix $Q$.
    \item We apply \Cref{lem:single-linkage} to recover all community labels
    using single linkage when $\eta_\infty<\delta/4$.
    For approximate $k$-means, \Cref{lem:kmeans}(ii) gives exact recovery
    after nearest-centroid reassignment under \eqref{eq:master-rounding}.
\end{enumerate}
The regime corollary then substitutes the signal bound (\Cref{lem:signal}) and the
balance and leverage hypotheses into the three summands of $\eta_\infty$. The same
three steps, with the substitutions recorded there, prove the DCSBM results in
\Cref{app:dcsbm-proofs} and the stacked results in \Cref{app:fanout-proofs}.

\subsection{Population structure}

\begin{lemma}[Population embedding]
    \label{lem:pop-embedding}
    Suppose 
    we observe data from the aggregated observation model defined in Section~\ref{subsec:obs-model}, %% uppdated this
    % \Cref{ass:standing}
    % holds, 
    and set $r := \rk(M) \geq 1$.

    \textup{(i) (SBM.)} The top-$r$ right singular vectors $V_M \in
    \mathbb{R}^{n \times r}$ of $M$ satisfy $V_M = ZC$ for a matrix $C \in
    \mathbb{R}^{K \times r}$ with $C^\top N C = I_r$, so the rows are
    community-constant:
    \begin{equation}
        \label{eq:embedding-rows}
        (V_M)_{i \cdot} = \mu_{z_i},
        \qquad \mu_a := n_a^{-1/2}\, W_{a \cdot},
        \qquad W := N^{1/2} C \in \mathbb{O}_{K \times r} .
    \end{equation}
    Consequently $\sum_a n_a \mu_a \mu_a^\top = I_r$ and the embedding is
    incoherent,
    \begin{equation}
        \label{eq:incoherence}
        \| V_M \|_{2 \to \infty} = \max_a \| \mu_a \|_2 \leq n_{\min}^{-1/2} .
    \end{equation}
    The $K$ community centroids are pairwise distinct, that is,
    $\delta := \min_{a \neq b} \| \mu_a - \mu_b \|_2 > 0$, if and only if
    $O w_{ab} \neq 0$ for every pair $a \neq b$, as in \Cref{cor:partition}.
    Quantitatively,
    \begin{equation}
        \label{eq:delta-margin}
        \frac{\| O w_{ab} \|_2}{\sigma_1(OP^0)}
        \;\leq\; \| \mu_a - \mu_b \|_2
        \;\leq\; \frac{\| O w_{ab} \|_2}{\sigma_r(OP^0)},
        \qquad a \neq b :
    \end{equation}
    the centroid separation is the separation margin of \Cref{prop:pairwise},
    weighted by the spectrum of the signal. If moreover $r = K$ (equivalently,
    $\rk(OZ) = K$ and $B$ is non-singular), then $W$ is orthogonal and the
    separation and incoherence take the closed forms
    \begin{equation}
        \label{eq:sep}
        \| \mu_a - \mu_b \|_2^2
        = \frac{1}{n_a} + \frac{1}{n_b} \geq \frac{2}{n_{\max}},
        \qquad \| V_M \|_{2 \to \infty} = n_{\min}^{-1/2} .
    \end{equation}

    \textup{(ii) (DCSBM.)} The top-$r$ right singular vectors satisfy $V_M =
    \Theta Z C$ with $C^\top \widetilde N C = I_r$, where $\widetilde N :=
    Z^\top \Theta^2 Z = \diag(\tilde n_1, \dots, \tilde n_K)$, $\tilde n_a :=
    \sum_{i \in C_a} \theta_i^2$; the rows are $(V_M)_{i \cdot} = \theta_i\,
    \tilde\mu_{z_i}$ with $\tilde\mu_a := C_{a \cdot}$, and after normalizing
    each nonzero row to unit length and assuming all $\tilde\mu_a\ne0$,
    they take the $K$ values $\tilde\mu_a / \|
    \tilde\mu_a \|_2$, which are pairwise distinct if and only if the rays
    $\{\mathbb{R}_{\geq 0}\, \nu_a\}$ of \Cref{prop:dcsbm}
    are pairwise distinct. If moreover $r =
    K$, then $\widetilde W := \widetilde N^{1/2} C$ is orthogonal and
    \begin{equation}
        \label{eq:dcsbm-embedding-rows}
        \tilde\mu_a=\tilde n_a^{-1/2}\,\widetilde W_{a\cdot},
        \qquad \frac{\tilde\mu_a}{\|\tilde\mu_a\|_2}=\widetilde W_{a\cdot},
        \qquad \Big\|\frac{\tilde\mu_a}{\|\tilde\mu_a\|_2}
        -\frac{\tilde\mu_b}{\|\tilde\mu_b\|_2}\Big\|_2=\sqrt2,
        \qquad a\neq b,
    \end{equation}
    so the normalized population rows are orthonormal.
\end{lemma}

\begin{proof}[Proof of \Cref{lem:pop-embedding}(i)]
    By \eqref{eq:P0}, $M = O Z B Z^\top = (OZ) B Z^\top$, so $\range(M^\top)
    \subseteq \range(Z) = \mathcal{S}$: every right singular vector with nonzero
    singular value lies in $\mathcal{S}$, and $V_M = ZC$ for a (unique, since $Z$
    has full column rank) matrix $C \in \mathbb{R}^{K \times r}$. Orthonormality
    gives $I_r = V_M^\top V_M = C^\top Z^\top Z C = C^\top N C$, i.e.\ $W := N^{1/2}
    C$ has orthonormal columns, $W \in \mathbb{O}_{K \times r}$. Then $(V_M)_{i
    \cdot} = (ZC)_{i \cdot} = C_{z_i \cdot} = n_{z_i}^{-1/2} W_{z_i \cdot} =
    \mu_{z_i}$, which is \eqref{eq:embedding-rows}, and $\sum_a n_a \mu_a
    \mu_a^\top = \sum_a W_{a\cdot}^\top W_{a\cdot} = W^\top W = I_r$. Since the $K$
    summands $n_a \mu_a \mu_a^\top$ are positive semi-definite with sum $I_r$, each
    satisfies $n_a \| \mu_a \|_2^2 \leq 1$, giving \eqref{eq:incoherence}.

    For the distinctness claim, write the compact SVD $M = U \Sigma V_M^\top$, so
    that column $j$ of $M$ is $M_{\cdot j} = U \Sigma\, \mu_{z_j}$; on the other
    hand $M_{\cdot j} = (OZ) B_{\cdot z_j}$ from the factorization. Since $U
    \Sigma$ has full column rank $r$, for any pair $a \neq b$
    \[
        \mu_a = \mu_b
        \iff U\Sigma \mu_a = U\Sigma \mu_b
        \iff (OZ)(B_{\cdot a} - B_{\cdot b}) = 0
        \iff O w_{ab} = 0.
    \]
    using $Z(B_{\cdot a} - B_{\cdot b}) = w_{ab}$ from \eqref{eq:diff-profile}.
    Hence $\delta > 0$ iff $O w_{ab} \neq 0$ for every pair, the condition of
    \Cref{cor:partition}. The same identity gives \eqref{eq:delta-margin}:
    subtracting the two column expressions, $U \Sigma (\mu_a - \mu_b) = O w_{ab}$,
    and since $O w_{ab} = (OZ)(B_{\cdot a} - B_{\cdot b}) \in (OZ)\,\range(B) =
    \range(M)$ (the columns of the symmetric $B$ lie in $\range(B)$, and $Z^\top$
    is surjective), applying $U^\top$ preserves its norm: $\Sigma(\mu_a - \mu_b) =
    U^\top O w_{ab}$ with $\| U^\top O w_{ab} \|_2 = \| O w_{ab} \|_2$. Since the
    diagonal entries of $\Sigma$ lie in $[\sigma_r(M), \sigma_1(M)]$, this yields
    $\| O w_{ab} \|_2 / \sigma_1(M) \leq \| \mu_a - \mu_b \|_2 \leq \| O w_{ab}
    \|_2 / \sigma_r(M)$.

    Finally, if $r = K$ then $W$ is square and orthogonal. The row-norm and
    pairwise-distance identities in \eqref{eq:sep} follow as in the proof of
    \cite[Lemma~2.1]{lei_consistency_2015}. Conversely $r = K$ requires
    $\rk(OZ) = K$ and $\rk(B)
    = K$ since $\rk(M) \leq \min\{\rk(OZ), \rk(B)\}$, and equals their common value
    $K$ when both hold, as each factor in $(OZ)BZ^\top$ is then injective on the
    relevant range.
\end{proof}

\begin{lemma}[Signal lower bound]
    \label{lem:signal}
    Under
    the observed data model of Section~\ref{subsec:obs-model} where $B=\rho_n \overline{B}$ %% updated this
    % \Cref{ass:standing} 
    with $\rk(OZ) = K$, for every $1 \leq j \leq \rk(B)$,
    \begin{equation}
        \label{eq:sigma-lb}
        \sigma_j(OP^0)
        \geq \sigma_{\min}(OZ)\,\sigma_j(B)\,\sqrt{n_{\min}}
        = \sigma_{\min}(OZ)\,\rho_n\,\sigma_j(\bar B)\,\sqrt{n_{\min}}.
    \end{equation}
    in particular $\rk(OP^0) = \rk(B)$, and at full rank $\sigma_K(OP^0) \geq
    \sigma_{\min}(OZ)\, \sigma_K(B) \sqrt{n_{\min}}$.
\end{lemma}

\begin{proof}
    Write $OP^0 = (OZ)(BZ^\top)$. Since $(BZ^\top)(BZ^\top)^\top = B N B^\top
    \succeq n_{\min} B B^\top$, we have $\sigma_j(BZ^\top) \geq \sqrt{n_{\min}}\,
    \sigma_j(B)$ for every $j$; and since $OZ$ is injective with smallest singular
    value $\sigma_{\min}(OZ)$, $\sigma_j((OZ)Y) \geq \sigma_{\min}(OZ)\,
    \sigma_j(Y)$ for any $Y$. Composing the two gives \eqref{eq:sigma-lb}.
\end{proof}

\begin{lemma}[Deterministic diagonal correction]
    \label{lem:diag}
    Under
    the observed data model of Section~\ref{subsec:obs-model} where $P^0=\rho_n \overline{P}$
    % \Cref{ass:standing}, 
    let $M=OP^0$ and $D=O\diag(P^0)$.
    Then
    \begin{equation}
        \label{eq:diag-bound}
        \|D\|_{2\to2}\leq\|O\|_{2\to2}\max_j P^0_{jj}
        \leq\rho_n\|\bar P\|_\infty\|O\|_{2\to2}.
    \end{equation}
    For any matrix $W$ with orthonormal columns,
    \[
        \|D^\top W\|_{2\to\infty}
        \leq\rho_n\|\bar P\|_\infty b_O,
        \qquad
        \|D^\top U_M\|_{2\to\infty}
        \leq\rho_n\|\bar P\|_\infty b_\infty.
    \]
    Thus $R=M+\mcN$ with $\mcN=OE-D$, and
    \[
        \frac{\|D\|_{2\to2}}{\|O\|_{2\to2}\sqrt{n\rho_n}}
        \leq\|\bar P\|_\infty\sqrt{\rho_n/n},
        \qquad
        \frac{\|D\|_{2\to2}}{\sigma_r(M)}
        \leq\varepsilon_{\rm SNR}\|\bar P\|_\infty\sqrt{\rho_n/n}.
    \]
    The deterministic correction can therefore be absorbed into the stated
    noise scales when $\varepsilon_{\rm SNR}$ is bounded. All signal
    quantities remain those of $M=OP^0$.
\end{lemma}

\begin{proof}
    The diagonal entries of $P^0$ are bounded by
    $\rho_n\|\bar P\|_\infty$, so sub-multiplicativity gives
    \eqref{eq:diag-bound}. Moreover
    $(D^\top W)_{i\cdot}=P^0_{ii}(O^\top W)_{i\cdot}$ and
    $\|(O^\top W)_{i\cdot}\|_2\leq\|O_{\cdot i}\|_2$.
    Taking $W=U_M$ gives the sharper $b_\infty$ bound. Dividing
    \eqref{eq:diag-bound} by the noise scale and by $\sigma_r(M)$ proves
    the two ratios. The total-noise event and its leave-one-out versions
    are \eqref{eq:total-noise-event}.
\end{proof}

\subsection{Design quantities of the examples in \Cref{subsec:design-anatomy}}
\label{app:design-examples}

\begin{lemma}[Indicator partitions]
    \label{lem:design-examples}
    Let $O\in\{0,1\}^{L\times n}$ be the indicator matrix of a partition of $[n]$
    into $L$ groups of equal size $n/L$. Then
    \[
        b_O=1,\qquad\|O\|_{2\to2}=\sqrt{n/L},\qquad g_1(O)=n/L,\qquad
        \varsigma(O)=\sqrt{n/L}.
    \]
    Hence $O$ is spread in the sense of \Cref{def:design-regimes} if and only if
    $L\leq C_2^2K$, and near-orthogonal if and only if $L\geq n/C_\varsigma^2$. In
    particular:
    \begin{enumerate}
        \item[(i)] the identity $O=I_n$, the case $L=n$, has
        $b_O=\|O\|_{2\to2}=g_1(O)=\varsigma(O)=1$, so it is near-orthogonal and, for
        $K<n/C_2^2$, not spread;
        \item[(ii)] the design of \Cref{ex:enriched-groups}, the case $L=K$ with
        $m=n/K$, has $\|O\|_{2\to2}=\sqrt m$, $b_O=1=\|O\|_{2\to2}\sqrt{K/n}$,
        $g_1(O)=m$ and $\varsigma(O)=\sqrt m$, so it is spread with $C_2=1$ and
        not near-orthogonal unless $m\leq C_\varsigma^2$. Moreover
        \[
            OZ=m\Big((1-\tau)I_K+\frac\tau K\mathbf 1_K\mathbf 1_K^\top\Big),
            \qquad
            \sigma_{\min}(OZ)=m(1-\tau)=(1-\tau)\|O\|_{2\to2}\sqrt{n_{\min}},
        \]
        so the conditioning hypothesis of \Cref{cor:sbm-regimes} holds with
        $c_1=1-\tau$, and $b_\infty\leq1$.
    \end{enumerate}
    For the degree query $O=\mathbf 1_n^\top$, $g_1(O)=n$.
\end{lemma}

\begin{proof}
    Write $g(j)$ for the group of node $j$. Each column of $O$ is a standard basis
    vector, $O_{\cdot j}=e_{g(j)}$, so $b_O=1$. The rows of $O$ have disjoint
    supports, so $OO^\top=(n/L)I_L$ and $\|O\|_{2\to2}=\sqrt{n/L}$. Moreover
    $\langle O_{\cdot k},O_{\cdot j}\rangle=\mathbf 1\{g(k)=g(j)\}$, so
    $\|O^\top O_{\cdot j}\|_1=|G_{g(j)}|=n/L$; hence $g_1(O)=n/L$ and
    $\varsigma(O)=\sqrt{n/L}$. Spreadness reads $1\leq C_2\sqrt{n/L}\sqrt{K/n}
    =C_2\sqrt{K/L}$, and near-orthogonality reads $\sqrt{n/L}\leq C_\varsigma$.
    Part (i) is the case $L=n$.

    For (ii), each node of $C_a$ lies in exactly one group, and
    $|G_g|=\sum_a|G_g\cap C_a|=m\big((1-\tau)+K\cdot\tau/K\big)=m$, so the design is
    an equal partition with $L=K$. Since $(OZ)_{ga}=|G_g\cap C_a|$, $OZ$ has the
    displayed form; this symmetric matrix has eigenvalue $m$ on $\mathbf 1_K$ and
    $m(1-\tau)$ on $\mathbf 1_K^\perp$, so $\sigma_{\min}(OZ)=m(1-\tau)$, and
    $\|O\|_{2\to2}\sqrt{n_{\min}}=\sqrt m\sqrt m=m$. Finally
    $b_\infty=\max_j\|U_M^\top O_{\cdot j}\|_2\leq b_O=1$. For the degree query,
    $\langle O_{\cdot k},O_{\cdot j}\rangle=1$ for all $k,j$, so $g_1(O)=n$.
\end{proof}

\subsection{Proof of the weak-recovery bound}

In the appendix calculations, write $\beta:=n_{\max}/n_{\min}$;
under the assumed balance, its factors are absorbed into the constants
in the main statements.

\begin{lemma}[Weak recovery, precise form]
    \label{lem:weak-recovery-precise}
    Set $r=\rk(OP^0)$ and suppose $\delta>0$. Let
    $\nu=C_0\|O\|_{2\to2}\sqrt{n\rho_n}$, with $C_0$ large enough that
    $\|\mcN\|_{2\to2}\leq\nu$ on the event of \Cref{lem:noise},
    using \Cref{lem:diag}. Under $n\rho_n\geq C_0\log n$ and
    $\nu\leq(1-2^{-1/2})\sigma_r(M)$,
    with probability at least $1 - n^{-\alpha}$,
    \begin{equation}
        \label{eq:general-rate}
        n_{\mathrm{err}}(\hat z, z)
        \;\leq\; C(\varepsilon_{\rm km})\; \frac{r\, \| \mcN \|_{2\to2}^2}
        {\sigma_r(OP^0)^2\, \delta^2},
    \end{equation}
    provided the right side is smaller than $n_{\min}$. At full rank ($r = K$),
    substituting $\delta^2 \geq 2/n_{\max}$ and the design condition of
    \Cref{cor:sbm-regimes}, with $\bar\sigma:=\sigma_K(\bar B)$, gives the balanced
    form
    \begin{equation}
        \label{eq:main-rate}
        n_{\mathrm{err}}
        \leq C(\varepsilon_{\rm km})\,n_{\max}K
        \frac{\|\mcN\|_{2\to2}^2}{\sigma_K(OP^0)^2},
        \qquad
        \frac{n_{\mathrm{err}}}{n}
        \leq C\frac{K}{\bar\sigma^2}\,\beta\,
        \frac{\|O\|_{2\to2}^2}{\sigma_{\min}(OZ)^2}\frac{1}{\rho_n}
        \leq\frac{C\beta^2K^2}{\bar\sigma^2c_1^2\,n\rho_n}.
    \end{equation}
    %%% TO DO, check this
\end{lemma}

\begin{proof}[Proof of the bounds in \Cref{lem:weak-recovery-precise}]
    Work on the event of \Cref{lem:noise} (probability $\geq 1 - n^{-\alpha}$), and
    use $M=OP^0$ and $\mcN=OE-D$. By \Cref{lem:diag},
    $\|\mcN\|_{2\to2}\leq\nu$ after enlarging $C_0$.
    Under the stated SNR condition,
    \Cref{lem:sin-theta} gives $\varrho := \| \widehat V - V_M Q \|_F \leq
    2\sqrt{2r}\, \| \mcN \|_{2 \to 2} / \sigma_r(OP^0)$. By \Cref{lem:pop-embedding}(i)
    the rows of $V_M$ take the $K$ values $\mu_a$ at separation $\delta > 0$, so
    \Cref{lem:kmeans} gives
    \[
        n_{\mathrm{err}}
        \leq 16 (1 + \sqrt{1+\varepsilon_{\rm km}})^2 \frac{\varrho^2}{\delta^2}
        \leq 16 (1 + \sqrt{1+\varepsilon_{\rm km}})^2
        \frac{8r\| \mcN \|_{2 \to 2}^2}{\sigma_r(OP^0)^2\delta^2}
        = C(\varepsilon_{\rm km})\frac{r\| \mcN \|_{2\to2}^2}{\sigma_r(OP^0)^2\delta^2},
    \]
    which is \eqref{eq:general-rate}. At full rank $r = K$, substituting $\delta^2
    \geq 2/n_{\max}$ from \eqref{eq:sep} gives
    \[
        n_{\mathrm{err}}
        \leq 16 (1 + \sqrt{1+\varepsilon_{\rm km}})^2 \cdot \frac{n_{\max}}{2}
        \cdot \frac{8 K \| \mcN \|_{2 \to 2}^2}{\sigma_K(OP^0)^2}
        = C(\varepsilon_{\rm km})\, n_{\max} K \frac{\| \mcN \|_{2 \to 2}^2}{\sigma_K(OP^0)^2},
    \]
    the first bound of \eqref{eq:main-rate}.
    %% TO DO - check this
    For the second, substitute the total-noise bound $\|\mcN\|_{2\to2}\leq\nu$
    and the signal bound \eqref{eq:sigma-lb}, and
    divide by $n$:
    \[
        \frac{n_{\mathrm{err}}}{n}
        \leq C K\frac{n_{\max}}{n}
        \frac{\| O \|_{2\to2}^2\, n \rho_n}{\sigma_{\min}(OZ)^2\, \rho_n^2\,
        \bar\sigma^2\, n_{\min}}
        = C\frac{K}{\bar\sigma^2}\frac{n_{\max}}{n_{\min}}
        \frac{\|O\|_{2\to2}^2}{\sigma_{\min}(OZ)^2}\frac{1}{\rho_n},
    \]
    the factors of $n$ cancelling as $(n_{\max}/n)(n/n_{\min}) = n_{\max}/n_{\min}
    = \beta$. Under the design condition $\sigma_+(OZ)=\sigma_{\min}(OZ) \geq c_1 \| O
    \|_{2\to2} \sqrt{n_{\min}}$ of \Cref{cor:sbm-regimes}, $\| O
    \|_{2\to2}^2/\sigma_{\min}(OZ)^2
    \leq 1/(c_1^2 n_{\min})$ and $\beta / (n_{\min} \rho_n) \leq \beta^2 K /
    (n\rho_n)$ using $n_{\min} \geq n/(\beta K)$, giving the final bound of
    \eqref{eq:main-rate}.
\end{proof}

\subsection{Proof of the master theorem}

\begin{proof}[Proof of \Cref{thm:sbm-main}]
    Work with $M=OP^0$ and $\mcN=OE-D$ throughout, using the total-noise
    event \eqref{eq:total-noise-event}. By \Cref{lem:pop-embedding}(i), $M$ has rank $r$
    with community-constant rows and incoherence $\| V_M \|_{2\to\infty} \leq
    n_{\min}^{-1/2}$ \eqref{eq:incoherence}. The leave-one-out chain of
    \Cref{sec:procrustes} runs at the retained dimension
    $r$. The stated hypotheses
    $n\rho_n\geq C_0\log n$, $\varepsilon_{\rm SNR}\leq c_\star$ and
    $b_O\log n\leq c_\star\sigma_r$ are hypothesis (H) of
    \Cref{loo:concentration-full}, and the two additional conditions for
    $\vartheta=0$ are (H$'$). \Cref{loo:twoinfty-sharp} therefore gives, with
    probability at least $1 - 4n^{-\alpha}$, an orthogonal $Q$ satisfying
    \eqref{eq:loo-twoinfty-sharp} with $h=\|V_M\|_{2\to\infty}\leq n_{\min}^{-1/2}=h_\star$.
    Here $\gamma_1=\sqrt{\rho_n}(\sqrt r+\sqrt{\log n})$ and
    $\gamma_1(o+b_O)\leq2\gamma_1o$, as $b_O\leq o=\|O\|_{2\to2}$. The two $q_D$-terms
    are absorbed into the first term: since $q_D=\rho_n\|\bar P\|_\infty b_O$ and
    $\rho_n\leq\sqrt{\rho_n}\leq\gamma_1$,
    \[
        \frac{q_D}{\sigma_r}\leq\|\bar P\|_\infty\frac{\gamma_1o}{\sigma_r},
        \qquad
        \frac{q_D\,o\log n}{\sigma_r^2}
        =\|\bar P\|_\infty\frac{b_O\log n}{\sigma_r}\cdot\frac{\rho_no}{\sigma_r}
        \leq\|\bar P\|_\infty\frac{\gamma_1o}{\sigma_r},
    \]
    the second by the hypothesis $b_O\log n\leq c_\star\sigma_r$. This is
    \eqref{eq:eta-infty} with $C:=2C_{13}(1+\|\bar P\|_\infty)$.

On the same event,
    $\|\mcN\|_{2\to2}\leq C_0\|O\|_{2\to2}\sqrt{n\rho_n}$.
    Wedin's bound gives $\|\widehat VQ-V_M\|_F\leq F_\star$
    for the same orthogonal matrix $Q$, as in
    \eqref{eq:master-frobenius}. The population rows are the $K$ distinct
    centers $\mu_{z_i}$. If $\eta_\infty<\delta/4$,
    \Cref{lem:single-linkage} proves exact recovery with single linkage.
    For approximate $k$-means, under \eqref{eq:master-rounding},
    \Cref{lem:kmeans}(ii) therefore gives exact recovery by the
    nearest-centroid refinement of Algorithm~1.
\end{proof}

\subsection{Proof of the regime corollary}

\begin{proof}[Proof of \Cref{cor:sbm-regimes}]
    Throughout, $\sigma_r:=\sigma_r(OP^0)$, and by the signal condition in
    \eqref{eq:regime-conditions},
    \begin{equation}
        \label{eq:regime-signal}
        \sigma_r \geq c_1\|O\|_{2\to2}\rho_n\bar\sigma n_{\min},\qquad
        \varepsilon_{\rm SNR} = \frac{\| O \|_{2\to2}\sqrt{n\rho_n}}{\sigma_r}
        \leq \frac{\sqrt{n\rho_n}}{c_1\bar\sigma\rho_n n_{\min}}
        \leq \frac{\beta K}{c_1\bar\sigma\sqrt{n\rho_n}}.
    \end{equation}
    using $1/n_{\min} \leq \beta K / n$ (from $n_{\min} = n_{\max}/\beta \geq
    n/(\beta K)$, since $K n_{\max} \geq n$). We also record, using the separation
    condition in \eqref{eq:regime-conditions} for the first bound,
    \begin{equation}
        \label{eq:regime-geometry}
        \frac{1}{\delta} \leq \frac{\sqrt{n_{\max}}}{c_\delta},
        \qquad
        \frac{\sqrt{n_{\max}}}{n_{\min}} = \frac{\sqrt\beta}{\sqrt{n_{\min}}}
        \leq \beta \sqrt{\frac{K}{n}},\qquad
        \sqrt{\frac{n_{\max}}{n_{\min}}} = \sqrt\beta,
        \qquad
        \sqrt{\frac{K n_{\max}}{n}} \leq \sqrt\beta.
    \end{equation}
    the last using $n_{\max} \leq \beta n / K$.

    \emph{(i) Spread designs.} Put $d=n\rho_n$, $\ell=\log n$,
    $o=\|O\|_{2\to2}$, and $b=b_\infty+b_O$. Since $\vartheta\leq b_O$ and
    $h_\star\leq1$, the bound \eqref{eq:eta-infty} gives
    \[
        \eta_\infty\leq C\left\{\frac{o\sqrt{\rho_n}(\sqrt K+\sqrt\ell)+b\ell}{\sigma_r}
        +\frac{\varepsilon_{\rm SNR}}{\sqrt{n_{\min}}}
        +\varepsilon_{\rm SNR}^2\frac{b}{o}\right\}.
    \]
    Divide each term by $\delta$, using
    $b\leq2C_2o\sqrt{K/n}$, which follows from \Cref{def:design-regimes} and
    $b_\infty\leq b_O$, and \eqref{eq:regime-signal}--\eqref{eq:regime-geometry}:
    \begin{equation}
        \label{eq:three-terms}
        \frac{o\sqrt{K\rho_n}}{\sigma_r\delta}\lesssim\frac{\beta K}{\sqrt d},
        \quad \frac{o\sqrt{\rho_n\ell}}{\sigma_r\delta}\lesssim\beta\sqrt{\frac{K\ell}{d}},
        \quad \frac{b\ell}{\sigma_r\delta}\lesssim\frac{\beta K\ell}{d},
        \quad \frac{\varepsilon_{\rm SNR}n_{\min}^{-1/2}}{\delta}\lesssim\frac{\beta^{3/2}K}{\sqrt d},
        \quad \frac{\varepsilon_{\rm SNR}^2b}{o\delta}\lesssim\frac{\beta^{5/2}K^2}{d}.
    \end{equation}
    The first three inequalities use
    $o/(\sigma_r\delta)\lesssim\sqrt{n_{\max}}/(\rho_n n_{\min})$
    and $\sqrt{n_{\max}}/n_{\min}\leq\beta\sqrt{K/n}$.
    The fourth uses $n_{\min}^{-1/2}/\delta\leq\sqrt\beta/c_\delta$;
    the fifth uses $b/(o\delta)\leq2C_2\sqrt\beta/c_\delta$.
    All constants depend only on the conditioning constants stated in
    the corollary and on $\alpha$, not on $n,K,\beta$.

    Set $H_\beta=\beta^3K^2+\beta^2K\log n$. Since $\beta\geq1$,
    the five terms give
    \begin{equation}
        \label{eq:spread-normalized-rate}
        \frac{\eta_\infty}{\delta}
        \leq C\left\{\sqrt{\frac{H_\beta}{d}}+\frac{H_\beta}{d}\right\}.
    \end{equation}
    For sufficiently large fixed $c$, $d\geq cH_\beta$ makes this
    strictly below $1/4$ and ensures $d\geq C_0\log n$,
    $\varepsilon_{\rm SNR}\leq c_\star$ and, since $\delta\leq2$ and
    $b_O\ell/\sigma_r\leq\delta\cdot b\ell/(\sigma_r\delta)$, also
    $b_O\log n\leq c_\star\sigma_r$. The master theorem and
    \Cref{lem:single-linkage} prove the stated exact-recovery result
    with failure at most $4n^{-\alpha}$. For bounded imbalance the degree
    condition is $d\geq c(K^2+K\log n)$. The normalized error bound
    need only be below this fixed margin, rather than converge to zero.

    For the approximate-$k$-means option, the constants may additionally
    depend on the fixed $\varepsilon_{\rm km}$. The same event gives
    \[
        \frac{C_{\rm km}F_\star^2}{n_{\min}\delta^2}
        \leq C\beta K\varepsilon_{\rm SNR}^2
        \leq\frac{C\beta^3K^3}{d}.
    \]
    Hence $d\geq c(\beta^3K^3+\beta^2K\log n)$ ensures both
    conditions of \eqref{eq:master-rounding}, and
    \Cref{lem:kmeans}(ii) proves exact recovery after nearest-centroid reassignment.

    \emph{(ii) Near-orthogonal designs.} First,
    $\varsigma(O)\leq C_\varsigma$ gives
    $\varsigma\varepsilon_{\rm SNR}\leq C_\varsigma\beta K/(c_1\bar\sigma\sqrt d)$ and
    \[
        \frac{\varsigma o\log n}{\sigma_r}
        \leq\frac{C_\varsigma\log n}{c_1\bar\sigma\rho_nn_{\min}}
        \leq\frac{C_\varsigma\beta K\log n}{c_1\bar\sigma d},
    \]
    so both conditions for $\vartheta=0$ in \Cref{thm:sbm-main} hold once
    $d\geq cH_\beta$; so does $b_O\log n\leq c_\star\sigma_r$, since
    $b_O\leq\varsigma o$. The leverage also satisfies
    $b_\infty\leq C_2o\sqrt{K/n}$ with $C_2:=C_\varsigma\sqrt\beta/c_1$. Indeed,
    $\range(U_M)=\range(M)\subseteq\range(OZ)$, so $U_M=OZX$ for some
    $X\in\mathbb R^{K\times r}$ with columns orthogonal to $\ker(OZ)$, and
    $\|X\|_{2\to2}\leq1/\sigma_+(OZ)$ since
    $\|y\|_2=\|OZXy\|_2\geq\sigma_+(OZ)\|Xy\|_2$ for every $y$. Row $j$ of
    $O^\top U_M$ is $(O^\top OZ)_{j\cdot}X$, and row $j$ of $O^\top OZ$ has
    $\ell_1$ norm at most $\|O^\top O_{\cdot j}\|_1\leq g_1(O)$, so
    \[
        \|(O^\top U_M)_{j\cdot}\|_2\leq\frac{g_1(O)}{\sigma_+(OZ)}
        \leq\frac{C_\varsigma ob_O}{c_1o\sqrt{n_{\min}}}
        \leq\frac{C_\varsigma\sqrt\beta}{c_1}\,o\sqrt{\frac Kn},
    \]
    using $b_O\leq o$ and $n_{\min}\geq n/(\beta K)$. With $\vartheta=0$, divide the
    terms of \eqref{eq:eta-infty} by $\delta$, using $b_O\leq o$,
    $h_\star=n_{\min}^{-1/2}\leq\sqrt{\beta K/n}$,
    $b_\infty\leq C_2o\sqrt{K/n}$ and
    \eqref{eq:regime-signal}--\eqref{eq:regime-geometry}:
    \[
    \begin{gathered}
        \frac{\sqrt{\rho_n}(\sqrt K+\sqrt\ell)(o+b_O)}{\sigma_r\delta}
        \lesssim\frac{\beta K}{\sqrt d}+\beta\sqrt{\frac{K\ell}{d}},
        \qquad
        \frac{(b_\infty+b_Oh_\star)\ell}{\sigma_r\delta}
        \lesssim\frac{\beta^{3/2}K\ell}{d},\\
        \frac{\varepsilon_{\rm SNR}h_\star}{\delta}\lesssim\frac{\beta^{3/2}K}{\sqrt d},
        \qquad
        \frac{\varepsilon_{\rm SNR}^2b_\infty}{o\delta}\lesssim\frac{\beta^3K^2}{d},
    \end{gathered}
    \]
    These give \eqref{eq:spread-normalized-rate}, and the conclusion follows as in (i),
    including the Frobenius condition for approximate $k$-means, which
    depends on the design only through \eqref{eq:master-frobenius}.

    \emph{(iii) Generic column designs.} Put $s=b_O/o$. \Cref{thm:sbm-main} gives
    two bounds, and each yields exact recovery under one term of the minimum in
    \eqref{eq:sbm-generic-degree}. In both, the terms of \eqref{eq:eta-infty} that do
    not contain $\vartheta$ are bounded exactly as in (ii), which used only
    $b_O\leq o$ and $b_\infty\leq C_2o\sqrt{K/n}$; they require only
    $d\gtrsim K^2+K\log n$.

    \emph{The bound with $\vartheta=b_O\varepsilon_{\rm SNR}=so\,\varepsilon_{\rm SNR}$.}
    The two terms containing $\vartheta$ satisfy
    \[
        \frac{\vartheta\log n}{\sigma_r\delta}
        \leq\frac{s\,\varepsilon_{\rm SNR}\log n\sqrt{n_{\max}}}
        {c_\delta c_1\bar\sigma\rho_nn_{\min}}
        \leq\frac{s\,\beta^2K^{3/2}\sqrt n\log n}{c_\delta c_1^2\bar\sigma^2d^{3/2}},
        \qquad
        \frac{\varepsilon_{\rm SNR}^2\vartheta}{o\delta}
        \leq\frac{s\,\varepsilon_{\rm SNR}^3}{c_\delta}\sqrt{\frac{\beta n}{K}}
        \leq\frac{s\,\beta^{7/2}K^{5/2}\sqrt n}{c_\delta c_1^3\bar\sigma^3d^{3/2}}.
    \]
    For balanced communities ($\beta=O(1)$) these are below any fixed margin
    when $d\gtrsim s^{2/3}Kn^{1/3}(\log n)^{2/3}$ and
    $d\gtrsim s^{2/3}K^{5/3}n^{1/3}$, that is, when $d\gtrsim A_s$. The hypothesis
    $b_O\log n\leq c_\star\sigma_r$ holds because
    $b_O\log n/\sigma_r=s\,\varepsilon_{\rm SNR}\log n/\sqrt d\lesssim K\log n/d$.

    \emph{The bound with $\vartheta=0$.} Its two conditions hold when
    $d\gtrsim G_\varsigma$, since \eqref{eq:regime-signal} gives
    $\varsigma\varepsilon_{\rm SNR}\leq\varsigma\beta K/(c_1\bar\sigma\sqrt d)$ and
    $\varsigma o\log n/\sigma_r\leq\varsigma\beta K\log n/(c_1\bar\sigma d)$; the
    hypothesis $b_O\log n\leq c_\star\sigma_r$ then follows from $b_O\leq\varsigma o$.

    Either bound therefore gives $\eta_\infty/\delta<1/4$ under
    \eqref{eq:sbm-generic-degree}, and single linkage recovers all community labels up
    to a permutation, by \Cref{lem:single-linkage}. For approximate $k$-means, the
    Frobenius condition requires $d\geq c(\beta^3K^3+\beta^2K\log n)$ exactly as in (i),
    and \Cref{lem:kmeans}(ii) proves exact recovery.

    Since $K^{4/3}(\log n)^{2/3}\leq K^2+K\log n$ by the inequality of arithmetic and
    geometric means, $A_s\lesssim K^2+K\log n$ when $s\leq C_2\sqrt{K/n}$, and
    $G_\varsigma\lesssim K^2+K\log n$ when $\varsigma\leq C_\varsigma$; the single-linkage
    thresholds of (i) and (ii) are thus the two extreme cases of
    \eqref{eq:sbm-generic-degree}. For indicator queries of a partition of $[n]$ into
    $L$ equal-sized groups, $s=\sqrt{L/n}$ and $\varsigma=\sqrt{n/L}$, so
    $A_s=L^{1/3}(K^{5/3}+K(\log n)^{2/3})$ and
    $G_\varsigma=K^2n/L+K\sqrt{n/L}\log n$.
\end{proof}

\section{Proofs of the main guarantees for the DCSBM}
\label{app:dcsbm-proofs}
This appendix proves the first-order DCSBM results of \Cref{sec:recovery}:
\Cref{lem:pop-embedding}(ii), the master guarantee \Cref{thm:dcsbm-main}, the
regime corollary \Cref{cor:dcsbm-regimes}, and the weak-recovery bound of
% \Cref{rem:dcsbm-weak}.
Lemma~\ref{lem:dcsbm-weak-precise}. %% updated this
The argument is that of \Cref{app:sbm-proofs} with three
substitutions, and we prove only what the substitutions change.
\begin{enumerate}
    \item We replace $(OZ) B Z^\top$ by
    $M = (O \Theta Z) B (\Theta Z)^\top$. The embedding rows become $\theta_i
    \tilde\mu_{z_i}$, the incoherence becomes $\| V_M
    \|_{2\to\infty} \leq \theta_{\max} \tilde n_{\min}^{-1/2}$, and, when
    $\rk(O\Theta Z)=K$, the signal bound of \Cref{lem:signal} becomes
    $\sigma_j(M) \geq \sigma_{\min}(O \Theta Z)\, \rho_n \sigma_j(\bar B)
    \sqrt{\smash[b]{\tilde n_{\min}}}$ for $j\leq\rk(B)$, by the same
    three-factor argument with $\sigma_K((\Theta Z)^\top) = \min_a \tilde
    n_a^{1/2}$.
    \item We apply the concentration and leave-one-out bounds with
    $E=A-P$, $\mcN=OE-D$, and $M=OP^0$, as in the SBM.
    These bounds require only independent hollow Bernoulli noise, with variances bounded
    by $\rho_n\|\bar P\|_\infty$. Here
    $\|\bar P\|_\infty\leq\theta_{\max}^2\|\bar B\|_\infty$; the
    deterministic correction satisfies \Cref{lem:diag} with this bound.
    \item We normalize the rows before clustering, as in Algorithm~1.
    \Cref{lem:spherical} bounds the error after normalization using the original
    rowwise error and the smallest row norm $\delta_{\mathrm{row}}$ in
    \eqref{eq:dcsbm-rownorm}. At full rank, the normalized community vectors
    are orthonormal, with separation $\sqrt2$.
\end{enumerate}

\subsection{Population structure}

\begin{proof}[Proof of \Cref{lem:pop-embedding}(ii)]
    Now $M = O \Theta Z B Z^\top \Theta$, so $\range(M^\top) \subseteq \range(\Theta
    Z) = \mathcal{S}_\theta$ and $V_M = \Theta Z C$ for a matrix $C \in
    \mathbb{R}^{K \times r}$; the computation of part (i), with $\Theta Z$ in place
    of $Z$, gives $C^\top \widetilde N C = I_r$, so $\widetilde W := \widetilde
    N^{1/2} C \in \mathbb{O}_{K \times r}$. Then $(V_M)_{i \cdot} = \theta_i
    C_{z_i\cdot} = \theta_i\, \tilde\mu_{z_i}$ with $\tilde\mu_a := C_{a
    \cdot}$, and $\theta_i > 0$ preserves direction under normalization. For the
    distinctness criterion, write the compact SVD $M = U \Sigma V_M^\top$: column
    $j$ of $M$ equals both $U \Sigma\, \theta_j \tilde\mu_{z_j}$ and $\theta_j
    \nu_{z_j}$ (\Cref{lem:column-constancy}), so $\nu_a = U\Sigma\, \tilde\mu_a$
    with $U\Sigma$ injective; hence $\tilde\mu_a$ and $\tilde\mu_b$ are positively
    collinear if and only if $\nu_a$ and $\nu_b$ are, and the normalized values are
    pairwise distinct exactly under the ray condition of \Cref{prop:dcsbm}. If $r =
    K$, the degree-scaled row representation and its spherical normalization in
    \eqref{eq:dcsbm-embedding-rows} follow as in the proof of
    \cite[Lemma~4.1]{lei_consistency_2015}.
\end{proof}

\subsection{Proof of the weak-recovery bound}

\begin{lemma}[DCSBM weak recovery, precise form]
    \label{lem:dcsbm-weak-precise}
    Set $r=\rk(OP^0)$ and suppose $\Delta_\star>0$. Under the degree and
    total-noise gap conditions of \Cref{lem:weak-recovery-precise}, with
    probability at least $1 - n^{-\alpha}$,
    \begin{equation}
        \label{eq:dcsbm-general-rate}
        n_{\mathrm{err}}(\hat z,z)
        \;\leq\; C(\varepsilon_{\rm km})\,
        \frac{r\,\| \mcN \|_{2\to2}^2}{\sigma_r(OP^0)^2\,\Delta_\star^2},
    \end{equation}
    provided the right side is smaller than $n_{\min}$. At full rank ($r=K$),
    this gives
    \begin{equation}
        \label{eq:dcsbm-rate}
        n_{\mathrm{err}}(\hat z,z)
        \;\leq\; C(\varepsilon_{\rm km})K
        \left(\frac{\theta_{\max}}{\theta_{\min}}\right)^{\!2}n_{\max}
        \frac{\| \mcN \|_{2\to2}^2}{\sigma_K(OP^0)^2}.
    \end{equation}
\end{lemma}

\begin{proof}[Proof of Lemma~\ref{lem:dcsbm-weak-precise}]
    Work on the event of \Cref{lem:noise}, use $M=OP^0$ and
    $\mcN=OE-D$ with \Cref{lem:diag}, and apply \Cref{lem:sin-theta} to obtain
    an orthogonal $Q$ with
    \[
        \varrho := \| \widehat V - V_M Q \|_F
        \leq \frac{2\sqrt{2r}\, \| \mcN \|_{2\to2}}{\sigma_r(OP^0)} .
    \]
    By \Cref{lem:pop-embedding}(ii), the population rows are $(V_M)_{i\cdot} =
    \theta_i\tilde\mu_{z_i}$, with norms at least $\delta_{\mathrm{row}}$, and the
    normalized population rows take the $K$ values
    $\tilde\mu_a/\|\tilde\mu_a\|_2$, with separation $\delta_\theta$. Let
    $\widehat V^+$ denote the row-normalized $\widehat V$ of step~2 of Algorithm~1
    (zero rows mapped to a fixed unit vector). The population row norms are
    invariant under $Q$, so by \Cref{lem:spherical} applied row by row with
    $u = \widehat V_{i\cdot}$, $v = (V_M Q)_{i\cdot}$,
    \[
        \tilde\varrho^2 := \sum_{i=1}^n
        \Big\|\widehat V^+_{i\cdot}-\frac{\tilde\mu_{z_i}}{\|\tilde\mu_{z_i}\|_2}Q\Big\|_2^2
        \leq \sum_{i=1}^n
        \frac{4\big\|\widehat V_{i\cdot}-(V_MQ)_{i\cdot}\big\|_2^2}
        {\|(V_M)_{i\cdot}\|_2^2}
        \leq \frac{4\varrho^2}{\delta_{\mathrm{row}}^2}.
    \]
    \Cref{lem:kmeans} (with $\varrho \leftarrow \tilde\varrho$ and
    $\delta = \delta_\theta$; its hypothesis \eqref{eq:kmeans-cond} is the
    ``smaller than $n_{\min}$'' proviso) gives
    \[
        n_{\mathrm{err}}
        \leq C_{\rm km}\frac{\tilde\varrho^2}{\delta_\theta^2}
        \leq C_{\rm km}\frac{\varrho^2}{\Delta_\star^2}
        \leq \frac{8r\,C_{\rm km}\| \mcN \|_{2\to2}^2}{\sigma_r(OP^0)^2\,\Delta_\star^2},
    \]
    which is \eqref{eq:dcsbm-general-rate}, since
    $\Delta_\star=\tfrac12\delta_{\mathrm{row}}\delta_\theta$. At full rank,
    \eqref{eq:dcsbm-embedding-rows} gives $\delta_\theta=\sqrt2$ and
    $\delta_{\mathrm{row}}\geq\theta_{\min}\tilde n_{\max}^{-1/2}$, so
    $\Delta_\star^2\geq\theta_{\min}^2/(2\tilde n_{\max})$, and
    $\tilde n_{\max} \leq \theta_{\max}^2 n_{\max}$ gives \eqref{eq:dcsbm-rate}.
\end{proof}

\subsection{Proof of the master theorem}

\begin{proof}[Proof of \Cref{thm:dcsbm-main}]
    By substitution~2, \Cref{loo:twoinfty-sharp} holds for the DCSBM verbatim:
    with probability at least $1 - 4n^{-\alpha}$ there is an orthogonal $Q$ with
    \[
        \| \widehat V Q - V_M \|_{2 \to \infty} \;\leq\; \eta_\infty^\theta,
    \]
    where $\eta_\infty^\theta$ substitutes the DCSBM incoherence $\| V_M
    \|_{2\to\infty} \leq \theta_{\max} \tilde n_{\min}^{-1/2}$ into
    \eqref{eq:loo-twoinfty-sharp}; the chain runs at the retained dimension $r$ as
    recorded in \Cref{sec:procrustes}. Suppose now that the exact-recovery
    condition $\eta_\infty^\theta < \Delta_\star/4$ of \Cref{thm:dcsbm-main}
    holds; by the definition of $\Delta_\star$ in \eqref{eq:dcsbm-rownorm} this
    reads
    \begin{equation}
        \label{eq:dcsbm-exact-cond}
        \eta_\infty^\theta<\tfrac18\delta_{\mathrm{row}}\delta_\theta,
    \end{equation}
    and since $\delta_\theta \leq 2$, it implies in particular
    $\eta_\infty^\theta < \delta_{\mathrm{row}}/4$.

    \emph{No degenerate rows.} For every $i$, $\| \widehat V_{i\cdot} \|_2 \geq \|
    (V_M)_{i\cdot} \|_2 - \eta_\infty^\theta \geq \tfrac34\,
    \delta_{\mathrm{row}} > 0$ by \eqref{eq:dcsbm-rownorm}, so step~2 of
    Algorithm~1 normalizes every row and the arbitrary-assignment convention is
    never invoked.

Write $V_M^+$ and
    $\widehat V^+$ for the row-normalized embeddings. For the same
    alignment $Q$, \Cref{lem:spherical} and \eqref{eq:master-frobenius} give
    \[
        \|\widehat V^+Q-V_M^+\|_F\leq\frac{2F_\star}{\delta_{\mathrm{row}}},
        \qquad
        \|\widehat V^+Q-V_M^+\|_{2\to\infty}
        \leq\frac{2\eta_\infty^\theta}{\delta_{\mathrm{row}}}
        <\frac{\delta_\theta}{4}.
    \]
    The normalized population rows take $K$ distinct values at separation
    $\delta_\theta>0$. The rowwise inequality alone proves exact
    recovery with single linkage by \Cref{lem:single-linkage}.
    For approximate $k$-means, since
    $\Delta_\star=\delta_{\mathrm{row}}\delta_\theta/2$,
    the first inequality in \eqref{eq:master-rounding} is precisely a
    sufficient Frobenius condition for these normalized embeddings.
    \Cref{lem:kmeans}(ii) now proves exact recovery after the prescribed
    recompute-means and nearest-centroid refinement.
\end{proof}

\subsection{Proof of the regime corollary}

\begin{proof}[Proof of \Cref{cor:dcsbm-regimes}]
    We first record the effect of the normalization $\sum_{i \in C_a} \theta_i =
    n_a$ of
    Section~\ref{subsec:models}:
    % \Cref{ass:standing}:
    block averages of $\theta$ are one, so
    $\theta_{\min} \leq 1 \leq \theta_{\max}$, and by Cauchy--Schwarz and the
    trivial bound,
    \[
        n_a \;\leq\; \tilde n_a = \sum_{i \in C_a} \theta_i^2 \;\leq\;
        \theta_{\max} n_a .
    \]
    With $\theta_{\max}/\theta_{\min} \leq \bar\kappa$ and balanced communities,
    the separation and signal conditions of \Cref{cor:dcsbm-regimes} give, up to
    constants depending on $\bar\kappa$ and $c_\delta$,
    \[
        \delta_{\mathrm{row}} \geq
        \frac{c_\delta\theta_{\min}}{\sqrt{\theta_{\max}n_{\max}}}\asymp\sqrt{K/n},
        \qquad
        \|V_M\|_{2\to\infty}\leq
        \frac{\theta_{\max}}{\sqrt{n_{\min}}}\asymp\sqrt{K/n},
    \]
    \[
        \sigma_r(OP^0)\geq c_1\bar\sigma\|O\|_{2\to2}\rho_n\tilde n_{\min}
        \asymp c_1\bar\sigma\|O\|_{2\to2}\rho_nn/K,
    \]
    using $\tilde n_{\max}\leq\theta_{\max}n_{\max}$ for the first bound. These
    are, up to constants depending on $\bar\kappa$, $c_1$, $\bar\sigma$ and
    $c_\delta$, the same three quantities that enter the proof of
    \Cref{cor:sbm-regimes} in the balanced case, with $\delta_{\mathrm{row}}$ in the
    role of $\delta$. The target inequality
    $\eta_\infty^\theta<\Delta_\star/4=\delta_{\mathrm{row}}\delta_\theta/8$ holds
    when $\eta_\infty^\theta<c_\delta\delta_{\mathrm{row}}/8$, which takes the role
    of $\eta_\infty < \delta/4$. Under
    hypothesis (ii), the leverage bound of the proof of \Cref{cor:sbm-regimes}(ii)
    holds with $O\Theta Z$ in place of $OZ$ and an extra factor
    $\theta_{\max}/\theta_{\min}\leq\bar\kappa$, since each row of $O^\top O\Theta Z$
    has $\ell_1$ norm at most $\theta_{\max}g_1(O)$ and
    $\tilde n_{\min}\geq\theta_{\min}^2n_{\min}$. Under
    hypothesis (i) or (ii), the computations in the proof of
    \Cref{cor:sbm-regimes}(i)--(ii) therefore bound
    $\eta_\infty^\theta/\delta_{\mathrm{row}}$ by a constant times
    \[
        \frac K{\sqrt d}+\sqrt{\frac{K\log n}{d}}
            +\frac{K\log n}{d}+\frac{K^2}{d},\qquad d=n\rho_n.
    \]
    Taking $d\geq c(K^2+K\log n)$ makes the normalized rowwise error
    smaller than the required fixed margin. By \Cref{lem:spherical}
    and \Cref{lem:single-linkage}, single linkage on the normalized
    rows is exact. For the approximate-$k$-means option,
    $F_\star^2\lesssim K^3/d$ and, by the separation conditions,
    $n_{\min}\Delta_\star^2\gtrsim1$, so the additional condition
    $d\geq c(K^3+K\log n)$ verifies \eqref{eq:master-rounding}.
    \Cref{thm:dcsbm-main} concludes. Under the generic-column hypothesis (iii), the
    estimates in the proof of \Cref{cor:sbm-regimes}(iii), for both values of
    $\vartheta$, apply with the same substitutions, giving the
    threshold \eqref{eq:sbm-generic-degree} for single linkage; approximate
    $k$-means additionally requires $d\geq cK^3$ for the Frobenius condition, as
    above. For unbalanced communities, the three population quantities displayed
    above carry the factors of \eqref{eq:regime-geometry} with $\tilde n_{\min},
    \tilde n_{\max}$ in place of $n_{\min}, n_{\max}$ (comparable up to
    $\bar\kappa$), and the $\beta$-dependence of \Cref{cor:sbm-regimes} follows
    unchanged.
\end{proof}

\section{Proofs for second-order (fan-out) designs}
\label{app:fanout-proofs}
This appendix states and proves the second-order results by extending the arguments
of \Cref{app:sbm-proofs} and \Cref{sec:procrustes} to the stacked queries.

\subsection{Master guarantee for second-order queries}

For a fixed weight $\omega>0$, use the matrices in
\eqref{eq:fanout-stack} and \eqref{eq:fanout-dcsbm-H}. For the
degree-corrected model, write
\[
    B^{(2)}_\theta:=B\widetilde N B,\qquad
    M^\theta_\omega:=H^\theta_\omega(\Theta Z)^\top.
\]
For the stochastic block model, the exact mean and deterministic correction are
\begin{equation}
 \label{eq:fanout-mean-correction}
 \bbe{A^{\langle2\rangle}}=\hol(P^2),\qquad
 \bbe{R_\omega}=M_\omega+B_{\rm det},\qquad
 B_{\rm det}=\begin{bmatrix}
 -O_1\diag(P^0)\\
 \omega O_2\{\hol(P^2)-(P^0)^2\}
 \end{bmatrix}.
\end{equation}
The correction and centered fluctuations are controlled in
\Cref{lem:fanout-pop,lem:fanout-noise} and \eqref{eq:loo2-notation}.
Write $(H_\omega^\star,M_\omega^\star)$ for $(H_\omega,M_\omega)$ in the SBM
and $(H_\omega^\theta,M_\omega^\theta)$ in the DCSBM, and set
$r_\omega:=\rk(H_\omega^\star)$. The stacked analogues of the design scalars are
\begin{equation}
    \label{eq:fanout-Omega}
    \Omega_\omega:=\|O_1\|_{2\to2}
        +\omega(1+\|\bar B\|_\infty)n\rho_n\|O_2\|_{2\to2},\qquad
    \varepsilon_\omega:=\frac{\Omega_\omega\sqrt{n\rho_n}}
        {\sigma_{r_\omega}(M_\omega^\star)},\qquad
    b_\infty^\omega:=\|\bar{\mcO}^\top U_\omega\|_{2\to\infty}.
\end{equation}
where $U_\omega$ collects the left singular vectors of $M_\omega^\star$ and $\bar{\mcO}
:= [\,O_1;\ \omega O_2 P\,]$ is the \emph{population effective design}. 
Correspondingly, the column
norm $b_O$ is replaced by
\begin{equation}
    \label{eq:bomega}
    b_\omega:=b_{O_1}+\omega\max_j\|O_2P_{\cdot j}\|_2
        +\omega b_{O_2}\sqrt{n\rho_n\log n}
        +\Omega_\omega\frac{b_{O_2}}{\|O_2\|_{2\to2}},
\end{equation}
where $b_{O_k}:=\max_j\|(O_k)_{\cdot j}\|_2$. Write
$\tau_\omega:=\omega(1+\|\bar B\|_\infty)n\rho_n$, so that
$\Omega_\omega=\|O_1\|_{2\to2}+\tau_\omega\|O_2\|_{2\to2}$, and $\tau_\omega=1$ for
the canonical weight. The leverage of the second battery and the Gram overlap
$g_1$ of \Cref{subsec:design-anatomy} enter through
\begin{equation}
    \label{eq:fanout-gram}
    b_{\infty,2}^\omega:=\tau_\omega\|O_2^\top U_{\omega,2}\|_{2\to\infty},
    \qquad
    g_\omega:=g_1(O_1)+\tau_\omega^2\,g_1(O_2),
    \qquad
    \varsigma_\omega:=\frac{g_\omega}{\Omega_\omega b_\omega},
\end{equation}
where $U_{\omega,2}$ denotes the last $L_2$ rows of $U_\omega$.
We require a weight to normalize the magnitude of the two sets of queries; the
canonical weight is $\omega \asymp (n\rho_n)^{-1}$, which equalizes the two blocks'
noise levels.

\begin{theorem}[Master guarantee, fan-out queries for SBM and DCSBM]
    \label{thm:fanout-main}
    \label{prop:fanout-dcsbm-main}
    Consider either the $\mathrm{SBM}(z,B)$ or $\mathrm{DCSBM}(z,B,\theta)$ 
    defined in
    Section~\ref{subsec:models}, %% updated this
    % \Cref{ass:standing},
    with $B=\rho_n\bar B$, batteries satisfying
    $L_1\vee L_2\leq n$, and a fixed weight $\omega>0$. Assume
    $r_\omega\geq1$. Run Algorithm~1
    on $R_\omega$ with $r_\omega$ components and the corresponding model flag.
    Fix $\alpha>1$. There are constants $C_0,c_\star,C$, depending only on
    $\|\bar P\|_\infty$ and $\alpha$, such that if $n\rho_n\geq C_0\log n$,
    $\varepsilon_\omega\leq c_\star$ and
    $b_\omega\log n\leq c_\star\sigma_{r_\omega}(M_\omega^\star)$, then with probability
    at least $1-4n^{-\alpha}$ there is an orthogonal $Q$ such that
    \begin{equation}
        \label{eq:fanout-row-bound}
        \|\widehat VQ-V_\omega\|_{2\to\infty}
        \leq\eta_\infty^{(2),\star},
    \end{equation}
    where $V_\omega$ contains the top-$r_\omega$ right singular vectors of
    $M_\omega^\star$ and
    \begin{equation}
        \label{eq:fanout-eta}
        \eta_\infty^{(2),\star}=C\Bigg[
        \frac{\sqrt{\rho_n}(\sqrt{r_\omega}+\sqrt{\log n})\,\Omega_\omega}
        {\sigma_{r_\omega}(M_\omega^\star)}
        +\frac{(b_\infty^\omega+b_{\infty,2}^\omega+b_\omega h_\star
        +\vartheta_\omega)\log n}
        {\sigma_{r_\omega}(M_\omega^\star)}
        +\varepsilon_\omega h_\star
        +\varepsilon_\omega^2
        \frac{b_\infty^\omega+b_{\infty,2}^\omega+\vartheta_\omega}{\Omega_\omega}\Bigg],
    \end{equation}
    with $\vartheta_\omega:=b_\omega\varepsilon_\omega$. If in addition
    $\varsigma_\omega\varepsilon_\omega\leq c_\star$ and
    $\varsigma_\omega\Omega_\omega\log n\leq c_\star\sigma_{r_\omega}(M_\omega^\star)$,
    the same holds with $\vartheta_\omega:=0$.
    Here $h_\star=n_{\min}^{-1/2}$ in the SBM and
    $h_\star\leq\theta_{\max}\tilde n_{\min}^{-1/2}$ in the DCSBM; write the
    corresponding errors as $\eta_\infty^{(2)}$ and
    $\eta_\infty^{(2),\theta}$.
    On the same event and for the same alignment,
    \[
        \|\widehat VQ-V_\omega\|_F\leq
        F_\omega:=2\sqrt{2r_\omega}\,C_0\varepsilon_\omega,
    \]
    with $C_0$ enlarged to bound the total stacked noise.
    Define $\Delta_\omega=\delta$ in the SBM and
    $\Delta_\omega=\tfrac12\delta_{\mathrm{row}}\delta_\theta$ in the DCSBM, where
    $\delta$, $\delta_{\mathrm{row}}$ and $\delta_\theta$ are the quantities of
    \Cref{subsec:sbm-main} computed for $M_\omega^\star$.
    Require $\Delta_\omega>0$. For the approximate-$k$-means option
    additionally require the Frobenius condition
    \begin{equation}
        \label{eq:fanout-rounding}
        C_{\rm km}\frac{F_\omega^2}{\Delta_\omega^2}<n_{\min}.
    \end{equation}
    For the SBM, exact recovery holds if
    $\eta_\infty^{(2)}<\delta/4$, where $\delta$ is the stacked-embedding
    centroid separation, positive exactly when the columns of $H_\omega$ are
    pairwise distinct and at least $\sqrt{2/n_{\max}}$ at full rank. For the
    DCSBM, exact recovery holds if
    \begin{equation}
        \label{eq:fanout-dcsbm-exact}
        \eta_\infty^{(2),\theta}
        <\frac{\Delta_\omega}{4}=\frac{\delta_{\mathrm{row}}\delta_\theta}{8}.
    \end{equation}
    Under the displayed rowwise condition, Algorithm~1 with single
    linkage recovers the partition exactly by \Cref{lem:single-linkage}.
    The approximate-$k$-means option with nearest-centroid refinement
    also recovers it under \eqref{eq:fanout-rounding}, by
    \Cref{lem:kmeans}. In either case $n_{\mathrm{err}}(\hat z,z)=0$.
\end{theorem}

\subsection{The structural rank identity}

\begin{proof}[Proof of \Cref{prop:stacked-rank}]
    We prove the statement in four steps: the kernel identity, the deterministic
    upper bound on the rank, the observation that the full-rank locus is the
    complement of an algebraic set, and the construction of an explicit witness
    showing that this set is proper. The last step is where the earlier
    ``general position'' intuition is made rigorous.

    \emph{Step 1: kernel identity and a reformulation.}
    A vector $c \in \mathbb{R}^K$ lies in $\ker H$ iff $O_k Z B^{(k)} c = 0$ for every
    $k$, i.e.\ iff $B^{(k)} c \in \ker(O_k Z)$, i.e.\ iff $c \in (B^{(k)})^{-1}\ker(O_k
    Z)$ (each $B^{(k)}$ is invertible since $B$ is, which holds for generic $B$);
    intersecting over $k$
    gives $\ker H = \bigcap_{k=1}^m (B^{(k)})^{-1}\ker(O_k Z)$, and each subspace
    $(B^{(k)})^{-1}\ker(O_k Z)$ has
    dimension $K - r_k$, $r_k := \rk(O_k Z)$. For the rank formula it is convenient
    to transpose. Set $T := NB$ and $S_k := \range\big( (O_k Z)^\top \big) \subseteq
    \mathbb{R}^K$, $\dim S_k = r_k$. Since $(B^{(k)})^\top = (BN)^{k-1} B = N^{-1}
    T^k$ (induction on $k$: for $k = 1$ this reads $B = N^{-1}NB$, and multiplying
    by $BN$ on the left advances $k$), the column space of $H^\top$ is
    \[
        \sum_{k=1}^m (B^{(k)})^\top S_k = N^{-1} \sum_{k=1}^m T^k S_k,
        \qquad\text{so}\qquad
        \rk H = \dim \sum_{k=1}^m T^k S_k .
    \]
    As $B$ ranges over invertible symmetric matrices, $T = NB$ ranges over the
    invertible matrices that are self-adjoint for the inner product
    $\langle x, y \rangle_{N^{-1}} := x^\top N^{-1} y$; every such $T$ is
    diagonalizable with real nonzero eigenvalues and an
    $\langle\cdot,\cdot\rangle_{N^{-1}}$-orthonormal eigenbasis, and conversely
    every matrix of that form is $NB$ for the symmetric invertible $B = N^{-1}T$.

    \emph{Step 2: upper bound.} Since $T$ is invertible, $\dim T^k S_k = r_k$ and,
    for every admissible $B$,
    \[
        \rk H = \dim \sum_k T^k S_k
        \leq \min\Big( K, \sum_k r_k \Big) =: R.
    \]

    \emph{Step 3: the full-rank locus is Zariski-open.} The entries of $H$ are
    polynomials in the entries of $B$, so each $R \times R$ minor of $H$ is a
    polynomial in $B$, and $\{ B : \rk H(B) < R \}$ is the common zero set of all
    such minors: an algebraic subset of the symmetric matrices. If it is proper
    (i.e.\ some minor is not identically zero), then $\rk H(B) = R$ holds for all
    $B$ outside the zero set of a nonzero polynomial, which is closed and
    Lebesgue-null, as required by the definition of ``generic'' in
    \Cref{subsec:fanout}. A
    witness $B_0$ with $\rk H(B_0) = R$ establishes properness, and $B_0$ need not
    itself be a valid connectivity matrix: nonvanishing of a polynomial anywhere
    on $\mathrm{Sym}_K$ implies its zero set is null in every open subset,
    including the interior of $[0,1]^{K \times K}$.

    \emph{Step 4: witness.} It suffices to treat $\sum_k r_k \leq K$: otherwise
    replace each $S_k$ by a subspace $S_k' \subseteq S_k$ with $\sum_k \dim S_k' =
    K$; a witness for the reduced data gives $\rk H \geq \dim \sum_k T^k S_k' = K =
    R$. Fix a basis $s_{k,1}, \dots, s_{k,r_k}$ of each $S_k$, choose disjoint
    index blocks $J_1, \dots, J_m \subseteq [K]$ with $\abs{J_k} = r_k$, and build
    $T$ from an eigenbasis $v_j := N^{1/2} w_j$ with $w = (w_1, \dots, w_K) \in
    O(K)$ and eigenvalues constant on blocks,
    \[
        T v_j = \lambda_j v_j,
        \qquad
        \lambda_j := t^{k} \text{ for } j \in J_k,
        \qquad
        \lambda_j := 1 \text{ for } j \notin \textstyle\bigcup_k J_k,
    \]
    with $t > 1$ to be chosen large. Any such $T$ is admissible: $N^{-1}T =
    N^{-1/2} (w \diag(\lambda) w^\top) N^{-1/2}$ is symmetric and invertible.
    Writing $c_j(s) := w_j^\top N^{-1/2} s$ for the eigencoordinates of $s$, we
    have $T^k s = \sum_j \lambda_j^k c_j(s) v_j$, so, the basis $(v_j)$ being
    invertible, $\rk H$ equals the rank of the coordinate matrix $\big[
    \lambda_j^k\, c_j(s_{k,l}) \big]_{j \in [K],\, (k,l)}$.

    Choose $w$ so that, for every $k$, the $r_k \times r_k$ block $\Gamma_k :=
    \big[ c_j(s_{k,l}) \big]_{j \in J_k,\, l \in [r_k]}$ is invertible. This is
    possible: $\Gamma_k$ is singular exactly when the fixed $r_k$-dimensional
    subspace $N^{-1/2} S_k$ meets $\mathrm{span}\{ w_j : j \notin J_k \}$
    nontrivially, and for a uniformly random $w \in O(K)$ the latter is a uniformly
    random $(K - r_k)$-dimensional subspace, which intersects a fixed subspace of
    complementary dimension trivially with probability one; a union bound over the
    $m$ blocks leaves almost every $w$ admissible.

    Now consider the $R \times R$ submatrix $D(t)$ with rows $j \in \bigcup_k J_k$.
    Each permutation $\sigma$ in the Leibniz expansion of $\det D(t)$ contributes
    the $t$-exponent $E(\sigma) = \sum_{(k,l)} k \cdot b(\sigma(k,l))$, where
    $b(j) = k'$ for $j \in J_{k'}$: a pairing of the multiset of column powers $\{
    k \text{ with multiplicity } r_k \}$ with the identical multiset of row-block
    labels. If a pairing matches some $k_1 < k_2$ to labels $k_2' > k_1'$
    respectively, swapping increases $E$ by $(k_2 - k_1)(k_2' - k_1') > 0$, so
    every maximizing pairing is monotone, and the monotone pairing of two identical
    multisets matches equal values. Hence $E(\sigma) \leq E^* := \sum_k k^2 r_k$,
    with equality exactly when $\sigma$ maps each column block $k$ bijectively onto
    $J_k$; summing the signed products of those permutations factorizes over
    blocks, giving
    \[
        \det D(t) \;=\; t^{E^*} \prod_{k=1}^m \det \Gamma_k \;+\; O\big( t^{E^* - 1}
        \big) \;\neq\; 0 \quad \text{for all $t$ large enough,}
    \]
    since each $\det \Gamma_k \neq 0$. Thus $\rk H(B_0) \geq R$ for $B_0 :=
    N^{-1}T$ with $t$ large, which with Steps 2--3 completes the proof of
    \eqref{eq:stacked-rank}. The final claim of the proposition follows since
    the population right singular space equals the entire block subspace iff
    $\rk H=K$, equivalently $\sum_k r_k\geq K$ for generic $B$.
    Partition separation only requires pairwise distinct columns of $H$,
    as in \Cref{lem:fanout-embedding}.
\end{proof}

\subsection{Population structure and operator-norm noise at order two}

\begin{lemma}[Population structure of the fan-out matrix]
    \label{lem:fanout-pop}
    Under
    the model of Section~\ref{subsec:models}
    % \Cref{ass:standing},
    for the SBM with $B = \rho_n \bar B$,
    \begin{equation}
        \label{eq:fanout-pop}
        \bbe{O A^{\langle 2 \rangle}} = O Z B^{(2)} Z^\top + \Delta_2,
        \qquad \| \Delta_2 \|_{2 \to 2}
        \leq 4\,\|O\|_{2 \to 2}\,n\rho_n^2\,\|\bar B\|_\infty^2.
    \end{equation}
    for any $O \in \mathbb{R}^{L \times n}$. Moreover $\ker B^{(2)} = \ker B$
    (since $B^{(2)} c = 0$ implies $(Bc)^\top N (Bc) = 0$, hence $Bc = 0$, $N$
    being positive definite); in particular $B^{(2)} = BNB$ is invertible whenever
    $B$ is, and then $\rk(O Z B^{(2)}) = \rk(OZ)$.
\end{lemma}

\begin{proof}
    For $i \neq j$, the edges $\{i,k\}$ and $\{k,j\}$ are distinct for every $k$, so
    by independence $\bbe{(A^2)_{ij}} = \sum_k \bbe{A_{ik}}\bbe{A_{kj}} = (P^2)_{ij}$
    (the terms $k \in \{i, j\}$ vanish since $P$ is hollow). Hence
    $\bbe{A^{\langle 2 \rangle}} = \hol(P^2)$. Writing $P = P^0 - D_0$ with
    $D_0 := \diag(P^0)$, $\| D_0 \|_{2\to2} \leq \rho_n \| \bar B \|_\infty$, and
    using $\| P^0 \|_{2\to2} \leq \max_i \sum_j P^0_{ij} \leq n \rho_n \| \bar B
    \|_\infty$,
    \[
        \| P^2 - (P^0)^2 \|_{2\to2}
        = \| -P^0 D_0 - D_0 P^0 + D_0^2 \|_{2\to2}
        \leq 3\, n \rho_n^2 \| \bar B \|_\infty^2,\qquad
        \| \diag(P^2) \|_{2\to2} = \max_i \sum_k P_{ik}^2
        \leq n \rho_n^2 \| \bar B \|_\infty^2.
    \]
    Since $(P^0)^2 = Z B (Z^\top Z) B Z^\top = Z B^{(2)} Z^\top$, the display
    \eqref{eq:fanout-pop} follows with $\Delta_2 := O(\hol(P^2) - (P^0)^2)$ and
    sub-multiplicativity. The kernel identity is the parenthetical computation in
    the statement, and invertibility for non-singular $B$ follows from it.
\end{proof}

\begin{lemma}[Second-order aggregated noise]
    \label{lem:fanout-noise}
    There is a constant $C_0=C_0(\|\bar P\|_\infty,\alpha)$ (the same as in \Cref{lem:noise}, enlarged if
    necessary) such that, if $n\rho_n \geq C_0 \log n$, then with probability at
    least $1 - n^{-\alpha}$, for any fixed $O \in \mathbb{R}^{L \times n}$,
    \begin{equation}
        \label{eq:fanout-noise}
        \big\| O \big( A^{\langle 2 \rangle} - \bbe{A^{\langle 2 \rangle}} \big)
        \big\|_{2 \to 2}
        \leq C_4\, \| O \|_{2 \to 2}\, (n\rho_n)^{3/2},
        \qquad C_4 := 6\, C_0^2\, (1 + \| \bar B \|_\infty)^{3/2}.
    \end{equation}
    The constant $C_4$ depends only on $C_0$ and $\| \bar B \|_\infty$; in
    particular it carries no hidden dependence on $n$, $K$, or the design $O$.
\end{lemma}

\begin{proof}
    Write $E = A - P$, so $A^2 = P^2 + PE + EP + E^2$ and $\bbe{A^2} = P^2 +
    \bbe{E^2}$, where $\bbe{E^2} = \diag\!\big(\sum_k \Var(A_{ik})\big)_i$ is
    diagonal. Applying the linear map $\hol$ to $A^2 - \bbe{A^2}$ kills every
    diagonal term:
    \begin{equation}
        \label{eq:fanout-noise-decomp}
        A^{\langle 2 \rangle} - \bbe{A^{\langle 2 \rangle}}
        = \hol(PE) + \hol(EP) + \hol(E^2) .
    \end{equation}
    For any square $X$, $\| \hol(X) \|_{2\to2} \leq \| X \|_{2\to2} + \max_i |X_{ii}|
    \leq 2 \| X \|_{2\to2}$. On the event of \Cref{lem:bernoulli-spectral}
    (probability $\geq 1 - n^{-\alpha}$, with $d = n\rho_n (1 + \| \bar B \|_\infty)$,
    using $n \max_{ij} P_{ij} \leq n\rho_n \| \bar B \|_\infty \leq d$ and
    $n\rho_n \geq C_0 \log n$), $\| E \|_{2\to2} \leq C_0 \sqrt{n\rho_n (1 + \| \bar
    B \|_\infty)}$, and $\| P \|_{2\to2} \leq n\rho_n \| \bar B \|_\infty$
    deterministically. Hence
    \[
        \| \hol(PE) + \hol(EP) + \hol(E^2) \|_{2\to2}
        \leq 2\big( 2 \| P \|_{2\to2} \| E \|_{2\to2} + \| E \|_{2\to2}^2 \big)
        \leq C_4\, (n\rho_n)^{3/2}.
    \]
    absorbing $\| \bar B \|_\infty$ factors into $C_4$ as stated. Multiplying by $O$
    and using sub-multiplicativity gives \eqref{eq:fanout-noise}.
\end{proof}

\begin{lemma}[Population embedding, stacked]
    \label{lem:fanout-embedding}
    Suppose
    the model defined in Section~\ref{subsec:models}
    % \Cref{ass:standing}
    and the stacked definitions of \ref{eq:fanout-stack} %%updated this
    hold, and set $r_\omega := \rk(H_\omega) =
    \rk(M_\omega)$. Then the conclusions of \Cref{lem:pop-embedding}(i) hold for
    the top-$r_\omega$ right singular vectors $V_\omega$ of $M_\omega$:
    community-constant rows \eqref{eq:embedding-rows}, the incoherence bound
    \eqref{eq:incoherence}, centroids pairwise distinct iff the columns of
    $H_\omega$ are, and when $r_\omega=K$, the closed forms \eqref{eq:sep}
    with $\range(V_\omega) = \mathcal{S}$.
\end{lemma}

\begin{proof}
    The proof of \Cref{lem:pop-embedding}(i) uses only that $M = F Z^\top$ for some
    matrix $F$ (there $F = (OZ)B$, whose columns $F_{\cdot a} = (OZ)B_{\cdot a}$
    index the communities); the single change is $F = H_\omega$, which is
    immaterial because no step of that proof inspects the internal structure of
    $F$. The distinctness criterion becomes $\mu_a = \mu_b \iff (H_\omega)_{\cdot
    a} = (H_\omega)_{\cdot b}$, by the same full-column-rank argument.
\end{proof}

\subsection{The weak-recovery bound}

Under the canonical weight $\omega = ((1+\|\bar B\|_\infty) n\rho_n)^{-1}$ of
\Cref{cor:fanout-regimes}, the stacked signals factor as $H_\omega = \rho_n \bar H$
and $H_\omega^\theta = \rho_n \bar H^\theta$, where
\begin{equation}
    \label{eq:fanout-Hbar}
    \bar H=
    \begin{bmatrix}
    O_1Z\bar B\\
    \dfrac{O_2Z\bar B(N/n)\bar B}{1+\|\bar B\|_\infty}
    \end{bmatrix},
    \qquad
    \bar H^\theta=
    \begin{bmatrix}
    O_1\Theta Z\bar B\\
    \dfrac{O_2\Theta Z\bar B(\widetilde N/n)\bar B}
    {1+\|\bar B\|_\infty}
    \end{bmatrix}.
\end{equation}

\begin{remark}[Weak recovery, stacked]
    \label{rem:fanout-weak}
    As at first order, the Frobenius route needs none of the local scalars.
    Suppose $n\rho_n\geq C_0\log n$, the canonical weight,
    $\Delta_\omega>0$ with $\Delta_\omega$ as in \Cref{thm:fanout-main}, and
    $C_3\Omega_\omega\sqrt{n\rho_n}\leq(1-2^{-1/2})
    \sigma_{r_\omega}(M_\omega^\star)$. Then, with probability at least
    $1-2n^{-\alpha}$,
    \begin{equation}
        \label{eq:fanout-general-rate}
        n_{\mathrm{err}}(\hat z,z)
        \leq C(\varepsilon_{\rm km})C_3^2\,
        \frac{r_\omega\,\Omega_\omega^2\,n\rho_n}
        {\sigma_{r_\omega}(M_\omega^\star)^2\,\Delta_\omega^2},
    \end{equation}
    provided the right side is smaller than $n_{\min}$. At $r_\omega=K$, the SBM
    satisfies
    \begin{equation}
        \label{eq:fanout-frac}
        \frac{n_{\mathrm{err}}}{n}
        \leq C(\varepsilon_{\rm km})C_3^2K\beta
        \frac{\Omega_\omega^2}{\sigma_K(\bar H)^2}\frac{1}{\rho_n},
        \qquad
        \frac{n_{\mathrm{err}}}{n}
        \lesssim\frac{\beta^2K^2}{c_1^2n\rho_n},
    \end{equation}
    the latter under \eqref{eq:fanout-well-cond}, provided $n$ times the applicable
    upper bound in \eqref{eq:fanout-frac} is smaller than $n_{\min}$, and the DCSBM
    satisfies
    \begin{equation}
        \label{eq:fanout-dcsbm-rate}
        n_{\mathrm{err}}(\hat z,z)
        \leq C(\varepsilon_{\rm km})C_3^2K
        \left(\frac{\theta_{\max}}{\theta_{\min}}\right)^{\!2}
        \frac{n_{\max}\Omega_\omega^2n\rho_n}
        {\sigma_K(M_\omega^\theta)^2},
    \end{equation}
    provided the right side is smaller than $n_{\min}$.
    Under these degree, gap, $k$-means, and conditioning assumptions,
    the error fraction vanishes if $K^2/(n\rho_n)\to0$, at fixed imbalance
    and bounded degree heterogeneity. For fixed $K$, this last condition is
    $n\rho_n\to\infty$; for growing $K$, the $k$-means condition remains an
    additional requirement.
\end{remark}

\begin{proof}[Proof of the bounds in \Cref{rem:fanout-weak}]
    Write $N_\omega := R_\omega - M_\omega$ and bound the two blocks separately.
    The first block is $O_1 A - O_1 Z B Z^\top = O_1 E - O_1 \diag(P^0)$, so by
    \Cref{lem:noise} and \eqref{eq:diag-bound}, with probability at least $1 -
    n^{-\alpha}$,
    \[
        \| O_1 A - O_1 Z B Z^\top \|_{2\to2}
        \leq \| O_1 \|_{2\to2} \big( C_0 \sqrt{n\rho_n} + \rho_n \| \bar B \|_\infty \big)
        \leq 2 C_0 \| O_1 \|_{2\to2} \sqrt{n\rho_n}\, (1 + \| \bar B \|_\infty).
    \]
    The second block is $\omega O_2 (A^{\langle 2 \rangle} - \bbe{A^{\langle 2
    \rangle}}) + \omega \Delta_2$ with $\Delta_2$ as in \Cref{lem:fanout-pop}; by
    \Cref{lem:fanout-noise} and \eqref{eq:fanout-pop}, with probability at least $1 -
    n^{-\alpha}$,
    \[
        \omega \big\| O_2 A^{\langle 2 \rangle} - O_2 Z B^{(2)} Z^\top \big\|_{2\to2}
        \leq \omega \| O_2 \|_{2\to2} \Big( C_4 (n\rho_n)^{3/2}
        + 4 n \rho_n^2 \| \bar B \|_\infty^2 \Big)
        \leq 2 C_4\, \omega \| O_2 \|_{2\to2} (n\rho_n)^{3/2}.
    \]
    the second inequality because $4 n\rho_n^2 \| \bar B \|_\infty^2 = 4 (n\rho_n)
    (\rho_n \| \bar B \|_\infty) \| \bar B \|_\infty \leq 4\, n\rho_n \| \bar B
    \|_\infty \leq C_4 (n\rho_n)^{3/2}$, using $\rho_n \| \bar B \|_\infty \leq 1$
    (entries of $B$ are probabilities) and $\sqrt{n\rho_n} \geq 1$. Since $\| [X; Y]
    \|_{2\to2} \leq \| X \|_{2\to2} + \| Y \|_{2\to2}$ for a stacked matrix, on the
    intersection of the two events (probability $\geq 1 - 2n^{-\alpha}$),
    \begin{equation}
        \label{eq:fanout-total-noise}
        \| N_\omega \|_{2\to2}
        \leq C_3 \Big( \| O_1 \|_{2\to2} + \omega (1 + \| \bar B \|_\infty)
        n\rho_n \| O_2 \|_{2\to2} \Big) \sqrt{n\rho_n}
        = C_3\, \Omega_\omega \sqrt{n\rho_n}.
    \end{equation}
    where
    \[
        C_3 := 2C_0(1 + \|\bar B\|_\infty)
        \vee \frac{2C_4}{1 + \| \bar B \|_\infty}.
    \]
    Under the stated SNR condition, \Cref{lem:sin-theta} applies with $(M, \mcN)$
    replaced by $(M_\omega,N_\omega)$ and gives an orthogonal matrix $Q$
    such that
    \[
        \varrho := \| \widehat V - V_\omega Q \|_F
        \leq \frac{2\sqrt{2r_\omega}\, \| N_\omega \|_{2\to2}}
        {\sigma_{r_\omega}(M_\omega)} .
    \]
    By \Cref{lem:fanout-embedding} the rows of $V_\omega$ take $K$ distinct values
    with separation $\delta=\Delta_\omega>0$. Applying
    \Cref{lem:kmeans}(i) under the stated $k$-means condition gives
    \[
        n_{\mathrm{err}}(\hat z, z)
        \leq C_{\rm km}\frac{\varrho^2}{\delta^2}
        \;\leq\; C(\varepsilon_{\rm km})\, C_3^2\,
        \frac{r_\omega\,\Omega_\omega^2\, n\rho_n}{\sigma_{r_\omega}(M_\omega)^2\,\delta^2},
    \]
    which is \eqref{eq:fanout-general-rate} for the SBM. At $r_\omega=K$,
    $\delta^2\geq2/n_{\max}$ by \Cref{lem:fanout-embedding}, so
    \[
        n_{\mathrm{err}}(\hat z, z)
        \leq C(\varepsilon_{\rm km})\, C_3^2\, K\,
        \frac{n_{\max}\, \Omega_\omega^2\, n\rho_n}{\sigma_K(M_\omega)^2} .
    \]
    For the fraction bound \eqref{eq:fanout-frac}: the canonical weight gives
    $H_\omega = \rho_n \bar H$ with $\bar H$ as in \eqref{eq:fanout-Hbar}, and the
    argument of \Cref{lem:signal} gives $\sigma_K(M_\omega) = \sigma_K(H_\omega
    Z^\top) \geq \sigma_K(H_\omega) \sqrt{n_{\min}} = \rho_n \sigma_K(\bar H)
    \sqrt{n_{\min}}$; substituting and dividing by $n$,
    \[
        \frac{n_{\mathrm{err}}}{n}
        \leq C(\varepsilon_{\rm km}) C_3^2 K\, \frac{n_{\max}}{n}\,
        \frac{\Omega_\omega^2\, n \rho_n}{\rho_n^2\, \sigma_K(\bar H)^2\, n_{\min}}
        = C(\varepsilon_{\rm km}) C_3^2 K\, \beta\,
        \frac{\Omega_\omega^2}{\sigma_K(\bar H)^2}\, \frac{1}{\rho_n} .
    \]
    Under \eqref{eq:fanout-well-cond} the middle factor is at most $1/(c_1^2
    n_{\min})$, and $\beta / (n_{\min} \rho_n) \leq \beta^2 K / (n\rho_n)$ gives
    the final bound of \eqref{eq:fanout-frac}.
\end{proof}

\subsection{The stacked leave-one-out chain}

We now extend \Cref{loo:concentration-full,loo:twoinfty-sharp} to the stacked
estimator. Assume throughout $L_1 \vee L_2 \leq n$; this
is used only to absorb the Bernstein dimension factor $\log(L_2 + 1)$ into
constants, and carrying it explicitly would replace $\log n$ by $\log(n \vee L_2)$
in the third piece of \eqref{eq:bomega}. Throughout this perturbation chain,
$r_\omega=\rk(M_\omega)\geq1$ is the retained rank; $K$ remains the
number of communities in the applications.
Write
\begin{equation}
    \label{eq:loo2-notation}
    \begin{aligned}
    R_\omega &= M_\omega + N_\omega, \qquad
    N_\omega = B_{\mathrm{det}} + S(E),\\
    S(E) &:= \begin{bmatrix} O_1 E \\ \omega\, O_2\, \Psi(E) \end{bmatrix},
    \qquad \Psi(E) := \hol(PE) + \hol(EP) + \hol(E^2).
    \end{aligned}
\end{equation}
Here the deterministic matrix
\[
    B_{\mathrm{det}}:=\begin{bmatrix}
        -O_1\diag(P^0)\\
        \omega O_2\{\hol(P^2)-(P^0)^2\}
    \end{bmatrix}
\]
satisfies, by \Cref{lem:fanout-pop},
\[
    \|B_{\mathrm{det}}\|_{2\to2}
    \leq\rho_n\|\bar P\|_\infty\|O_1\|_{2\to2}
       +4\omega\|O_2\|_{2\to2}n\rho_n^2\|\bar P\|_\infty^2.
\]
Define the leave-one-out estimator by
\[
    R_\omega^{(i)}:=M_\omega+B_{\mathrm{det}}+S(E^{(i)}),
\]
with $E^{(i)}$ as in \eqref{eq:loo-defs}. Let
$\widehat U_\omega^{(i)},\widehat V_\omega^{(i)}$ be its top-$r_\omega$
singular vectors;
 $\widehat U_\omega, \widehat V_\omega$ are those of
$R_\omega$, and $U_\omega, V_\omega$ those of $M_\omega$. The \emph{effective
design at node $i$} is the $\mcF_i$-measurable matrix
\begin{equation}
    \label{eq:loo2-effective}
    \mcO^{(i)} := \begin{bmatrix} O_1 \\ \omega\, O_2 (P + E^{(i)}) \end{bmatrix}
    \in \mathbb{R}^{(L_1 + L_2) \times n},\qquad
    \bar{\mcO} := \begin{bmatrix} O_1 \\ \omega\, O_2 P \end{bmatrix}.
\end{equation}
as in \Cref{subsec:fanout}. The scalars $\Omega_\omega, \varepsilon_\omega,
b_\infty^\omega, b_\omega$ are those of
\eqref{eq:fanout-Omega}--\eqref{eq:bomega}, and $\Gamma_\omega := \| N_\omega^\top
\widehat U_\omega \|_{2\to\infty}$.

\begin{lemma}[Good event, stacked]
    \label{loo2:event}
    Let $\mcE$ be the event of \Cref{loo:event} and let
    \[
        \mcB := \Big\{ \max_{j \in [n]} \| O_2 E_{\cdot j} \|_2 \leq C_B\,
        b_{O_2} \sqrt{n \rho_n \log n} \Big\},
    \]
    with $C_B = C_B(\| \bar P \|_\infty,\alpha)$. If $n\rho_n \geq C_0 \log n$ then
    \[
        \bbp{\mcE\cap\mcB}\geq1-2n^{-\alpha}-n^{-(\alpha+1)}.
    \]
    On $\mcE\cap\mcB$, for every $i,j\in[n]$ simultaneously,
    \begin{equation}
        \label{eq:loo2-event-consequences}
        \| S(E) \|_{2\to2} \vee \| S(E^{(i)}) \|_{2\to2}
        \leq C_{14}\, \Omega_\omega \sqrt{n\rho_n},\qquad
        \| \mcO^{(i)} \|_{2\to2} \leq C_{14}\, \Omega_\omega,\qquad
        \| \mcO^{(i)}_{\cdot j} \|_2 \leq C_{14}\, b_\omega.
    \end{equation}
    with $C_{14}=C_{14}(\|\bar P\|_\infty,\alpha)$.
    Writing $N_\omega^{(i)}:=R_\omega^{(i)}-M_\omega$, we also have
    \[
        \|N_\omega\|_{2\to2}\vee\|N_\omega^{(i)}\|_{2\to2}
        \leq2C_{14}\Omega_\omega\sqrt{n\rho_n},
    \]
    since
    \[
        \|B_{\mathrm{det}}\|_{2\to2}\leq C\sqrt{\rho_n}\Omega_\omega.
    \]
\end{lemma}

\begin{proof}
For each $j$, write
\[
    O_2E_{\cdot j}=\sum_k E_{kj}(O_2)_{\cdot k}.
\]
The summands are independent, centered, and bounded in norm by $b_{O_2}$.
Their variance proxy satisfies
\[
    \sum_k\Var(E_{kj})\|(O_2)_{\cdot k}\|_2^2
    \leq\rho_n\|\bar P\|_\infty\|O_2\|_F^2
    \leq\rho_n\|\bar P\|_\infty n b_{O_2}^2.
\]
Apply \Cref{lem:vector-bernstein}(i) with $\epsilon=n^{-(\alpha+2)}$ and
use $n\rho_n\geq\log n$ to absorb the $b_{O_2}\log n$ term.
A union bound over $j$ gives
\[
    \bbp{\mcB^c}\leq n^{-(\alpha+1)},\qquad
    \bbp{\mcE\cap\mcB}\geq1-2n^{-\alpha}-n^{-(\alpha+1)}.
\]
As in the proof of \Cref{lem:fanout-noise},
\[
    \|\Psi(E)\|_{2\to2}
    \leq2\bigl(2\|P\|_{2\to2}\|E\|_{2\to2}+\|E\|_{2\to2}^2\bigr),
    \qquad \|E^{(i)}\|_{2\to2}\leq\|E\|_{2\to2}.
\]
Consequently, on $\mcE$,
\[
\begin{aligned}
    \|S(E)\|_{2\to2}\vee\|S(E^{(i)})\|_{2\to2}
    &\leq C_0\|O_1\|_{2\to2}\sqrt{n\rho_n}
       +C\omega\|O_2\|_{2\to2}(n\rho_n)^{3/2}
     \leq C_{14}\Omega_\omega\sqrt{n\rho_n},\\
    \|\mcO^{(i)}\|_{2\to2}
    &\leq\|O_1\|_{2\to2}
       +\omega\|O_2\|_{2\to2}(\|P\|_{2\to2}+\|E\|_{2\to2})
     \leq C_{14}\Omega_\omega.
\end{aligned}
\]
For the column bounds, use
\[
    \mcO^{(i)}_{\cdot j}
    =\begin{bmatrix}(O_1)_{\cdot j}\\
       \omega(O_2P_{\cdot j}+O_2E^{(i)}_{\cdot j})\end{bmatrix},
    \qquad
    \|O_2E^{(i)}_{\cdot j}\|_2\leq\|O_2E_{\cdot j}\|_2+b_{O_2}.
\]
For $j\ne i$, the second inequality follows because only entry $i$ of the
column is removed; for $j=i$, the left side is zero. Thus, on $\mcB$,
\[
    \|\mcO^{(i)}_{\cdot j}\|_2
    \leq b_{O_1}+\omega\max_j\|O_2P_{\cdot j}\|_2
       +\omega C_Bb_{O_2}\sqrt{n\rho_n\log n}+\omega b_{O_2}
    \leq C_{14}b_\omega.
\]
Finally, the deterministic correction satisfies
\[
\begin{aligned}
    \|B_{\mathrm{det}}\|_{2\to2}
    &\leq\rho_n\|\bar P\|_\infty\|O_1\|_{2\to2}
       +4\omega\|O_2\|_{2\to2}n\rho_n^2\|\bar P\|_\infty^2\\
    &\leq C\sqrt{\rho_n}\Omega_\omega.
\end{aligned}
\]
Adding this bound proves the assertion for $N_\omega$ and $N_\omega^{(i)}$.
\end{proof}

\begin{lemma}[Column decoupling]
    \label{loo2:column}
    Write $u := E_{\cdot i}$. For every $i$,
    \begin{equation}
        \label{eq:loo2-column}
        (N_\omega)_{\cdot i} \;=\; \mcO^{(i)} u \;+\; r_i,
        \qquad
        r_i := \begin{bmatrix}
        -P^0_{ii}(O_1)_{\cdot i} \\[2pt]
        \begin{aligned}
        &\omega O_2E^{(i)}P_{\cdot i}
        -\omega(P_{i\cdot}u)(O_2)_{\cdot i}\\
        &\quad+\omega[O_2\{\hol(P^2)-(P^0)^2\}]_{\cdot i}
        \end{aligned}
        \end{bmatrix},
    \end{equation}
    On $\mcE$, uniformly in $i$,
    \[
        \|r_i\|_2\leq C_7\sqrt{\rho_n}\Omega_\omega,
        \qquad C_7=C_7(\|\bar P\|_\infty,\alpha).
    \]
    Conditional on $\mcF_i$, $\mcO^{(i)}$ is fixed and $u$ is independent
    of $\mcF_i$. The remainder $r_i$ can still depend on $u$ through
    \[
        -\omega(P_{i\cdot}u)(O_2)_{\cdot i}.
    \]
    We use only the displayed norm bound for $r_i$.
\end{lemma}

\begin{proof}
Since $u=E_{\cdot i}$ and $P_{ii}=u_i=0$,
\[
    E=E^{(i)}+ue_i^\top+e_iu^\top,\qquad
    EP_{\cdot i}=E^{(i)}P_{\cdot i}+(u^\top P_{\cdot i})e_i,\qquad
    Eu=E^{(i)}u+\|u\|_2^2e_i.
\]
The three terms in column $i$ of $\Psi(E)$ are therefore
\[
\begin{aligned}
    \hol(PE)_{\cdot i}&=Pu-(P_{i\cdot}u)e_i,\\
    \hol(EP)_{\cdot i}&=EP_{\cdot i}-(E_{i\cdot}P_{\cdot i})e_i
        =E^{(i)}P_{\cdot i},\\
    \hol(E^2)_{\cdot i}&=Eu-\|u\|_2^2e_i=E^{(i)}u.
\end{aligned}
\]
Thus setting the diagonal to zero removes the term quadratic in $u$, and
\[
    \Psi(E)_{\cdot i}
    =(P+E^{(i)})u+E^{(i)}P_{\cdot i}-(P_{i\cdot}u)e_i.
\]
Combining this identity with the first block and the deterministic correction,
\[
    (N_1)_{\cdot i}=O_1u-P^0_{ii}(O_1)_{\cdot i},\qquad
    \omega\bigl[O_2\{\hol(P^2)-(P^0)^2\}\bigr]_{\cdot i},
\]
gives \eqref{eq:loo2-column}. On $\mcE$, the terms in the remainder satisfy
\[
\begin{aligned}
    \|O_2E^{(i)}P_{\cdot i}\|_2
    &\leq\|O_2\|_{2\to2}\|E\|_{2\to2}\|P_{\cdot i}\|_2
     \leq C\|O_2\|_{2\to2}n\rho_n\sqrt{\rho_n},\\
    |P_{i\cdot}u|
    &\leq\|P_{i\cdot}\|_2\|u\|_2
     \leq Cn\rho_n\sqrt{\rho_n}.
\end{aligned}
\]
Since $\omega n\rho_n\|O_2\|_{2\to2}\leq\Omega_\omega$, these terms contribute
at most $C\sqrt{\rho_n}\Omega_\omega$ to $\|r_i\|_2$. The first block and the
remaining deterministic term are bounded by $\|B_{\mathrm{det}}\|_{2\to2}$,
which has the same bound by \Cref{loo2:event}.
\end{proof}

\begin{lemma}[Structure of the coupling, stacked]
    \label{loo2:structure}
    For every $i$, with $u = E_{\cdot i}$,
    \begin{equation}
        \label{eq:loo2-rank4}
        \begin{aligned}
        \Delta^\omega_i := R_\omega-R_\omega^{(i)} &= \Lambda_i-\Theta_i,
        \qquad \rk(\Lambda_i)\leq4,\\
        \Lambda_i &:= g_ie_i^\top+h_iu^\top
        +\omega\binom{0}{O_2u}P_{i\cdot}
        +\omega\binom{0}{(O_2)_{\cdot i}}(Pu+E^{(i)}u)^\top.
        \end{aligned}
    \end{equation}
    where
    \[
        \begin{aligned}
        g_i&:=\binom{O_1u}{\omega O_2((P+E^{(i)})u-2(P_{i\cdot}u)e_i)},\\
        h_i&:=\binom{(O_1)_{\cdot i}}{\omega O_2(P_{\cdot i}+u)},
        \qquad
        \Theta_i:=\omega\binom{0}{O_2(2D_{Pu}+D_{u^2})}.
        \end{aligned}
    \]
    Here
    \[
        D_{Pu}:=\diag(P_{ji}u_j)_j,\qquad D_{u^2}:=\diag(u_j^2)_j,
    \]
    and
    \[
        \|\Theta_i\|_{2\to2}
        \leq\omega\|O_2\|_{2\to2}(1+2\rho_n\|\bar P\|_\infty)
        \leq\frac{C\Omega_\omega}{n\rho_n}.
    \] Moreover $R_\omega^{(i)}$,
    and hence $\widehat U_\omega^{(i)}$, is $\mcF_i$-measurable, independent of
    $u$.
\end{lemma}

\begin{proof}
Write $E-E^{(i)}=ue_i^\top+e_iu^\top$. The first block of
$S(E)-S(E^{(i)})$ is
\[
    O_1(E-E^{(i)})=(O_1u)e_i^\top+(O_1)_{\cdot i}u^\top.
\]
For the second block,
\[
\begin{aligned}
    P(E-E^{(i)})&=(Pu)e_i^\top+P_{\cdot i}u^\top,\\
    (E-E^{(i)})P&=uP_{i\cdot}+e_i(Pu)^\top,\\
    E^2-(E^{(i)})^2
    &=E(E-E^{(i)})+(E-E^{(i)})E^{(i)}\\
    &=(Eu)e_i^\top+uu^\top+e_i(E^{(i)}u)^\top.
\end{aligned}
\]
Here $Ee_i=u$ and $e_i^\top E^{(i)}=0$. The diagonal matrices removed from
these three expressions are
\[
\begin{aligned}
    \diag\{P(E-E^{(i)})\}&=(P_{i\cdot}u)e_ie_i^\top+D_{Pu},\\
    \diag\{(E-E^{(i)})P\}&=(P_{i\cdot}u)e_ie_i^\top+D_{Pu},\\
    \diag\{E^2-(E^{(i)})^2\}&=\|u\|_2^2e_ie_i^\top+D_{u^2}.
\end{aligned}
\]
Substitute $Eu=E^{(i)}u+\|u\|_2^2e_i$. Collecting the terms with
$e_i^\top$ gives $g_ie_i^\top$, and collecting those with $u^\top$ gives
$h_iu^\top$. In particular, the correction in $g_i$ simplifies to
\[
    -(2P_{i\cdot}u+\|u\|_2^2)e_i+\|u\|_2^2e_i
    =-2(P_{i\cdot}u)e_i.
\]
The remaining terms give \eqref{eq:loo2-rank4}. Finally, since $|u_j|\leq1$,
\[
\begin{aligned}
    \|\Theta_i\|_{2\to2}
    &\leq\omega\|O_2\|_{2\to2}\max_j|2P_{ji}u_j+u_j^2|\\
    &\leq\omega\|O_2\|_{2\to2}\max_j(2P_{ji}|u_j|+u_j^2)\\
    &\leq\omega\|O_2\|_{2\to2}(2\rho_n\|\bar P\|_\infty+1)
     \leq\frac{C\Omega_\omega}{n\rho_n},
\end{aligned}
\]
where the last inequality uses
\[
    \omega\|O_2\|_{2\to2}
    \leq\frac{\Omega_\omega}{(1+\|\bar B\|_\infty)n\rho_n}.
\]
The measurability and independence follow as in \Cref{loo:structure}.
\end{proof}

With \Cref{loo2:event,loo2:column,loo2:structure} in place, we adapt the argument
of \Cref{sec:procrustes}. The objects that replace their linear counterparts are
\[
    O \longmapsto \mcO^{(i)},\qquad
    OE_{\cdot i} \longmapsto (N_\omega)_{\cdot i},\qquad
    b_O \longmapsto b_\omega,\qquad
    \|O\|_{2\to2} \longmapsto \Omega_\omega,\qquad
    b_\infty\longmapsto b_\infty^\omega+b_{\infty,2}^\omega,\qquad
    g_1(O)\longmapsto g_\omega.
\]
Three features prevent a literal invocation of
\Cref{loo:stability-full,loo:concentration-full}: the effective design
$\mcO^{(i)}$ is random but $\mcF_i$-measurable, the coupling contains the
diagonal remainder $\Theta_i$ in addition to the rank-four term $\Lambda_i$, and
the residuals of $\Lambda_i$ involve three second-order row quantities.
\Cref{loo2:stability,loo2:leverage,loo2:concentration} treat these features.

Within this subsection write $d:=n\rho_n$, $\Omega:=\Omega_\omega$,
$\varepsilon:=\varepsilon_\omega$, $\tau:=\tau_\omega$,
$\sigma:=\sigma_{r_\omega}(M_\omega)$, $h:=\|V_\omega\|_{2\to\infty}$,
$\gamma_1:=\sqrt{\rho_n}(\sqrt{r_\omega}+\sqrt{\log n})$ and $u_i:=E_{\cdot i}$,
and let $\widehat U_{\omega,2}$ denote the last $L_2$ rows of $\widehat U_\omega$.
Besides $\Gamma_\omega$, the row-noise quantities are
\begin{equation}
    \label{eq:loo2-scalars}
    \Gamma_{V,\omega}:=\|E\widehat V_\omega\|_{2\to\infty}
    =\max_i\|\widehat V_\omega^\top u_i\|_2,
    \qquad
    \Upsilon_\omega:=\Gamma_\omega+b_\omega\Gamma_{V,\omega},
\end{equation}
and the three auxiliary quantities
\begin{equation}
    \label{eq:loo2-aux}
    \Gamma_{2,\omega}:=\|EO_2^\top\widehat U_{\omega,2}\|_{2\to\infty},
    \qquad
    \Gamma_{P,\omega}:=\|EP\widehat V_\omega\|_{2\to\infty},
    \qquad
    \Gamma_{E,\omega}:=\|E^2\widehat V_\omega\|_{2\to\infty}.
\end{equation}
Since row $i$ of $E$ is $u_i^\top$, these are the maxima over $i$ of
$\|\widehat U_{\omega,2}^\top O_2u_i\|_2$, $\|\widehat V_\omega^\top Pu_i\|_2$ and
$\|\widehat V_\omega^\top Eu_i\|_2$. The estimated leverage and incoherence are
\[
    \widehat b_\omega:=\max_{i',j}\|\widehat U_\omega^\top\mcO^{(i')}_{\cdot j}\|_2,
    \qquad
    \widehat b_{0,\omega}:=\|\bar{\mcO}^\top\widehat U_\omega\|_{2\to\infty},
    \qquad
    \widehat b_{2,\omega}:=\|O_2^\top\widehat U_{\omega,2}\|_{2\to\infty},
    \qquad
    \widehat h_\omega:=\|\widehat V_\omega\|_{2\to\infty},
\]
with leave-one-out versions
$b_\omega^{(i)}:=\|\mcO^{(i)\top}\widehat U_\omega^{(i)}\|_{2\to\infty}$,
$b_{2,\omega}^{(i)}:=\|O_2^\top\widehat U_{\omega,2}^{(i)}\|_{2\to\infty}$ and
$h_\omega^{(i)}:=\|\widehat V_\omega^{(i)}\|_{2\to\infty}$. The residuals of the
coupling of \Cref{loo2:structure} are
$\Phi_i:=\|(\Delta_i^\omega)^\top\widehat U_\omega\|_F
+\|\Delta_i^\omega\widehat V_\omega\|_F$ and $\Phi:=\max_i\Phi_i$, and the two
smallness parameters are
\[
    \lambda_\omega:=\frac{b_\omega\log n}{\sigma},\qquad
    \lambda_{\varsigma,\omega}:=\frac{\varsigma_\omega\Omega\log n}{\sigma}
    =\frac{g_\omega\log n}{b_\omega\sigma}.
\]
Throughout we use
\begin{equation}
    \label{eq:loo2-weights}
    \omega d\leq\tau,\qquad
    \omega d\|O_2\|_{2\to2}\leq\Omega,\qquad
    \tau b_{O_2}\leq b_\omega\leq C_b\Omega,\qquad C_b:=3+\|\bar P\|_\infty.
\end{equation}
The first two follow from the definitions of $\tau$ and $\Omega$, and
$\tau b_{O_2}\leq\Omega b_{O_2}/\|O_2\|_{2\to2}\leq b_\omega$. For the last, the four
pieces of \eqref{eq:bomega} satisfy $b_{O_1}\leq\|O_1\|_{2\to2}\leq\Omega$,
$\omega\|O_2P_{\cdot j}\|_2\leq\omega\|O_2\|_{2\to2}\rho_n\sqrt n\|\bar P\|_\infty
\leq\|\bar P\|_\infty\Omega$, $\omega b_{O_2}\sqrt{d\log n}\leq\omega d\|O_2\|_{2\to2}
\leq\Omega$ (as $\log n\leq d$) and $\Omega b_{O_2}/\|O_2\|_{2\to2}\leq\Omega$. In \Cref{loo2:stability,loo2:leverage,loo2:concentration}
and their proofs we drop the subscript $\omega$ from $\widehat U_\omega$,
$\widehat V_\omega$, $U_\omega$, $V_\omega$ and from the quantities above.

\begin{lemma}[Leave-one-out stability, stacked]
    \label{loo2:stability}
    On $\mcE\cap\mcB$, suppose $\varepsilon\leq1/(16C_{14})$. Then for every
    $i\in[n]$ there are orthogonal $Q_i,Q_i'\in\mathbb R^{r_\omega\times r_\omega}$ with
    \begin{equation}
        \label{eq:loo2-stability}
        \|\widehat U-\widehat U^{(i)}Q_i\|_F\vee\|\widehat V-\widehat V^{(i)}Q_i'\|_F
        \leq\frac{2\sqrt2}{\sigma}\Phi_i,
    \end{equation}
    \begin{equation}
        \label{eq:loo2-residual}
        \Phi_i\leq\Gamma+Cb_\omega\Gamma_V
        +C\sqrt d\,\big(\widehat b+\omega\Gamma_2+\tau\widehat b_2\big)
        +C\Omega\sqrt d\,\widehat h
        +\omega b_{O_2}(\Gamma_P+\Gamma_E)+C\Omega\sqrt{\rho_n}.
    \end{equation}
    Consequently, for every $i\in[n]$,
    \begin{equation}
        \label{eq:loo2-stability-rows}
        b^{(i)}\leq\widehat b+\frac{2\sqrt2\,C_{14}b_\omega}{\sigma}\Phi,
        \qquad
        b_2^{(i)}\leq\widehat b_2+\frac{2\sqrt2\,b_{O_2}}{\sigma}\Phi,
        \qquad
        h^{(i)}\leq\widehat h+\frac{2\sqrt2}{\sigma}\Phi.
    \end{equation}
    The constant $C$ depends only on $\|\bar P\|_\infty$ and $\alpha$.
\end{lemma}

\begin{proof}
    \emph{Gap.} On $\mcE\cap\mcB$, \Cref{loo2:event} gives
    $\|N\|_{2\to2}\vee\|N^{(i)}\|_{2\to2}\leq2C_{14}\varepsilon\sigma\leq\sigma/8$,
    so Weyl's inequality gives
    $\sigma_{r_\omega}(R_\omega)-\sigma_{r_\omega+1}(R_\omega^{(i)})
    \geq\sigma-4C_{14}\varepsilon\sigma\geq\tfrac34\sigma$.
    \Cref{loo:wedin-res} with $R_\omega^{(i)}$ and $\Delta_i^\omega$ in the roles of
    $M$ and $N$, so that $M+N=R_\omega$, then bounds both Procrustes distances by
    $\sqrt2\,\Phi_i/(3\sigma/4)\leq2\sqrt2\,\Phi_i/\sigma$, which is
    \eqref{eq:loo2-stability}.

    \emph{Left residual.} Write $u=u_i$. Transposing \eqref{eq:loo2-rank4},
    \[
        \Lambda_i^\top\widehat U
        =e_i(\widehat U^\top g_i)^\top+u(\widehat U^\top h_i)^\top
        +\omega P_{i\cdot}^\top(\widehat U_2^\top O_2u)^\top
        +\omega(P+E^{(i)})u\,\big((O_2)_{\cdot i}^\top\widehat U_2\big),
    \]
    so $\|\Lambda_i^\top\widehat U\|_F$ is at most
    $\|\widehat U^\top g_i\|_2+\|u\|_2\|\widehat U^\top h_i\|_2
    +\omega\|P_{i\cdot}\|_2\|\widehat U_2^\top O_2u\|_2
    +\omega\|(P+E^{(i)})u\|_2\|\widehat U_2^\top(O_2)_{\cdot i}\|_2$.
    By \Cref{loo2:column} and the definition of $g_i$,
    $g_i=N_{\cdot i}-r_i-2\omega(P_{i\cdot}u)\binom{0}{(O_2)_{\cdot i}}$. On $\mcE$,
    $\|u\|_2\leq C_{10}\sqrt d$ and
    $|P_{i\cdot}u|\leq\|P_{i\cdot}\|_2\|u\|_2\leq C d\sqrt{\rho_n}$, so
    \eqref{eq:loo2-weights} and $\|r_i\|_2\leq C_7\sqrt{\rho_n}\Omega$ give
    $\|\widehat U^\top g_i\|_2\leq\Gamma+C\Omega\sqrt{\rho_n}$.
    Since $E^{(i)}_{\cdot i}=0$, the second block of $\mcO^{(i)}_{\cdot i}$ is
    $\omega O_2P_{\cdot i}$, so $h_i=\mcO^{(i)}_{\cdot i}+\omega\binom{0}{O_2u}$ and
    $\|\widehat U^\top h_i\|_2\leq\widehat b+\omega\Gamma_2$. The third term is at
    most $\omega\rho_n\sqrt n\|\bar P\|_\infty\|O_2\|_{2\to2}C_{10}\sqrt d
    \leq C\Omega\sqrt{\rho_n}$ by \eqref{eq:loo2-weights}. For the fourth,
    $\|(P+E^{(i)})u\|_2\leq(\|P\|_{2\to2}+\|E\|_{2\to2})\|u\|_2\leq Cd^{3/2}$ and
    $\|\widehat U_2^\top(O_2)_{\cdot i}\|_2\leq\widehat b_2$, so it is at most
    $C\omega d^{3/2}\widehat b_2\leq C\tau\sqrt d\,\widehat b_2$. Finally
    $\Theta_i^\top\widehat U=\omega D_iO_2^\top\widehat U_2$ with the diagonal
    matrix $D_i:=2D_{Pu}+D_{u^2}$. Each row of $D_iX$ is a multiple of the
    corresponding row of $X$, so
    $\|D_iO_2^\top\widehat U_2\|_F\leq\|D_i\|_F\,\widehat b_2$, and
    $|(D_i)_{jj}|\leq(1+2\rho_n\|\bar P\|_\infty)|u_j|$ gives
    $\|D_i\|_F\leq C\sqrt d$. Hence
    $\|\Theta_i^\top\widehat U\|_F\leq C\omega\sqrt d\,\widehat b_2
    \leq C\tau\sqrt d\,\widehat b_2$, as $\omega\leq\omega d\leq\tau$.

    \emph{Right residual.} Similarly,
    \[
        \Lambda_i\widehat V
        =g_i\widehat V_{i\cdot}+h_i(u^\top\widehat V)
        +\omega\binom{0}{O_2u}(P_{i\cdot}\widehat V)
        +\omega\binom{0}{(O_2)_{\cdot i}}\big(u^\top(P+E^{(i)})\widehat V\big).
    \]
    Here $\|g_i\|_2\leq\|\mcO^{(i)}u\|_2+2\omega|P_{i\cdot}u|\,b_{O_2}\leq C\Omega\sqrt d$,
    so the first term is at most $C\Omega\sqrt d\,\widehat h$. On $\mcB$,
    $\|h_i\|_2\leq C_{14}b_\omega+\omega C_Bb_{O_2}\sqrt{d\log n}\leq Cb_\omega$ by
    \eqref{eq:bomega}, so the second term is at most $Cb_\omega\Gamma_V$. The third
    is at most $\omega\|O_2\|_{2\to2}C_{10}\sqrt d\,\rho_n\sqrt n\|\bar P\|_\infty
    \leq C\Omega\sqrt{\rho_n}$. For the fourth, $u^\top P\widehat V$ is row $i$ of
    $EP\widehat V$, so $\|\widehat V^\top Pu\|_2\leq\Gamma_P$. Since
    $E-E^{(i)}=ue_i^\top+e_iu^\top$ and $u_i=0$, we have
    $u^\top E^{(i)}=u^\top E-\|u\|_2^2e_i^\top$, whence
    $\|\widehat V^\top E^{(i)}u\|_2\leq\|(E^2\widehat V)_{i\cdot}\|_2
    +\|u\|_2^2\widehat h\leq\Gamma_E+C_{10}^2d\,\widehat h$. The fourth term is
    therefore at most $\omega b_{O_2}(\Gamma_P+\Gamma_E)+C\omega db_{O_2}\widehat h
    \leq\omega b_{O_2}(\Gamma_P+\Gamma_E)+C\Omega\widehat h$ by
    \eqref{eq:loo2-weights}. Finally
    $\|\Theta_i\widehat V\|_F=\omega\|O_2D_i\widehat V\|_F
    \leq\omega\|O_2\|_{2\to2}\|D_i\|_F\,\widehat h\leq C\Omega\widehat h/\sqrt d$.
    Adding the two residual bounds and using $\Omega\widehat h\leq\Omega\sqrt d\,\widehat h$
    gives \eqref{eq:loo2-residual}.

    \emph{Rows.} As in the proof of \Cref{loo:stability-full}, for every $j$,
    $\|\widehat U^{(i)\top}\mcO^{(i)}_{\cdot j}\|_2
    \leq\|\widehat U^\top\mcO^{(i)}_{\cdot j}\|_2
    +\|\widehat U^{(i)}Q_i-\widehat U\|_{2\to2}\|\mcO^{(i)}_{\cdot j}\|_2$, and
    $\|\mcO^{(i)}_{\cdot j}\|_2\leq C_{14}b_\omega$ by
    \eqref{eq:loo2-event-consequences}; this gives the bound on $b^{(i)}$. The
    same argument with $(O_2)_{\cdot j}$ and the last $L_2$ rows gives the bound on
    $b_2^{(i)}$, and $h^{(i)}\leq\widehat h
    +\|\widehat V-\widehat V^{(i)}Q_i'\|_F$ gives the last. Take $\Phi_i\leq\Phi$.
\end{proof}

\begin{lemma}[Leverage, stacked]
    \label{loo2:leverage}
    On $\mcE\cap\mcB$, suppose $\varepsilon\leq1/(16C_{14})$. Then
    \begin{equation}
        \label{eq:loo2-leverage-split}
        \widehat b\leq\widehat b_0+\omega(\Gamma_2+\widehat b_2),
        \qquad
        \widehat h\leq\frac{4\Gamma}{3\sigma}+3h,
    \end{equation}
    and $\widehat B:=\widehat b_0+\tau\widehat b_2$ satisfies
    \begin{equation}
        \label{eq:loo2-leverage}
        \widehat B\leq2\big(b_\infty^\omega+b_{\infty,2}^\omega\big)
        +\frac{C}{\sigma}\min\Big\{b_\omega\varepsilon\sigma,\;
        g_\omega\Big(\Gamma_V+\widehat h+\frac{\Gamma_P+\Gamma_E}{d}\Big)
        +b_\omega\sqrt{\rho_n}\,\Omega\Big\},
    \end{equation}
    where $C$ depends only on $\|\bar P\|_\infty$ and $\alpha$.
\end{lemma}

\begin{proof}
    \emph{Splitting.} For $j\neq i'$, the second block of $\mcO^{(i')}_{\cdot j}$ is
    $\omega O_2(P_{\cdot j}+E^{(i')}_{\cdot j})$ with
    $E^{(i')}_{\cdot j}=u_j-E_{i'j}e_{i'}$, while $\mcO^{(i')}_{\cdot i'}=\bar{\mcO}_{\cdot i'}$.
    Hence, using $|E_{i'j}|\leq1$,
    $\|\widehat U^\top\mcO^{(i')}_{\cdot j}\|_2\leq\|\widehat U^\top\bar{\mcO}_{\cdot j}\|_2
    +\omega\|\widehat U_2^\top O_2u_j\|_2+\omega\|\widehat U_2^\top(O_2)_{\cdot i'}\|_2
    \leq\widehat b_0+\omega(\Gamma_2+\widehat b_2)$. The bound on $\widehat h$ is
    \eqref{eq:loo-incoherence}, whose derivation from \Cref{lem:master} uses only
    $\|N\|_{2\to2}\leq\sigma/8$ and $\|\sin\Theta(\widehat V,V)\|_{2\to2}\leq1/4$;
    both hold on $\mcE\cap\mcB$ by \Cref{loo2:event} and \Cref{lem:wedin}.

    \emph{Reduction.} The proof of \Cref{lem:master-left} uses $O$ only as a fixed
    left factor. For any fixed $G\in\mathbb R^{(L_1+L_2)\times n}$ it therefore
    gives, with the Procrustes alignment $Q_U$ of $\widehat U$ to $U$,
    \[
        \|G^\top(\widehat UQ_U-U)\|_{2\to\infty}
        \leq\frac{4}{3\sigma}\Big(\|G^\top N\widehat V\|_{2\to\infty}
        +\|G^\top U\|_{2\to\infty}\|N\|_{2\to2}\Big)
        +\|G^\top U\|_{2\to\infty}\|\sin\Theta(U,\widehat U)\|_{2\to2}^2.
    \]
    On $\mcE\cap\mcB$, $\|N\|_{2\to2}\leq\sigma/8$ and
    $\|\sin\Theta(U,\widehat U)\|_{2\to2}\leq1/4$, so
    $\|G^\top\widehat U\|_{2\to\infty}\leq2\|G^\top U\|_{2\to\infty}
    +\frac{4}{3\sigma}\|G^\top N\widehat V\|_{2\to\infty}$. Take $G=\bar{\mcO}$ and
    $G=\tau\binom{0}{O_2}$, for which $\|G^\top U\|_{2\to\infty}$ equals
    $b_\infty^\omega$ and $b_{\infty,2}^\omega$. Writing
    $N_2:=\omega O_2\Psi(E)+\omega O_2\{\hol(P^2)-(P^0)^2\}$ for the second block
    of $N$ and adding,
    \[
        \widehat B\leq2\big(b_\infty^\omega+b_{\infty,2}^\omega\big)
        +\frac{4}{3\sigma}\Big(\|\bar{\mcO}^\top N\widehat V\|_{2\to\infty}
        +\tau\|O_2^\top N_2\widehat V\|_{2\to\infty}\Big).
    \]

    \emph{First branch.} $\|G^\top N\widehat V\|_{2\to\infty}\leq
    \max_j\|G_{\cdot j}\|_2\,\|N\|_{2\to2}$. Since
    $\|\bar{\mcO}_{\cdot j}\|_2\leq b_{O_1}+\omega\|O_2P_{\cdot j}\|_2\leq b_\omega$ and
    $\tau b_{O_2}\leq b_\omega$, each of the two terms is at most
    $2C_{14}b_\omega\varepsilon\sigma$.

    \emph{Second branch.} Since $\bar{\mcO}^\top=[O_1^\top,\ \omega PO_2^\top]$,
    \[
        \bar{\mcO}^\top N=O_1^\top O_1E+\omega^2PO_2^\top O_2\Psi(E)
        +\bar{\mcO}^\top B_{\mathrm{det}}.
    \]
    Row $j$ of $O_1^\top O_1E\widehat V$ is $\sum_k(O_1^\top O_1)_{jk}(E\widehat V)_{k\cdot}$,
    of norm at most $g_1(O_1)\Gamma_V$. Row $j$ of $PO_2^\top O_2\Psi(E)\widehat V$ has
    norm at most $\|(PO_2^\top O_2)_{j\cdot}\|_1\,\|\Psi(E)\widehat V\|_{2\to\infty}$, and
    $\|(PO_2^\top O_2)_{j\cdot}\|_1=\|O_2^\top O_2P_{\cdot j}\|_1
    \leq g_1(O_2)\|P_{\cdot j}\|_1\leq g_1(O_2)\,d\|\bar P\|_\infty$. For the last
    factor, $\Psi(E)=PE+EP+E^2-\diag(PE+EP+E^2)$, the diagonal entries are at most
    $2|P_{j\cdot}u_j|+\|u_j\|_2^2\leq Cd$ on $\mcE$, and
    $\|P_{j\cdot}\|_1\leq d\|\bar P\|_\infty$, so
    \[
        \|\Psi(E)\widehat V\|_{2\to\infty}
        \leq d\|\bar P\|_\infty\Gamma_V+\Gamma_P+\Gamma_E+Cd\,\widehat h.
    \]
    Also $\|\bar{\mcO}^\top B_{\mathrm{det}}\widehat V\|_{2\to\infty}
    \leq\max_j\|\bar{\mcO}_{\cdot j}\|_2\|B_{\mathrm{det}}\|_{2\to2}
    \leq Cb_\omega\sqrt{\rho_n}\,\Omega$ by \Cref{loo2:event}. Using
    $\omega^2d^2\leq\tau^2$ and $\omega^2d=(\omega d)^2/d\leq\tau^2/d$,
    \[
        \|\bar{\mcO}^\top N\widehat V\|_{2\to\infty}
        \leq g_1(O_1)\Gamma_V
        +C\tau^2g_1(O_2)\Big(\Gamma_V+\widehat h+\frac{\Gamma_P+\Gamma_E}{d}\Big)
        +Cb_\omega\sqrt{\rho_n}\,\Omega.
    \]
    Similarly, rows of $\omega O_2^\top O_2\Psi(E)\widehat V$ have norm at most
    $\omega g_1(O_2)\|\Psi(E)\widehat V\|_{2\to\infty}$. Since
    $\hol(P^2)-(P^0)^2=-P^0D_0-D_0P^0+D_0^2-\diag(P^2)$ by the proof of
    \Cref{lem:fanout-pop}, where the rows of $P^0D_0$ and $D_0P^0$ have $\ell_1$ norm
    at most $d\rho_n\|\bar P\|_\infty^2$ and
    $\|D_0^2\|_{2\to2}\vee\|\diag(P^2)\|_{2\to2}\leq d\rho_n\|\bar P\|_\infty^2$, each
    row of $\hol(P^2)-(P^0)^2$ has $\ell_1$ norm at most $4d\rho_n\|\bar P\|_\infty^2$,
    so rows of $\omega O_2^\top O_2\{\hol(P^2)-(P^0)^2\}\widehat V$ have norm at most
    $C\omega g_1(O_2)d\rho_n\widehat h$. Using $\tau\omega d\leq\tau^2$ and
    $\tau\omega\leq\tau^2/d$,
    \[
        \tau\|O_2^\top N_2\widehat V\|_{2\to\infty}
        \leq C\tau^2g_1(O_2)\Big(\Gamma_V+\widehat h+\frac{\Gamma_P+\Gamma_E}{d}\Big).
    \]
    Adding the last two displays and using $g_1(O_1)+\tau^2g_1(O_2)=g_\omega$ gives
    the second branch of \eqref{eq:loo2-leverage}.
\end{proof}

\begin{lemma}[Row concentration, stacked]
    \label{loo2:concentration}
    There are constants $c_\star^\omega\in(0,1)$ and $C_8$, depending only on
    $\|\bar P\|_\infty$ and $\alpha$, such that the following holds. Suppose
    \begin{enumerate}
        \item[\textup{(H$_\omega$)}] $n\rho_n\geq C_0\log n$,
        $\varepsilon\leq c_\star^\omega$ and $\lambda_\omega\leq c_\star^\omega$.
    \end{enumerate}
    Then with probability at least $1-4n^{-\alpha}$,
    \begin{equation}
        \label{eq:loo2-gamma}
        \Upsilon=\Gamma+b_\omega\Gamma_V\leq C_8\Big[\Omega\gamma_1
        +\big(b_\infty^\omega+b_{\infty,2}^\omega+b_\omega h+\vartheta\big)\log n
        +\varepsilon^2\sigma\Big(h+\frac{b_\infty^\omega+b_{\infty,2}^\omega+\vartheta}{\Omega}\Big)\Big]
    \end{equation}
    with $\vartheta=b_\omega\varepsilon$. If in addition
    \begin{enumerate}
        \item[\textup{(H$'_\omega$)}] $\varsigma_\omega\varepsilon\leq c_\star^\omega$ and
        $\lambda_{\varsigma,\omega}\leq c_\star^\omega$,
    \end{enumerate}
    then \eqref{eq:loo2-gamma} also holds with $\vartheta=0$.
\end{lemma}

\begin{proof}
    The proof has four steps.
    \begin{enumerate}
        \item[]\textbf{Step 1.} Conditional Bernstein bounds for five independent
        terms, uniform in $i$, on an event $\mcA_\omega$ of probability at least
        $1-5n^{-(\alpha+1)}$.
        \item[]\textbf{Step 2.} Bounds on $\Gamma$, $\Gamma_V$ and the auxiliary
        quantities \eqref{eq:loo2-aux} in terms of the leverage and $\Phi$.
        \item[]\textbf{Step 3.} A self-bounding system for $\Upsilon$ and $\Phi$.
        \item[]\textbf{Step 4.} Absorption, with the first branch of
        \eqref{eq:loo2-leverage} under (H$_\omega$) (Step~4a) and the second under
        (H$'_\omega$) (Step~4b).
    \end{enumerate}
    Assume $c_\star^\omega\leq1/(16C_{14})$, so that
    \Cref{loo2:stability,loo2:leverage} apply on $\mcE\cap\mcB$; $c_\star^\omega$
    shrinks further in Steps~3--4. Throughout, $C$ depends only on
    $\|\bar P\|_\infty$ and $\alpha$.

    \textbf{Step 1: independent terms.} Fix $i$ and condition on $\mcF_i$. By
    \Cref{loo2:structure}, $\widehat U^{(i)}$, $\widehat V^{(i)}$ and $\mcO^{(i)}$ are
    $\mcF_i$-measurable, and $u_i$ is independent of $\mcF_i$, with independent
    centred entries bounded by one and of variance at most $\rho_n\|\bar P\|_\infty$.
    Each of
    \[
        \widehat U^{(i)\top}\mcO^{(i)}u_i,\qquad
        \widehat V^{(i)\top}u_i,\qquad
        \widehat U_2^{(i)\top}O_2u_i,\qquad
        (P\widehat V^{(i)})^\top u_i,\qquad
        (E^{(i)}\widehat V^{(i)})^\top u_i
    \]
    is therefore a conditional sum $\sum_j(u_i)_jw_j$ whose weights are the rows of
    an $\mcF_i$-measurable matrix, namely $W_1=\mcO^{(i)\top}\widehat U^{(i)}$,
    $W_2=\widehat V^{(i)}$, $W_3=O_2^\top\widehat U_2^{(i)}$, $W_4=P\widehat V^{(i)}$ and
    $W_5=E^{(i)}\widehat V^{(i)}$. Apply \Cref{lem:vector-bernstein}(ii)
    conditionally to each, with $\epsilon=n^{-(\alpha+2)}$, and remove the
    conditioning as in Step~1 of \Cref{loo:concentration-full}. A union bound over
    $i$ and the five sums gives an event $\mcA_\omega$ with
    $\bbp{\mcA_\omega^c}\leq5n^{-(\alpha+1)}$ on which, for every $i$ and $k$,
    \[
        \|W_k^\top u_i\|_2\leq C_1\Big\{\sqrt{\rho_n}\|W_k\|_F
        +\sqrt{\rho_n\log n}\,\|W_k\|_{2\to2}+\|W_k\|_{2\to\infty}\log n\Big\}.
    \]
    On $\mcE\cap\mcB$, \Cref{loo2:event} and the deterministic bounds
    $\|P\|_{2\to2}\leq d\|\bar P\|_\infty$ and
    $\max_j\|P_{j\cdot}\|_2\leq\rho_n\sqrt n\|\bar P\|_\infty$ give
    $\|W_1\|_F\leq C_{14}\sqrt{r_\omega}\Omega$ and $\|W_1\|_{2\to2}\leq C_{14}\Omega$;
    $\|W_2\|_F=\sqrt{r_\omega}$ and $\|W_2\|_{2\to2}=1$;
    $\|W_3\|_F\leq\sqrt{r_\omega}\|O_2\|_{2\to2}$ and $\|W_3\|_{2\to2}\leq\|O_2\|_{2\to2}$;
    $\|W_4\|_F\leq\sqrt{r_\omega}d\|\bar P\|_\infty$,
    $\|W_4\|_{2\to2}\leq d\|\bar P\|_\infty$ and
    $\|W_4\|_{2\to\infty}\leq\rho_n\sqrt n\|\bar P\|_\infty$; and
    $\|W_5\|_F\leq C_0\sqrt{r_\omega d}$ and $\|W_5\|_{2\to2}\leq C_0\sqrt d$. The
    remaining row norms are $b^{(i)}$, $h^{(i)}$ and $b_2^{(i)}$ for $W_1,W_2,W_3$,
    and, since row $j\neq i$ of $E^{(i)}\widehat V$ is
    $(E\widehat V)_{j\cdot}-E_{ji}\widehat V_{i\cdot}$,
    $\|W_5\|_{2\to\infty}\leq\|E^{(i)}\widehat V\|_{2\to\infty}
    +\|E^{(i)}\|_{2\to2}\|\widehat V^{(i)}Q_i'-\widehat V\|_F
    \leq\Gamma_V+\widehat h+2\sqrt2C_0\sqrt d\,\Phi/\sigma$. Since
    $N_{\cdot i}=\mcO^{(i)}u_i+r_i$ with $\|r_i\|_2\leq C_7\sqrt{\rho_n}\Omega$
    (\Cref{loo2:column}), and
    $\rho_n\sqrt n\log n=\sqrt{\rho_n}\sqrt d\log n\leq d\gamma_1$, on
    $\mcE\cap\mcB\cap\mcA_\omega$, for every $i$,
    \begin{equation}
        \label{eq:loo2-bernstein}
        \begin{gathered}
        \|\widehat U^{(i)\top}N_{\cdot i}\|_2\leq C\big(\Omega\gamma_1+b^{(i)}\log n\big),
        \qquad
        \|\widehat V^{(i)\top}u_i\|_2\leq C\big(\gamma_1+h^{(i)}\log n\big),\\
        \|\widehat U_2^{(i)\top}O_2u_i\|_2
        \leq C\big(\|O_2\|_{2\to2}\gamma_1+b_2^{(i)}\log n\big),
        \qquad
        \|\widehat V^{(i)\top}Pu_i\|_2\leq Cd\gamma_1,\\
        \|\widehat V^{(i)\top}E^{(i)}u_i\|_2
        \leq C\Big(\sqrt d\,\gamma_1+\big(\Gamma_V+\widehat h+\sqrt d\,\Phi/\sigma\big)\log n\Big).
        \end{gathered}
    \end{equation}

    \textbf{Step 2: coupling terms.} On $\mcE\cap\mcB\cap\mcA_\omega$, split each
    full-sample quantity as in \eqref{eq:loo-split}, aligned by the $Q_i$ or $Q_i'$ of
    \Cref{loo2:stability}, and bound the coupling by \eqref{eq:loo2-stability} and
    the bounds $\|N_{\cdot i}\|_2\leq C\varepsilon\sigma$,
    $\|u_i\|_2\leq C_{10}\sqrt d$, $\|O_2u_i\|_2\leq C_Bb_{O_2}\sqrt{d\log n}$,
    $\|P\|_{2\to2}\leq d\|\bar P\|_\infty$ and $\|E^{(i)}\|_{2\to2}\leq C_0\sqrt d$ of
    \Cref{loo2:event}. For $\Gamma_E$ use in addition
    $(E^2\widehat V)_{i\cdot}=u_i^\top E^{(i)}\widehat V+\|u_i\|_2^2\widehat V_{i\cdot}$,
    which follows from $u_i^\top E=u_i^\top E^{(i)}+\|u_i\|_2^2e_i^\top$. With
    \eqref{eq:loo2-bernstein} and \eqref{eq:loo2-stability-rows}, taking maxima over
    $i$ and using $\widehat h\log n\leq d\widehat h$,
    \begin{align*}
        \Gamma&\leq C\big(\Omega\gamma_1+\widehat b\log n\big)+C(\lambda_\omega+\varepsilon)\Phi,\\
        \Gamma_V&\leq C\big(\gamma_1+\widehat h\log n\big)+C(\log n+\sqrt d)\Phi/\sigma,\\
        \Gamma_2&\leq C\big(\|O_2\|_{2\to2}\gamma_1+\widehat b_2\log n\big)
            +Cb_{O_2}\big(\log n+\sqrt{d\log n}\big)\Phi/\sigma,\\
        \Gamma_P&\leq Cd\gamma_1+Cd^{3/2}\Phi/\sigma,\\
        \Gamma_E&\leq C\sqrt d\,\gamma_1+C\log n\,\Gamma_V+Cd\,\widehat h
            +C\big(d+\sqrt d\log n\big)\Phi/\sigma.
    \end{align*}
    We multiply these by the coefficients with which they enter
    \eqref{eq:loo2-residual} and \eqref{eq:loo2-leverage-split}. By
    \eqref{eq:loo2-weights}, $\omega\leq\omega d$, $\log n\leq\sqrt{d\log n}\leq d$
    and $b_\omega\sqrt d/\sigma=(b_\omega/\Omega)\varepsilon\leq C_b\varepsilon$,
    \begin{equation}
        \label{eq:loo2-aux-bounds}
        \begin{gathered}
        b_\omega\Gamma_V\leq C\Omega\gamma_1+Cb_\omega\widehat h\log n
            +C(\lambda_\omega+\varepsilon)\Phi,\\
        \omega\sqrt d\,\Gamma_2\leq C\Omega\gamma_1+C\tau\sqrt d\,\widehat b_2
            +C(\lambda_\omega+\varepsilon)\Phi,\qquad
        \omega(\log n)\Gamma_2\leq C\Omega\gamma_1+C\tau(\log n)\widehat b_2
            +C(\lambda_\omega+\varepsilon)\Phi,\\
        \omega b_{O_2}(\Gamma_P+\Gamma_E)\leq C\Omega\gamma_1+Cb_\omega\Gamma_V
            +C\Omega\widehat h+C(\lambda_\omega+\varepsilon)\Phi,\qquad
        \frac{\Gamma_P+\Gamma_E}{d}\leq C\big(\gamma_1+\Gamma_V+\widehat h\big)
            +C\sqrt d\,\Phi/\sigma.
        \end{gathered}
    \end{equation}
    For instance,
    $\omega\sqrt d\,b_{O_2}\sqrt{d\log n}/\sigma\leq b_\omega\sqrt{\log n}/\sigma
    =(b_\omega/\Omega)\varepsilon\sqrt{\log n/d}\leq C_b\varepsilon$ and
    $\omega b_{O_2}\log n\sqrt{d\log n}/\sigma
    \leq\lambda_\omega\sqrt{\log n/d}\leq\lambda_\omega$ bound the
    $\Phi$-coefficients of the two $\Gamma_2$ bounds, and
    $\omega\sqrt d\,\widehat b_2\log n\leq\tau\sqrt d\,\widehat b_2$ and
    $\omega\widehat b_2\log^2n\leq\tau\widehat b_2\log n$ bound their leverage
    terms; the other coefficients are bounded in the same way.

    \textbf{Step 3: the self-bounding system.} Insert
    \eqref{eq:loo2-aux-bounds} and the first bound of
    \eqref{eq:loo2-leverage-split} into \eqref{eq:loo2-residual} and into the
    bounds on $\Gamma$ and $b_\omega\Gamma_V$. With $\widehat B=\widehat b_0+\tau\widehat b_2$,
    $\omega\leq\tau$, $\Gamma+Cb_\omega\Gamma_V\leq C\Upsilon$ and
    $\sqrt{\rho_n}\leq\gamma_1$,
    \[
        \Phi\leq C\big[\Upsilon+\sqrt d\,\widehat B+\Omega\sqrt d\,\widehat h+\Omega\gamma_1\big]
            +C(\lambda_\omega+\varepsilon)\Phi,
        \qquad
        \Upsilon\leq C\big[\Omega\gamma_1+(\widehat B+b_\omega\widehat h)\log n\big]
            +C(\lambda_\omega+\varepsilon)\Phi.
    \]
    For $c_\star^\omega$ small the coefficient of $\Phi$ in the first inequality is at
    most $1/2$, and $\Phi$ is finite, so
    \begin{equation}
        \label{eq:loo2-system}
        \Phi\leq C\big[\Upsilon+\sqrt d\,\widehat B+\Omega\sqrt d\,\widehat h+\Omega\gamma_1\big].
    \end{equation}
    By the second bound of \eqref{eq:loo2-leverage-split} and $\Gamma\leq \Upsilon$,
    $\widehat h\leq4\Upsilon/(3\sigma)+3h$, so
    $\Omega\sqrt d\,\widehat h\leq\tfrac43\varepsilon \Upsilon+3\Omega\sqrt d\,h$ and
    $b_\omega\widehat h\log n\leq\tfrac43\lambda_\omega \Upsilon+3b_\omega h\log n$.

    \textbf{Step 4a: the branch $\vartheta=b_\omega\varepsilon$.} By the first branch of
    \eqref{eq:loo2-leverage}, $\widehat B\leq B_\star:=2(b_\infty^\omega+b_{\infty,2}^\omega)
    +Cb_\omega\varepsilon$. Inserting this and the bounds on $\widehat h$ into
    \eqref{eq:loo2-system} and the bound on $\Upsilon$,
    \[
        \Upsilon\leq C\big[\Omega\gamma_1+(B_\star+b_\omega h)\log n\big]+C\lambda_\omega \Upsilon
        +C(\lambda_\omega+\varepsilon)\big[(1+\varepsilon)\Upsilon+\sqrt d\,B_\star
        +\Omega\sqrt d\,h+\Omega\gamma_1\big].
    \]
    For $c_\star^\omega$ small the total coefficient of $\Upsilon$ is at most $1/2$; since
    $\Upsilon\leq\|N\|_{2\to2}+b_\omega\|E\|_{2\to2}<\infty$, subtract and double. Finally
    $\lambda_\omega\sqrt d=(b_\omega/\Omega)\varepsilon\log n\leq C_b\log n$,
    $\varepsilon\sqrt d=\varepsilon^2\sigma/\Omega$,
    $\lambda_\omega\Omega\sqrt d\,h=\varepsilon b_\omega h\log n\leq b_\omega h\log n$ and
    $\varepsilon\Omega\sqrt d\,h=\varepsilon^2\sigma h$, so
    \[
        \Upsilon\leq C\Big[\Omega\gamma_1+(B_\star+b_\omega h)\log n
        +\varepsilon^2\sigma\Big(h+\frac{B_\star}{\Omega}\Big)\Big],
    \]
    which is \eqref{eq:loo2-gamma} with $\vartheta=b_\omega\varepsilon$.

    \textbf{Step 4b: the branch $\vartheta=0$.} Under (H$'_\omega$), use the second
    branch of \eqref{eq:loo2-leverage} with the last bound of
    \eqref{eq:loo2-aux-bounds}, and write $g_\omega=\varsigma_\omega\Omega b_\omega$,
    so that $g_\omega\Gamma_V\leq\varsigma_\omega\Omega \Upsilon$:
    \[
        \widehat B\leq2\big(b_\infty^\omega+b_{\infty,2}^\omega\big)
        +\frac{C}{\sigma}\Big[\varsigma_\omega\Omega \Upsilon
        +\varsigma_\omega\Omega b_\omega\big(\gamma_1+\widehat h\big)
        +\varsigma_\omega\Omega b_\omega\sqrt d\,\Phi/\sigma
        +b_\omega\sqrt{\rho_n}\,\Omega\Big].
    \]
    Multiplying by $\sqrt d$ and by $\log n$, and using $\Omega\sqrt d/\sigma=\varepsilon$,
    $\varsigma_\omega\varepsilon\leq1$, $\lambda_{\varsigma,\omega}\leq1$ and
    $b_\omega\leq C_b\Omega$,
    \[
    \begin{gathered}
        \sqrt d\,\widehat B\leq2\sqrt d\big(b_\infty^\omega+b_{\infty,2}^\omega\big)
        +C\Upsilon+C\Omega\big(\gamma_1+\widehat h\big)+C\varepsilon\Phi,\\
        \widehat B\log n\leq2\big(b_\infty^\omega+b_{\infty,2}^\omega\big)\log n
        +C\lambda_{\varsigma,\omega}\Upsilon+C\Omega\gamma_1+Cb_\omega\widehat h
        +C\lambda_{\varsigma,\omega}\varepsilon\Phi.
    \end{gathered}
    \]
    Inserting the first into \eqref{eq:loo2-system} and absorbing $C\varepsilon\Phi$ gives
    $\Phi\leq C[\Upsilon+\sqrt d(b_\infty^\omega+b_{\infty,2}^\omega)+\Omega\sqrt d\,\widehat h
    +\Omega\gamma_1]$. Inserting this and the second into the bound on $\Upsilon$, together
    with the bounds on $\widehat h$ from Step~3,
    \[
    \begin{aligned}
        \Upsilon&\leq C\big[\Omega\gamma_1+\big(b_\infty^\omega+b_{\infty,2}^\omega+b_\omega h\big)\log n\big]
        +C(\lambda_\omega+\lambda_{\varsigma,\omega})\Upsilon\\
        &\qquad+C(\lambda_\omega+\varepsilon)\big[(1+\varepsilon)\Upsilon
        +\sqrt d\big(b_\infty^\omega+b_{\infty,2}^\omega\big)+\Omega\sqrt d\,h+\Omega\gamma_1\big].
    \end{aligned}
    \]
    For $c_\star^\omega$ small the coefficient of $\Upsilon$ is at most $1/2$, and the
    simplifications of Step~4a give \eqref{eq:loo2-gamma} with $\vartheta=0$.

    \textbf{Probability.} By \Cref{loo2:event} and Step~1,
    $\bbp{(\mcE\cap\mcB\cap\mcA_\omega)^c}\leq2n^{-\alpha}+6n^{-(\alpha+1)}
    \leq4n^{-\alpha}$ for $n\geq3$.
\end{proof}

\begin{theorem}[Sharp $\ell_{2\to\infty}$ bound, stacked]
    \label{loo2:twoinfty-sharp}
    Under hypothesis (H$_\omega$) of \Cref{loo2:concentration}, with probability at
    least $1 - 4n^{-\alpha}$ there is an orthogonal $Q$ with
    \begin{equation}
        \label{eq:loo2-twoinfty-sharp}
        \begin{aligned}
        \| \widehat V_\omega Q - V_\omega \|_{2\to\infty}
        &\leq C_9\,
        \frac{\sqrt{\rho_n}(\sqrt{r_\omega}+\sqrt{\log n})\,\Omega_\omega
        +\big(b_\infty^\omega+b_{\infty,2}^\omega
        +b_\omega\|V_\omega\|_{2\to\infty}+\vartheta_\omega\big)\log n}
        {\sigma_{r_\omega}(M_\omega)}\\
        &\qquad+C_9\left(\| V_\omega \|_{2\to\infty}\, \varepsilon_\omega
        +\varepsilon_\omega^2\,
        \frac{b_\infty^\omega+b_{\infty,2}^\omega+\vartheta_\omega}{\Omega_\omega}\right),
        \end{aligned}
    \end{equation}
    with $\vartheta_\omega=b_\omega\varepsilon_\omega$; under (H$'_\omega$) also with
    $\vartheta_\omega=0$. Here $C_9 = C_9(\| \bar P \|_\infty,\alpha)$.
\end{theorem}

\begin{proof}
    Work on the event of \Cref{loo2:concentration}, which is contained in $\mcE\cap\mcB$.
    \Cref{lem:master} applies to the pair $(M_\omega,N_\omega)$, since
    $\|N_\omega\|_{2\to2}\leq2C_{14}\varepsilon_\omega\sigma_{r_\omega}(M_\omega)
    \leq\sigma_{r_\omega}(M_\omega)/8$, and together with \Cref{lem:wedin} it gives
    $\|\widehat V_\omega Q-V_\omega\|_{2\to\infty}\leq4\Gamma_\omega/(3\sigma_{r_\omega}(M_\omega))
    +C\|V_\omega\|_{2\to\infty}\varepsilon_\omega$. Insert $\Gamma_\omega\leq \Upsilon_\omega$
    and \eqref{eq:loo2-gamma}, and divide each summand by
    $\sigma_{r_\omega}(M_\omega)$, using $\varepsilon_\omega^2\|V_\omega\|_{2\to\infty}
    \leq\varepsilon_\omega\|V_\omega\|_{2\to\infty}$.
\end{proof}

\begin{remark}[What changed relative to the linear case]
    \label{loo2:rem-changes}
    Compared with \Cref{sec:procrustes}, the main changes are the random
    coefficient matrix $\mcO^{(i)}$ in each noise column and the additional
    diagonal term $\Theta_i$ in the coupling. Conditional on $\mcF_i$,
    $\mcO^{(i)}$ is fixed, so \Cref{lem:vector-bernstein} applies to the linear term in
    \eqref{eq:loo2-column}. The decomposition in \Cref{loo2:structure}
    replaces the rank-two update by a rank-four term and $\Theta_i$, whose
    norm is bounded there. The resulting bounds replace
    $(\|O\|_{2\to2},b_O,b_\infty,g_1(O))$ by
    $(\Omega_\omega,b_\omega,b_\infty^\omega+b_{\infty,2}^\omega,g_\omega)$, and the
    second-order noise adds the auxiliary row quantities \eqref{eq:loo2-aux}, each
    controlled by the leave-one-out argument of \Cref{loo2:concentration}.
    The quantitative spread and near-orthogonality bounds in
    \Cref{cor:fanout-regimes} therefore give the same recovery rates as at first
    order.
\end{remark}

\subsection{Proof of the stacked master theorem and regime corollary}

\begin{proof}[Proof of \Cref{thm:fanout-main}]
    By \Cref{lem:fanout-embedding}, $M_\omega$ has rank $r_\omega$ with
    community-constant rows, incoherence $\| V_\omega \|_{2\to\infty} \leq
    n_{\min}^{-1/2}$, and centroid separation $\delta > 0$ under the stated
    distinct-columns condition. The stated hypotheses are (H$_\omega$) of
    \Cref{loo2:concentration} and, for $\vartheta_\omega=0$, (H$'_\omega$).
    \Cref{loo2:twoinfty-sharp}
    (run at the retained dimension $r_\omega$)
    therefore gives, with probability at least $1 - 4n^{-\alpha}$, the bound $\| \widehat V Q -
    V_\omega \|_{2\to\infty} \leq \eta_\infty^{(2)}$ of \eqref{eq:fanout-eta} with
    $C := C_9$, since $\|V_\omega\|_{2\to\infty}\leq h_\star$.
    On the same event, the total-noise bound and
    \Cref{lem:sin-theta} give $\|\widehat VQ-V_\omega\|_F\leq F_\omega$
    for the same orthogonal matrix $Q$. When
    $\eta_\infty^{(2)}<\delta/4$, \Cref{lem:single-linkage} proves
    exact recovery with single linkage. For approximate $k$-means,
    condition
    \eqref{eq:fanout-rounding} and $\eta_\infty^{(2)}<\delta/4$ are
    precisely the two hypotheses of \Cref{lem:kmeans}(ii), which gives
    exact recovery after the nearest-centroid step.
\end{proof}

\begin{proof}[Proof of \Cref{cor:fanout-regimes}]
    \emph{(i) Spread designs.} Since
    $b_\infty^\omega\leq\max_j\|\bar{\mcO}_{\cdot j}\|_2\leq b_\omega$, spreadness gives
    \[
        b_\infty^\omega+b_\omega\leq2C_2\Omega_\omega\sqrt{K/n}.
    \]
    The conditioning
    \eqref{eq:fanout-well-cond} gives $\sigma_{r_\omega}(M_\omega) \geq
    \sigma_{r_\omega}(H_\omega) \sqrt{n_{\min}} = \rho_n \sigma_{r_\omega}(\bar H)
    \sqrt{n_{\min}} \geq c_1\, \Omega_\omega\, \rho_n\, n_{\min}$, whence
    $\varepsilon_\omega \lesssim K/\sqrt{n\rho_n}$ for balanced communities, and
    $\delta \asymp \sqrt{K/n}$ by the separation condition and
    $\delta\leq2n_{\min}^{-1/2}$ (\Cref{lem:fanout-embedding}).
    Since $\|U_{\omega,2}^\top(O_2)_{\cdot j}\|_2\leq\|(O_2)_{\cdot j}\|_2$ gives
    $b_{\infty,2}^\omega\leq\tau_\omega b_{O_2}$, and $\tau_\omega b_{O_2}\leq b_\omega$ by
    \eqref{eq:loo2-weights}, $h_\star\leq1$ and $\vartheta_\omega\leq b_\omega$, the
    terms of \eqref{eq:fanout-eta} containing $b_\infty^\omega$, $b_{\infty,2}^\omega$,
    $b_\omega$ or $\vartheta_\omega$ are at most
    $3(b_\infty^\omega+b_\omega)\log n/\sigma_{r_\omega}(M_\omega)$ and
    $2\varepsilon_\omega^2(b_\infty^\omega+b_\omega)/\Omega_\omega$.
    Dividing \eqref{eq:fanout-eta} by $\delta$ therefore gives the balanced
    calculation of \eqref{eq:three-terms}:
    \[
        \frac K{\sqrt d}+\sqrt{\frac{K\log n}{d}}
        +\frac{K\log n}{d},\qquad
        \frac K{\sqrt d},\qquad \frac{K^2}{d},
        \quad d=n\rho_n.
    \]
    All terms are below the required fixed margin once
    $d\geq c(K^2+K\log n)$, which also ensures
    $\varepsilon_\omega\leq c_\star$ and, since $\delta\leq2$, the hypothesis
    $b_\omega\log n\leq c_\star\sigma_{r_\omega}(M_\omega)$ of \Cref{thm:fanout-main}.
    The stacked master theorem and \Cref{lem:single-linkage} give
    exact recovery with failure at most $4n^{-\alpha}$.
    For approximate $k$-means, additionally
    \[
        \frac{C_{\rm km}F_\omega^2}{n_{\min}\delta^2}
        \leq CK\varepsilon_\omega^2\leq\frac{CK^3}{d}.
    \]
    Thus $d\geq c(K^3+K\log n)$ ensures both the rowwise and
    Frobenius conditions for \Cref{lem:kmeans}(ii).
     For unbalanced communities, the displayed
    quantities are the same functions of $(\sigma_{r_\omega}(M_\omega), \delta, n_{\min})$
    as in \Cref{cor:sbm-regimes}, whose $\beta$-computation
    \eqref{eq:regime-signal}--\eqref{eq:regime-geometry} carries over unchanged.

    \emph{(ii) Near-orthogonal designs.} As in (i),
    $\sigma_{r_\omega}(M_\omega)\geq c_1\Omega_\omega\rho_nn_{\min}$ and
    $\varepsilon_\omega\lesssim K/\sqrt d$. The hypothesis $\varsigma_\omega\leq C_\varsigma$
    gives $\varsigma_\omega\varepsilon_\omega\lesssim K/\sqrt d$ and
    $\varsigma_\omega\Omega_\omega\log n/\sigma_{r_\omega}(M_\omega)
    \leq C_\varsigma\log n/(c_1\rho_nn_{\min})\lesssim K\log n/d$, and
    $b_\omega\leq C_b\Omega_\omega$ from \eqref{eq:loo2-weights} gives
    $b_\omega\log n/\sigma_{r_\omega}(M_\omega)\lesssim K\log n/d$. The hypotheses of
    \Cref{thm:fanout-main} with $\vartheta_\omega=0$ therefore hold once
    $d\geq c(K^2+K\log n)$. Divide \eqref{eq:fanout-eta} by $\delta\asymp\sqrt{K/n}$,
    using $h_\star=n_{\min}^{-1/2}\lesssim\sqrt{K/n}$,
    $b_\infty^\omega+b_{\infty,2}^\omega\leq C_2\Omega_\omega\sqrt{K/n}$ and
    $b_\omega\leq C_b\Omega_\omega$: the terms are bounded by the same quantities as in
    (i), and the conclusion follows as there, including the Frobenius condition for
    approximate $k$-means.

    For $O_1=O_2=I_n$ in the SBM, $g_\omega=2=\Omega_\omega$ and
    $b_\omega\geq\Omega_\omega b_{O_2}/\|O_2\|_{2\to2}=2$, so $\varsigma_\omega\leq1/2$.
    Here $M_\omega=\diag(Z,Z)\binom{B}{\omega B^{(2)}}Z^\top$, so
    $U_\omega=\diag(Z,Z)X$ for some $X\in\mathbb R^{2K\times r_\omega}$. Testing $U_\omega^\top$
    against the unit vector $n_a^{-1/2}\mathbf 1_{C_a}$ placed in either block gives
    $n_a^{1/2}\|X_{c\cdot}\|_2\leq1$ for the corresponding row $c$ of $X$, hence
    $\|U_\omega\|_{2\to\infty}\leq n_{\min}^{-1/2}$. Writing $U_{\omega,1}$ for the first
    $n$ rows of $U_\omega$, it follows that
    $b_\infty^\omega=\|U_{\omega,1}+\omega PU_{\omega,2}\|_{2\to\infty}
    \leq(1+\omega d\|\bar B\|_\infty)n_{\min}^{-1/2}\leq2n_{\min}^{-1/2}$ and
    $b_{\infty,2}^\omega\leq n_{\min}^{-1/2}$, so the leverage hypothesis holds with
    $C_2=3\sqrt\beta/2$. The conditioning \eqref{eq:fanout-well-cond} holds with
    $c_1=\sigma_{r_\omega}(\bar B)/2$, since $r_\omega=\rk(B)$,
    $\sigma_{r_\omega}(\bar H)\geq\sigma_{r_\omega}(Z\bar B)
    \geq\sqrt{n_{\min}}\,\sigma_{r_\omega}(\bar B)$ and $\Omega_\omega=2$. In the degree-corrected
    model the same argument with $\Theta Z$ and the unit vectors
    $\tilde n_a^{-1/2}\Theta\mathbf 1_{C_a}$ gives
    $\|U_\omega^\theta\|_{2\to\infty}\leq\theta_{\max}\tilde n_{\min}^{-1/2}
    \leq(\theta_{\max}/\theta_{\min})n_{\min}^{-1/2}$ and
    $b_\infty^\omega\leq(1+\theta_{\max}^2\|\bar B\|_\infty)\|U_\omega^\theta\|_{2\to\infty}$,
    so the leverage hypothesis holds with $C_2$ depending on $\bar\kappa$,
    $\|\bar B\|_\infty$ and $\beta$, and the conditioning holds with the same $c_1$,
    since $\sigma_{r_\omega}(\bar H^\theta)\geq\sigma_{r_\omega}(\Theta Z\bar B)
    \geq\sqrt{\tilde n_{\min}}\,\sigma_{r_\omega}(\bar B)$.
\end{proof}

\Cref{cor:fanout-regimes} extends to generic column designs as follows.

\begin{corollary}[Generic stacked designs]
    \label{cor:fanout-generic}
    Under the hypotheses of \Cref{cor:fanout-regimes},
    $b_\infty^\omega+b_\omega\leq6\Omega_\omega$ for all batteries (in the DCSBM once
    $\theta_{\max}^2\leq\sqrt n$), and either method recovers $z$ exactly, up to
    relabeling, with probability at least $1-4n^{-\alpha}$ when
    \begin{equation}
        \label{eq:fanout-generic-degree}
        n\rho_n\geq c\max\{K^{3/2}\sqrt n,\sqrt{Kn}\log n\}.
    \end{equation}
    If in addition $b_\infty^\omega+b_{\infty,2}^\omega\leq C_2\Omega_\omega\sqrt{K/n}$,
    single linkage suffices when
    $n\rho_n\geq c\max\{K^{5/3}n^{1/3},Kn^{1/3}(\log n)^{2/3}\}$, and approximate
    $k$-means when in addition $n\rho_n\geq cK^3$. The constant $c$ depends on the
    same quantities as in \Cref{cor:fanout-regimes}.
\end{corollary}

\begin{proof}[Proof of \Cref{cor:fanout-generic}]
    We work in the setting of the proof of \Cref{cor:fanout-regimes}.
    We first record $b_\infty^\omega + b_\omega \leq
    6 \Omega_\omega$ for any batteries and any $\omega$: $b_{O_1} \leq \| O_1
    \|_{2\to2} \leq \Omega_\omega$; $\omega \max_j \| O_2 P_{\cdot j} \|_2 \leq
    \omega \| O_2 \|_{2\to2}\, \rho_n \| \bar P \|_\infty \sqrt n =
    \big[ \omega (1 + \| \bar B \|_\infty) n \rho_n \| O_2 \|_{2\to2} \big] \cdot
    \frac{\| \bar P \|_\infty}{(1 + \| \bar B \|_\infty) \sqrt n} \leq
    \Omega_\omega$ (using $\| P_{\cdot j} \|_2 \leq \rho_n \| \bar P \|_\infty
    \sqrt n$ and $\|\bar P\|_\infty\leq(1+\|\bar B\|_\infty)\sqrt n$, which holds in the
    SBM and, as $\|\bar P\|_\infty\leq\theta_{\max}^2\|\bar B\|_\infty$, in the DCSBM once
    $\theta_{\max}^2\leq\sqrt n$); $\omega b_{O_2}
    \sqrt{n\rho_n \log n} \leq \omega \| O_2 \|_{2\to2} n \rho_n \leq
    \Omega_\omega$ (using $\log n \leq n\rho_n$); $\Omega_\omega b_{O_2}/\| O_2
    \|_{2\to2} \leq \Omega_\omega$; and $b_\infty^\omega \leq \max_j \|
    \bar{\mcO}_{\cdot j} \|_2 \leq b_{O_1} + \omega \max_j \| O_2 P_{\cdot j} \|_2
    \leq 2 \Omega_\omega$. The four pieces of $b_\omega$ are each at most
    $\Omega_\omega$, so $(b_\infty^\omega + b_\omega)/\Omega_\omega \leq 6$. With $\sigma_{r_\omega}(M_\omega)\geq c_1\Omega_\omega\rho_nn_{\min}$,
    $\varepsilon_\omega\lesssim K/\sqrt{n\rho_n}$ and
    $1/\delta\leq\sqrt{n_{\max}}/c_\delta\lesssim\sqrt{n/K}$, as in part (i) of that proof, the two terms
    of \eqref{eq:fanout-eta} containing $b_\infty^\omega$, $b_{\infty,2}^\omega$,
    $b_\omega$ or $\vartheta_\omega$ are, by the bounds recorded there, at most a
    constant times
    \[
        \frac{(b_\infty^\omega+b_\omega)\log n}{\sigma_{r_\omega}(M_\omega)\,\delta}
        \lesssim\frac{\Omega_\omega\log n}{\sigma_{r_\omega}(M_\omega)\,\delta}
        \lesssim\frac{\sqrt{Kn}\log n}{n\rho_n},
        \qquad
        \varepsilon_\omega^2\frac{b_\infty^\omega+b_\omega}{\Omega_\omega\delta}
        \lesssim\frac{\varepsilon_\omega^2}{\delta}
        \lesssim\frac{K^{3/2}\sqrt n}{n\rho_n}.
    \]
    The remaining terms are bounded as in part (i), and their requirements $K^2$
    and $K\log n$ are dominated because $K\leq n$; the hypothesis
    $b_\omega\log n\leq c_\star\sigma_{r_\omega}(M_\omega)$ holds because
    $b_\omega\log n/\sigma_{r_\omega}(M_\omega)\leq C_b\varepsilon_\omega\log n/\sqrt d
    \lesssim K\log n/d$. This gives the balanced
    threshold $c \max\{ K^{3/2}\sqrt n, \sqrt{Kn} \log n \}$.
    It also verifies the Frobenius condition: feasibility
    $n\rho_n\leq n$ and $n\rho_n\geq cK^{3/2}\sqrt n$ imply
    $n\geq c^2K^3$ and then $n\rho_n\geq c^2K^3$.
    Increasing $c$ makes the preceding ratio in the $k$-means condition smaller than one.

    If in addition $b_\infty^\omega+b_{\infty,2}^\omega\leq C_2\Omega_\omega\sqrt{K/n}$,
    keep $\vartheta_\omega=b_\omega\varepsilon_\omega\leq C_b\Omega_\omega\varepsilon_\omega$.
    The terms of \eqref{eq:fanout-eta} not containing $\vartheta_\omega$ are bounded
    as in part (ii) of that proof, and the two containing it satisfy
    \[
        \frac{\vartheta_\omega\log n}{\sigma_{r_\omega}(M_\omega)\,\delta}
        \lesssim\frac{\Omega_\omega\varepsilon_\omega\log n}{\sigma_{r_\omega}(M_\omega)\,\delta}
        \lesssim\frac{K^{3/2}\sqrt n\log n}{(n\rho_n)^{3/2}},
        \qquad
        \varepsilon_\omega^2\frac{\vartheta_\omega}{\Omega_\omega\delta}
        \lesssim\frac{\varepsilon_\omega^3}{\delta}
        \lesssim\frac{K^{5/2}\sqrt n}{(n\rho_n)^{3/2}}.
    \]
    These are below any fixed margin when $n\rho_n\gtrsim Kn^{1/3}(\log n)^{2/3}$ and
    $n\rho_n\gtrsim K^{5/3}n^{1/3}$ respectively, and the remaining requirements
    $K^2$ and $K\log n$ are dominated because $K\leq n$. Single linkage then
    recovers all community labels, by \Cref{lem:single-linkage}. For approximate
    $k$-means, the Frobenius condition requires $n\rho_n\geq cK^3$ as in part (i).
\end{proof}

\subsection{The degree-corrected stacked estimator}

\begin{proof}[Proof of \Cref{prop:fanout-dcsbm-main}]
    For the DCSBM, $P^0 = \Theta Z B Z^\top \Theta$, so $(P^0)^2
    = \Theta Z B (Z^\top \Theta^2 Z) B Z^\top \Theta = \Theta Z\, B^{(2)}_\theta\,
    (\Theta Z)^\top$. The proof of \Cref{lem:fanout-pop} uses only three facts
    about $P$: that $\bbe{A^{\langle 2 \rangle}} = \hol(P^2)$ (independence of
    distinct edges, model-agnostic), $\| P^0 \|_{2\to2} \leq n\rho_n \| \bar P
    \|_\infty$, and $\| \diag(P^0) \|_{2\to2} \leq \rho_n \| \bar P \|_\infty$;
    all hold here with $\| \bar P \|_\infty \leq \theta_{\max}^2 \| \bar B
    \|_\infty$. Hence $\bbe{O_2 A^{\langle 2 \rangle}} = O_2 \Theta Z
    B^{(2)}_\theta (\Theta Z)^\top + \Delta_2$ with $\| \Delta_2 \|_{2\to2} \leq 4
    \| O_2 \|_{2\to2} n \rho_n^2 \| \bar P \|_\infty^2$, and likewise the
    first-order block has mean $O_1 \Theta Z B (\Theta Z)^\top$ up to the diagonal
    correction of \Cref{lem:diag}. So $\bbe{R_\omega} = M^\theta_\omega +
    (\text{bias})$, with $M^\theta_\omega = H^\theta_\omega (\Theta Z)^\top$ of
    rank $r_\omega$, and the DCSBM computation in the proof of
    \Cref{lem:pop-embedding}(ii) applies to this factorization and gives
    population rows $\theta_i \tilde\mu_{z_i}$ whose norms are at least
    $\delta_{\mathrm{row}}$ and whose normalized values have separation
    $\delta_\theta$.

    \emph{(i) Master bound and exact recovery.} Every ingredient of the stacked
    chain (\Cref{loo2:event,loo2:column,loo2:structure,loo2:stability,%
    loo2:leverage,loo2:concentration,loo2:twoinfty-sharp}) uses only the hollow independent-edge structure of $E$
    and $\rk(M^\theta_\omega) = r_\omega$, so \Cref{loo2:twoinfty-sharp} holds for the
    DCSBM stack verbatim, with the incoherence $\| V^\theta_\omega \|_{2\to\infty}
    \leq \theta_{\max} \tilde n_{\min}^{-1/2}$; this is the stated
    $\eta_\infty^{(2),\theta}$.
The total-noise bound also gives the raw Frobenius bound $F_\omega$
    on this event for the same orthogonal matrix $Q$. By \Cref{lem:spherical},
    \[
        \|\widehat V^+Q-(V_\omega^\theta)^+\|_F
            \leq\frac{2F_\omega}{\delta_{\mathrm{row}}},\qquad
        \|\widehat V^+Q-(V_\omega^\theta)^+\|_{2\to\infty}
            \leq\frac{2\eta_\infty^{(2),\theta}}{\delta_{\mathrm{row}}}.
    \]
    The targets have separation $\delta_\theta$, and
    \eqref{eq:fanout-dcsbm-exact} gives
    $2\eta_\infty^{(2),\theta}/\delta_{\mathrm{row}}<\delta_\theta/4$, so single linkage recovers all
    community labels up to a permutation, by \Cref{lem:single-linkage}. For approximate $k$-means, additionally
    \eqref{eq:fanout-rounding} verifies \Cref{lem:kmeans}(ii).
    The rowwise condition also ensures that no empirical row is zero.

    \emph{(ii) Regimes.} Under bounded heterogeneity and balancedness, the
    separation conditions of \Cref{cor:fanout-regimes} give
    $\delta_{\mathrm{row}} \asymp \sqrt{K/n}$ and $\delta_\theta\geq c_\delta$, so
    \eqref{eq:fanout-dcsbm-exact} holds when
    $\eta_\infty^{(2),\theta}<c_\delta\delta_{\mathrm{row}}/8$, and
    $\sigma_{r_\omega}(M^\theta_\omega) \geq c_1 \Omega_\omega \rho_n \tilde n_{\min}
    \asymp c_1 \Omega_\omega \rho_n n / K$ as in the proof of
    \Cref{cor:dcsbm-regimes}.
    If the stacked design is spread, or under the
    near-orthogonal hypotheses of \Cref{cor:fanout-regimes}(ii),
    dividing the terms of $\eta_\infty^{(2),\theta}$ by
    $\delta_{\mathrm{row}}$ gives the same bounds as in the proof of
    \Cref{cor:fanout-regimes}(i)--(ii), up to $\bar\kappa$-dependent constants;
    all are small once
    $n\rho_n\geq c(K^2+K\log n)$. The hypotheses of \Cref{thm:fanout-main},
    including those for $\vartheta_\omega=0$ under the near-orthogonal hypotheses,
    follow from the computations in that proof with $\tilde n_{\min}$ in place of
    $n_{\min}$.
    Single linkage then recovers all community labels up to a permutation.
    For approximate $k$-means, the ratio in the Frobenius condition is at most
    $C_{\bar\kappa}K\varepsilon_\omega^2
    \leq C_{\bar\kappa}K^3/(n\rho_n)$, giving the stated
    $c(K^3+K\log n)$ condition. The threshold \eqref{eq:fanout-generic-degree}
    implies this additional requirement by the feasibility argument above.
    The generic-column clauses of \Cref{cor:fanout-generic} combine the bound
    $b_\infty^\omega + b_\omega \leq 6 \Omega_\omega$ of its proof, valid here once
    $\theta_{\max}^2\leq\sqrt n$, with the estimates displayed there, with
    $\delta_{\mathrm{row}}$ in place of $\delta$ and $\tilde n_{\min}$ in place of
    $n_{\min}$; under the second clause, approximate $k$-means requires
    $n\rho_n\geq cK^3$ as above.

    \emph{(iii) Weak recovery.} For approximate $k$-means, the noise bound \eqref{eq:fanout-total-noise}
    applies verbatim (its inputs use the model only through the maximal entry and
    maximal expected degree of $P$), so on its event, under the stated SNR
    condition, \Cref{lem:sin-theta} with $(M, \mcN)$ replaced by $(M^\theta_\omega,
    R_\omega - M^\theta_\omega)$, the spherical projection bound
    (\Cref{lem:spherical}) with population row norms at least
    $\delta_{\mathrm{row}}$, and \Cref{lem:kmeans}(i) with
    $\delta=\delta_\theta$ give, under the stated $k$-means condition,
    \[
        n_{\mathrm{err}}
        \leq\frac{C_{\rm km}}{\delta_\theta^2}
        \|\widehat V^+Q-(V_\omega^\theta)^+\|_F^2
        \leq\frac{4C_{\rm km}\varrho^2}{\delta_{\mathrm{row}}^2\delta_\theta^2}
        \leq C(\varepsilon_{\rm km})\,C_3^2\,
        \frac{r_\omega\,\Omega_\omega^2\, n\rho_n}
        {\sigma_{r_\omega}(M^\theta_\omega)^2\,\Delta_\omega^2},
    \]
    where $\varrho:=\|\widehat VQ-V_\omega^\theta\|_F$; this is
    \eqref{eq:fanout-general-rate} for the DCSBM. At $r_\omega=K$,
    $\Delta_\omega^2=\delta_{\mathrm{row}}^2/2\geq\theta_{\min}^2/(2\tilde n_{\max})$,
    and $\tilde n_{\max} \leq \theta_{\max}^2 n_{\max}$ yields
    \eqref{eq:fanout-dcsbm-rate}.
\end{proof}

\section{Proofs for Section~\ref{sec:identifiability} (identifiability)}
\label{app:ident-proofs}
This appendix contains all proofs for \Cref{sec:identifiability}.

\begin{proof}[Proof of \Cref{lem:column-constancy}]
    In the SBM, $P^0_{\cdot j} = Z B Z^\top e_j = Z B e_{z_j} = Z B_{\cdot z_j}$,
    since $Z^\top e_j = e_{z_j}$ is the $z_j$-th standard basis vector of
    $\mathbb{R}^K$. Applying $O$ gives $M_{\cdot j} = (OZ) B_{\cdot z_j}$. In the
    DCSBM, $P^0_{\cdot j} = \Theta Z B Z^\top \Theta e_j = \theta_j \Theta Z
    B_{\cdot z_j}$, and applying $O$ gives the stated form.
\end{proof}

\begin{proof}[Proof of \Cref{prop:pairwise}]
    By \Cref{lem:column-constancy}, $\mu_b - \mu_{b'} = (OZ)(B_{\cdot b} - B_{\cdot b'}) =
    O\bigl(Z(B_{\cdot b} - B_{\cdot b'})\bigr) = O w_{bb'}$. Thus $\mu_b = \mu_{b'}$
    if and only if $O w_{bb'} = 0$.
\end{proof}

\begin{proof}[Proof of \Cref{cor:partition}]
    By \Cref{lem:column-constancy} the columns of $M$ take exactly the values
    $\{\mu_b\}_{b=1}^K$, one per community, so the equality classes of the columns
    of $M$ are the communities if and only if the $\mu_b$ are pairwise distinct, that
    is, if and only if no pair collapses. By \Cref{prop:pairwise} this holds exactly
    when $w_{bb'} \notin \ker O$ for all $b \neq b'$. For degree queries,
    $\mu_b=\mathbf 1_n^\top ZB_{\cdot b}=\sum_an_aB_{ab}=D_b$.
\end{proof}

\begin{proof}[Proof of \Cref{thm:full-identifiability}]
    The three conditions in \Cref{thm:full-identifiability} are equivalent because $OZ$ has
    columns $O \mathbf{1}_{C_a}$ and $\mathcal{S} = \range(Z)$ with $Z$ of full column
    rank $K$
    since every community is nonempty;
    % (\Cref{ass:standing});
    hence $\rk(OZ) = \dim(O\mathcal{S})$, which equals
    $K$ exactly when $O$ is injective on $\mathcal{S}$.

    ($\Leftarrow$) Suppose $O$ is injective on $\mathcal{S}$. By
    \Cref{lem:column-constancy}, the columns of $M$ take the community-specific
    values $\mu_b=(OZ)B_{\cdot b}$. The columns of $B$ are pairwise distinct by the
    standing assumption of \Cref{subsec:models}, and $OZ$ is injective, so the
    $\mu_b$ are distinct. Their
    equality classes therefore recover the partition $z$ up to relabelling, and
    hence recover $Z$ and $\mathcal{S}$ under the same relabelling. Each column
    $P^0_{\cdot j}$ lies in $\mathcal{S}$ and is now the unique preimage of
    $M_{\cdot j}=OP^0_{\cdot j}$ under $O|_{\mathcal{S}}$. Thus $M$ determines
    $P^0$. Finally, $Z^\top P^0 Z=NBN$, so
    $B=N^{-1}(Z^\top P^0 Z)N^{-1}$ is recovered with the same labelling.

    ($\Rightarrow$) Suppose $\rk(OZ) < K$. Then $\ker O \cap \mathcal{S} \neq \{0\}$,
    so pick $0 \neq u = Z c \in \mathcal{S}$ with $O u = 0$, $c \in \mathbb{R}^K$.
    Set $D := u u^\top = Z (c c^\top) Z^\top$, a symmetric rank-one matrix with
    $\range(D) \subseteq \operatorname{span}(u) \subseteq \ker O$, so that $O D = 0$.
    For every sufficiently small nonzero $\varepsilon$, the matrix
    $B' := B + \varepsilon c c^\top$ still lies in $[0,1]^{K \times K}$ because $B$
    is interior. The two distinct block models $P^0=ZBZ^\top$ and
    $ZB'Z^\top=P^0+\varepsilon D$ then satisfy
    $OZB'Z^\top=OP^0=M$. Hence $M$ does not determine $B$, and taking
    $\varepsilon$ arbitrarily small proves the final claim of the theorem.
\end{proof}

\begin{proof}[Proof of \Cref{prop:dcsbm}]
    By \Cref{lem:column-constancy}, $M_{\cdot j} = \theta_j \nu_{z_j}$ with
    $\theta_j > 0$ unknown, so each column of $M$ reveals only the ray
    $\mathbb{R}_{\geq 0} \nu_{z_j}$. Two nodes in communities $b, b'$ yield the same
    ray if and only if $\nu_b$ and $\nu_{b'}$ are positively collinear, so the rays of
    the columns of $M$ determine the partition exactly when no pair collapses, giving (i);
    the degree case is \Cref{ex:dcsbm-degree}. For (ii) we recover the parameters in
    order. First, $\rk(O\Theta Z)=K$ makes $O\Theta Z$ injective.
    It therefore preserves the nonzero, pairwise distinct rays of the columns
    of $B$, so part (i) recovers $z$ from the rays of $M$. Given $z$, the column magnitudes $\| M_{\cdot j} \|_2 =
    \theta_j \| \nu_{z_j} \|_2$ determine each $\theta_j$ within its block up to the
    common scale $\| \nu_{z_j} \|_2$, which the normalization
    $\sum_{i \in C_a}\theta_i = n_a$ then fixes; thus $M$ determines $\Theta Z$. Finally,
    $\rk(O\Theta Z) = K$ makes the map $B_{\cdot b} \mapsto \nu_b = (O\Theta Z)
    B_{\cdot b}$ injective, so $B$ is recovered from the $K$ profiles
    $\nu_b=M_{\cdot j}/\theta_j$ for $j\in C_b$. Thus $M$ determines
    $P^0 = \Theta Z B Z^\top \Theta$, and $(\theta, B, z)$ up to relabelling;
    compare \cite{park_note_2024}.
\end{proof}

\clearpage
\section{Additional experiments}
\label{app:experiments-k}
\subsection{Varying the number of communities}

We varied community numbers and sizes in stochastic block models with
$B=\rho_n(\tfrac12\mathbf1\mathbf1^\top+I_K)$, keeping within-community
connection probabilities three times the between-community probabilities.
We used $n\in\{1200,2400,4800\}$, $K\in\{2,3,4,6,8\}$, and one community
approximately $\beta\in\{1,2,4\}$ times as large as each of the others,
rounding sizes to sum to $n$. We compared full adjacency with $K$ queries
on the same 20 graph draws per setting, using \Cref{sec:experiments}'s
$k$-means implementation at the population signal rank, which is $K$ here.

Using the true community assignments, we randomly assigned a fraction
$a=0.75$ of community $b$ to query group $G_b$ and divided the remaining
vertices as evenly as possible among the other $K-1$ groups.
For balanced communities, ignoring integer allocation errors,
\[
    OZ=\frac{n}{K}\left[aI_K+\frac{1-a}{K-1}
    (\mathbf1\mathbf1^\top-I_K)\right],\qquad
    \kappa_{\mathcal S}(O)=\left(a-\frac{1-a}{K-1}\right)^{-1}.
\]
We chose $\rho_n$ to give average expected degree
\[
    d=mK^2\log n,\qquad
    m\in\{0.25,0.375,0.5,0.75,1,1.5,2,3,4,6,8,12,16\}.
\]
We excluded settings requiring an edge probability greater than one.
For equal sizes, $d/\rho_n=n/2+n/K-3/2$.

We estimated the degree $d_*$ needed for $50\%$ exact-recovery frequency
by linear interpolation between the last sampled degree below this
frequency and the following degree. We report $d_*/\log n$, or a lower
bound when recovery remained below $50\%$ throughout the feasible grid.

In \Cref{fig:exp-k-scaling}, $d_*/\log n$ increases with $K$ for full
adjacency, but initially decreases for queries. This is consistent with
their improving conditioning: distributing the remaining vertices over
more groups increases the contrast $a-(1-a)/(K-1)$ and reduces the design
contribution to $\varepsilon_{\rm SNR}$ in \Cref{thm:sbm-main}.

Unequal community sizes increase $d_*/\log n$, particularly for queries
(\Cref{fig:exp-k-imbalance}). In the $k$-means condition
\eqref{eq:master-rounding}, smaller $n_{\min}$ likewise requires tighter
control of total embedding error relative to the population separation
$\delta$.

\begin{figure}[!htbp]
    \centering
    \includegraphics[width=\textwidth]{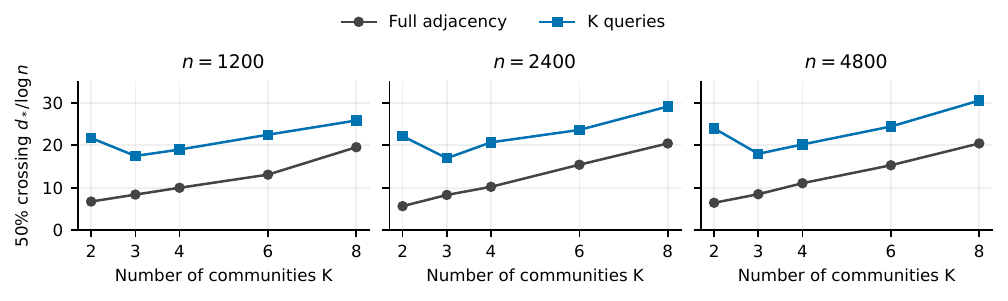}
    \caption{Estimated $d_*/\log n$ for equal-sized communities, using
    20 graph draws per sampled degree.}
    \label{fig:exp-k-scaling}
    \vspace{0.7em}
    \includegraphics[width=0.9\textwidth]{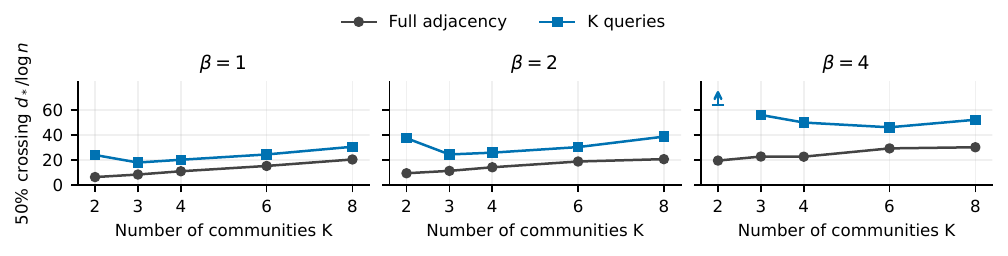}
    \caption{Estimated $d_*/\log n$ at $n=4800$ for size ratios $\beta$;
    the arrow marks the lower bound $d_*/\log n>64$.}
    \label{fig:exp-k-imbalance}
\end{figure}

\subsection{Effect of the clustering step}

To investigate the effect of the clustering step we
compare the use of k-means clustering with single linkage.
We consider the k-means implementation of
the main paper with Euclidean single linkage
cut into two clusters without refinement, with $K=2$
known in each instance. 
We generate data from $K=2$ equally sized communities with $n\in \{800,1600\}$
and expected average degree of the form
$$
d = c\log n, \ c\in \{4,8,12,16,20,24,32 \}.
$$ 
For the degree corrected SBM we have node degree parameters of $0.8$ and $1.2$, 
equally in each community. Throughout the signal is fixed at $s=0.6$.
We consider queries using degree only, the full adjacency matrix, 
and two queries where the alignment is $a=0.75$ for each query.
We show the mean misclassification rate for both clustering
approaches, across the queries considered 
previously, in Figure~\ref{fig:exp-k-linkage}
and the difference in error between these in
Figure~\ref{fig:exp-k-linkage-error}. In each case these are across 100 
draws for each setting. Single linkage requires higher expected 
degree to achieve comparable recovery to k-means.

%%% to be completed.

\begin{figure}[!htbp]
    \centering
    \includegraphics[width=0.8\textwidth]{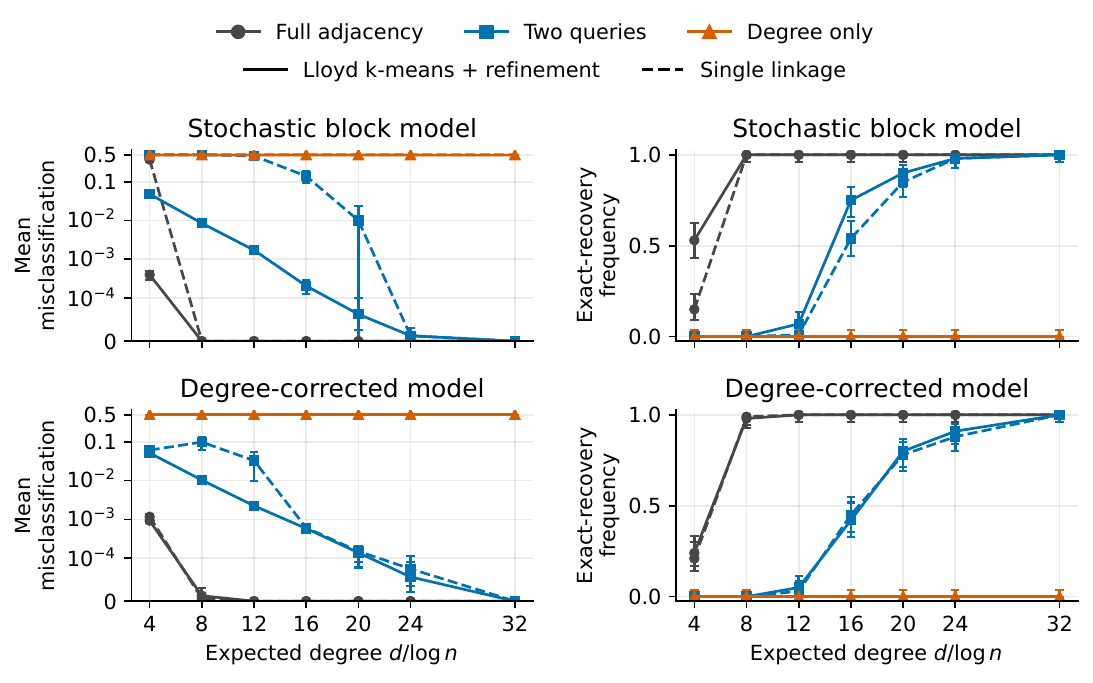}
    \caption{Misclassification rate and exact-recovery frequency for single linkage and Lloyd \(k\)-means under the SBM and DCSBM, with \(n=1{,}600\)}
    \label{fig:exp-k-linkage}
    \vspace{0.7em}
    \includegraphics[width=0.8\textwidth]{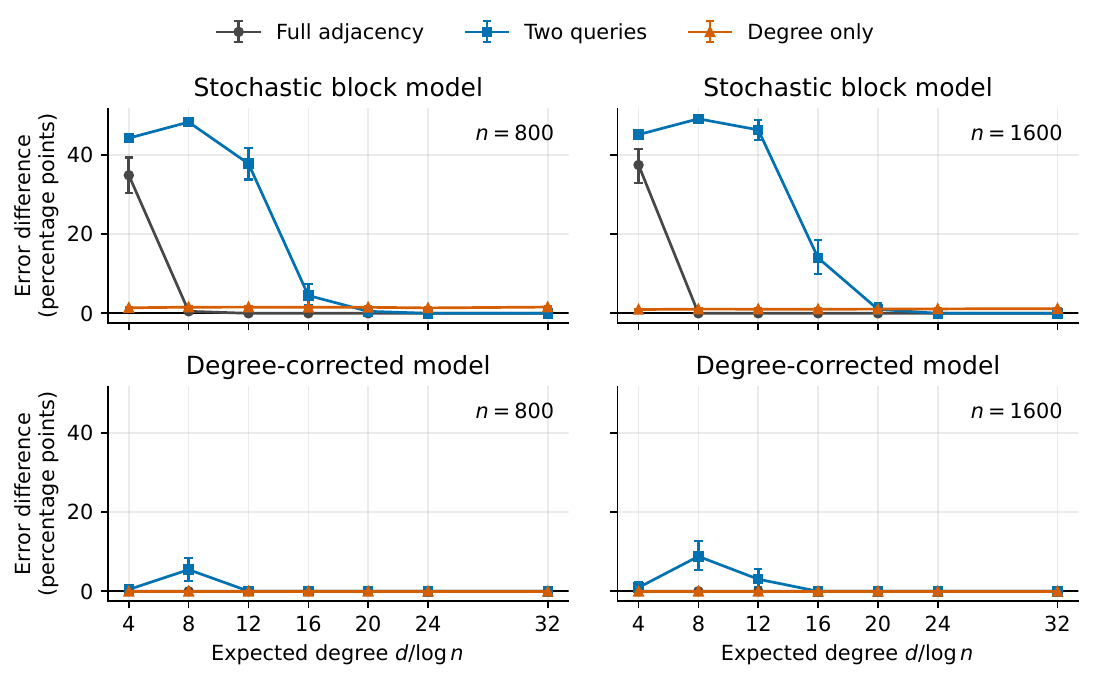}
    \caption{Paired misclassification differences (single linkage minus Lloyd \(k\)-means) for \(n=800\) and \(1{,}600\). Positive values favour \(k\)-means; error bars indicate 95\% paired bootstrap confidence intervals.}
    \label{fig:exp-k-linkage-error}
\end{figure}

\vfill

\end{document}